%% file: main.tex
\documentclass[sigconf, nonacm]{acmart}

\newcommand\vldbdoi{XX.XX/XXX.XX}

\newcommand\vldbavailabilityurl{URL_TO_YOUR_ARTIFACTS}

\newtheorem{myDef}{Definition}
\newtheorem{myLe}{Lemma}

\newtheorem{myEx}{Example}
\usepackage{subfigure}
\usepackage{bm}
\usepackage{verbatim}
\usepackage{color}
\usepackage[linesnumbered,boxed,ruled,commentsnumbered]{algorithm2e}
\usepackage{balance}
\begin{document}
\title{Evolutionary Clustering of Streaming Trajectories}

\author{Tianyi Li$^{\mathsection}$, Lu Chen$^{\dagger}$, Christian S. Jensen$^{\mathsection}$, Torben Bach Pedersen$^{\mathsection}$, Jilin Hu$^{\mathsection}$}
\affiliation{%
  \institution{
   {\large$^{\mathsection}$}Department of Computer Science, Aalborg University, Denmark\\
  {\large$^{\dagger}$}College of Computer Science, Zhejiang University, Hangzhou, China\\
{\large$^{\mathsection}$}\{tianyi, csj, tbp, hujilin\}@cs.aau.dk ~~~~~ \hspace*{0.5in} ~~~~~~{\large$^{\dagger}$}luchen@zju.edu.cn
  }
}

\begin{abstract}
The widespread deployment of smartphones and location-enabled, networked in-vehicle devices renders it increasingly feasible to collect streaming trajectory data of moving objects. The continuous clustering of such data can enable a variety of real-time services, such as identifying representative paths or common moving trends among objects in real-time. However, little attention has so far been given to the quality of clusters---for example, it is beneficial to smooth short-term fluctuations in clusters to achieve robustness to exceptional data.

We propose the notion of evolutionary clustering of streaming trajectories, abbreviated ECO, that enhances streaming-trajectory clustering quality by means of temporal smoothing that prevents abrupt changes in clusters across successive timestamps. Employing the notions of snapshot and historical trajectory costs, we formalize ECO and then formulate ECO as an optimization problem and prove that ECO can be performed approximately in linear time, thus eliminating the iterative processes employed in previous studies. Further, we propose a minimal-group structure and a seed point shifting strategy to facilitate temporal smoothing. Finally, we present all algorithms underlying ECO along with a set of optimization techniques. Extensive experiments with two real-life datasets offer insight into ECO and show that it outperforms state-of-the-art solutions in terms of both clustering quality and efficiency.

\end{abstract}
\maketitle

%\pagestyle{\vldbpagestyle}
%\begingroup\small\noindent\raggedright\textbf{PVLDB Reference Format:}\\
%%\vldbauthors. \vldbtitle. PVLDB, \vldbvolume(\vldbissue): \vldbpages, \vldbyear.\\
%\href{https://doi.org/\vldbdoi}{doi:\vldbdoi}
%\endgroup
%\begingroup
%\renewcommand\thefootnote{}\footnote{\noindent
%This work is licensed under the Creative Commons BY-NC-ND 4.0 International License. Visit \url{https://creativecommons.org/licenses/by-nc-nd/4.0/} to view a copy of this license. For any use beyond those covered by this license, obtain permission by emailing \href{mailto:info@vldb.org}{info@vldb.org}. Copyright is held by the owner/author(s). Publication rights licensed to the VLDB Endowment. \\
%\raggedright Proceedings of the VLDB Endowment, Vol. \vldbvolume, No. \vldbissue\ %
%ISSN 2150-8097. \\
%\href{https://doi.org/\vldbdoi}{doi:\vldbdoi} \\
%}\addtocounter{footnote}{-1}\endgroup

%\ifdefempty{\vldbavailabilityurl}{}{
%\vspace{.3cm}
%\begingroup\small\noindent\raggedright\textbf{PVLDB Artifact Availability:}\\
%The source code, data, and/or other artifacts have been made available at \url{\vldbavailabilityurl}.

\input{introduction}

\input{preliminary}

\input{problem}

\input{solution}

\input{algorithm}

\input{experiment}

\input{related_work}

\input{conclusion}

\bibliographystyle{ACM-Reference-Format}
\balance
\bibliography{sample}

\end{document}

%% file: introduction.tex
\section{Introduction}
\begin{figure}[t]
\begin{center}
\subfigcapskip=-20pt
\includegraphics[width=0.43\textwidth]{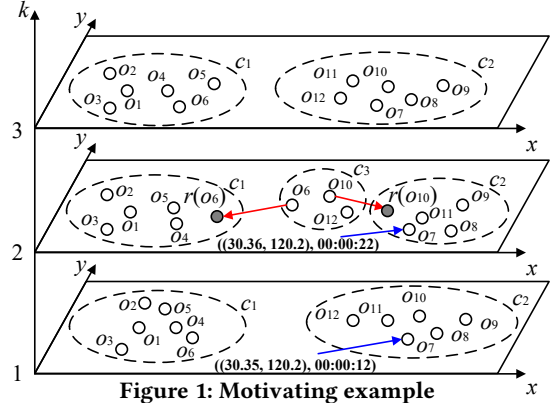}\vspace{-8mm}
\caption{{Motivating example}}
\vspace{-5mm}
\label{fig:motivation}
\end{center}
\end{figure}
It is increasingly possible to equip moving objects with positioning devices that are capable of transmitting object positions to a central location in real time. Examples include people with smartphones and vehicles with built-in navigation devices or tracking devices. This scenario opens new opportunities for the real-time discovery of  hidden mobility patterns. These
patterns allow characterizing individual mobility for a certain time interval and enable a broad range of
important services and applications such as route planning~\cite{zeng2019last,wang2020demand}, intelligent transportation management~\cite{wang2021public}, and road infrastructure optimization~\cite{wu2015glue}.

{As a typical moving pattern discovery approach, clustering aims to group a set of trajectories into comparatively homogeneous clusters to extract  representative paths or movement patterns shared by moving objects. Considering a streaming setting, many works are proposed to cluster the trajectories in real-time~\cite{jensen2007continuous,li2010incremental,yu2013online,costa2014dealing,deng2015scalable,da2016online,chen2019real, tang2012discovery,li2012effective}. However, existing real-time clustering methods focus on the most current data, achieving low computational cost at the expense of clustering quality~\cite{xu2014adaptive}. In streaming settings, clusterings should be robust to short-term fluctuations in the underlying trajectory data, which may be achieved by means of smoothing~\cite{chi2007evolutionary}. 
An example illustrates this.
\begin{myEx}
Figure~\ref{fig:motivation} shows the trajectories of 12 moving objects at three timestamps, $k=1,2,3$.  
Traditional clustering algorithms return the two clusters $c_1=\{o_1, o_2, o_3, o_4, o_5, o_6\}$ and $c_2=\{o_7, o_8, o_9, o_{10}, o_{11}, $\\$o_{12}\}$ at the first timestamp, the three clusters $c_1=\{o_1, o_2, o_3, o_4, o_5\}$, $c_2=\{o_7, o_8, o_9, o_{11}\}$, and $c_3=\{o_6, o_{10}, o_{12}\}$ at the second timestamp, and the same two clusters at the third timestamp as at the first timestamp. 
\label{ex:motivation}
\end{myEx}

The underlying reason  for this result is the unusual behavior of objects $o_6$ and $o_{10}$ at the second timestamp. Clearly, returning the same two stable
clusters for all three timestamps is a more robust and better-quality result. A naive approach to eliminating the effect of the two objects' unusual behavior is to perform cleaning before clustering. However,
studies of on two real-life datasets show that among the trajectories that cause the mutations of clusterings, 
88.9\% and 75.9\% of the trajectories follow the speed constraint, while 97.8\% and 96.1\% of them are categorized as inliers~\cite{ester1996density}.
Moreover, in real-time applications, it is impractical to correct previous clusterings retroactively. Hence, it is difficult for existing cleaning techniques to facilitate smoothly shifted clustering  sequences~\cite{li2020data, patil2018geosclean, idrissov2012trajectory}.

However, this problem can be addressed by applying evolutionary clustering~\cite{kim2009particle, fenn2009dynamic, chen2020collaborative,chakrabarti2006evolutionary,chi2007evolutionary,gupta2011evolutionary,xu2014adaptive,yin2021multi,ma2017evolutionary,liu2020detecting}, where a good current clustering result is one that fits the current data well, while not deviating too much from the recent history of clusterings. 
Specifically, \textit{temporal smoothness} is integrated into the measure of clustering quality~\cite{chi2007evolutionary}. This way, evolutionary clustering is able to outperform traditional clustering as it can reflect long term trends while being robust to short-term variability. Put differently, applying evolutionary clustering to trajectories can mitigate  adverse effects of intermittent noise on clustering and present users with smooth and consistent movement patterns. In Example~\ref{ex:motivation}, clustering with temporal consistency is obtained if $o_6$ is smoothed to $r(o_6)$ and $o_{10}$ is smoothed to $r(o_{10})$ at the second timestamp. Motivated by this, we study evolutionary clustering of trajectories.}

Existing evolutionary clustering studies target dynamic networks and are not suitable for trajectory applications, mainly for three reasons. First, the solutions are  designed specifically for dynamic networks, which differ substantially from two-dimensional trajectory data. Second, the movement in trajectories is generally much faster than the evolution of dynamic networks, which renders the temporal smoothness used in existing studies too "strict" for trajectories. Third, existing studies often optimize the clustering quality iteratively at each timestamp~\cite{kim2009particle,chakrabarti2006evolutionary,yin2021multi,folino2013evolutionary,liu2020detecting,liu2019evolutionary}, which is computationally costly and is infeasible for large-scale trajectories.

We propose an efficient and effective method for \underline{e}volutionary \underline{c}lustering of streaming traject\underline{o}ries (\textbf{ECO}).
First, we adopt the idea of neighbor-based smoothing~\cite{kim2009particle} and develop a structure called \textit{minimal group} that is summarized by a \textit{seed point} in order to facilitate smoothing.
Second, following existing studies~\cite{chakrabarti2006evolutionary,yin2021multi,xu2014adaptive,folino2013evolutionary,liu2020detecting,liu2019evolutionary}, we formulate ECO as an optimization problem that employs the new notions of snapshot cost and historical cost.
The snapshot cost evaluates the true concept shift of clustering defined according to the distances between smoothed and original locations. The historical cost evaluates the temporal distance between locations at adjacent timestamps by the \textit{degree of closeness}.
Next, we prove that the proposed optimization function can be decomposed and that each component can be solved approximately in constant time. The effectiveness of smoothing is further improved by a \textit{seed point shifting} strategy. Finally, we introduce a grid index structure and present algorithms for each component of evolutionary clustering along with a 
set of optimization techniques, to improve clustering performance. The paper's main contributions are summarized as follows,
%s\vspace{-3mm}
\begin{itemize}\setlength{\itemsep}{-\itemsep}
    \item We formalize ECO  problem. To the best of our knowledge, this is the first proposal for streaming trajectory clustering that takes into account temporal smoothness.
    \item We formulate ECO as an optimization problem, based on the new notions of snapshot cost and historical cost. We prove that the optimization problem can be solved approximately in linear time.
    \item We propose a \textit{minimal group} structure to facilitate temporal smoothing and a \textit{seed point shifting} strategy to  improve clustering  quality of evolutionary clustering. Moreover, we present all algorithms needed to enable evolutionary clustering, along with a 
    set of optimization techniques.
    \item Extensive experiments on two real-life datasets show that ECO advances the  state-of-the-arts in terms of both clustering quality and efficiency.
\end{itemize}
The rest of paper is organized as follows. We present preliminaries in Section~\ref{sec:preliminary}. We formulate the problem in Section~\ref{sec:ec} and derive its solution in Section~\ref{sec: Quadratic_time_solution}. Section~\ref{sec:algorithm} presents the algorithms and optimization techniques. Section~\ref{sec:experiments} covers the experimental study. Section~\ref{sec:related_work} reviews related work, and Section~\ref{sec:conclusion} concludes and offers directions for future work.

%% file: preliminary.tex
\section{Preliminaries}\label{sec:preliminary}

\begin{table}
\centering\small
\caption{Frequently used notation}\vspace{-4mm} 
\begin{tabular}{|l|l|}
 \bottomrule 
\textbf{Notation}&\textbf{Description}\\ \hline
	$o$ & A trajectory\\\hline
	$\textit{dt}_k$ & The $k^{th}$ time step \\\hline
$o.l_k$, $o.t_k$ & The location and timestamp of $o$ at 	$\textit{dt}_k$\\\hline
$o.l$, $o.t$ & A simplification of $o.l_k$, $o.t_k$ at  $\textit{dt}_k$\\\hline
$o.\tilde{l}$, $o.\tilde{t}$ & A simplification of $o.l_{k-1}$, $o.t_{k-1}$ at $\textit{dt}_k$\\\hline
$\mathcal{O}_k$ & A set of trajectories at 	$\textit{dt}_k$\\\hline
$r(o)$ & An adjustment of $o.l$\\\hline
$\mathcal{R}_k$ & The set of adjustments of $\mathcal{O}_k$\\\hline
$s$ & A seed point of $o$ at the current time step $\textit{dt}_k$  
\\\hline
$\tilde{s}$ & A seed point of $o$ at the previous time step $\textit{dt}_{k-1}$\\\hline
$\mathcal{S}_k$ & The set of seed points at $\textit{dt}_k$  \\\hline
$\mathcal{M}_{k}(s)$ & A minimal group summarized by a seed point $s$ at $\textit{dt}_k$  \\\hline
$\mathcal{SC}_k(r(o))$ & The snapshot cost of a trajectory $o$ w.r.t. $r(o)$ at $\textit{dt}_k$  \\\hline
$\mathcal{TC}_k(r(o))$ & The historical cost of a trajectory $o$ w.r.t. $r(o)$ at $\textit{dt}_k$  \\\hline
$c$ & A cluster $c$\\\hline
$\mathcal{C}_k$ & The set of clusters obtained  at $\textit{dt}_k$ \\
\bottomrule
\end{tabular}\label{tb:paramater}\vspace{-4mm}
\end{table}

\subsection{Data Model}\label{sec:data_model}
\begin{myDef}
A \textbf{GPS record} is a pair $(l,t)$, where $t$ is a \textbf{timestamp} and $l=(x,y)$ is the \textbf{location}, with $x$ being a longitude and $y$ being a latitude.
\end{myDef}

\begin{myDef}
A \textbf{streaming trajectory} $o$ is an unbounded ordered sequence of GPS records, $\langle (o.l_1, o.t_1), (o.l_2, o.t_2) \cdots \rangle$. 
\end{myDef}
The GPS records of a trajectory may be transmitted to a central location in an unsynchronized manner. To avoid this affecting the subsequent processing, we adopt an existing approach~\cite{chen2019real} and discretize time into short intervals that are indexed by integers. We then map the timestamp of each GPS record to the index of the interval that the timestamp belongs to.
In particular, we assume that the start time is 00:00:00 UTC, and we partition time into intervals of duration $\Delta t=10s$. Then time series $\langle$00:00:01, 00:00:12, 00:00:20, 00:00:31, 00:00:44$\rangle$ and $\langle$00:00:00, 00:00:13, 00:00:21, 00:00:31, 00:00:40$\rangle$
are both mapped $\langle 0, 1, 2, 3, 4\rangle$. We call such a sequence a discretized time sequence and call each discretized timestamp a \textbf{\textit{time step}} $\textit{dt}$.
We use trajectory and streaming trajectory interchangeably.

\begin{myDef}
A trajectory is \textbf{active} at time step $dt = [t_1,t_2]$ if it contains a GPS record $(l,t)$ such that $t \in [t_1,t_2]$.
\end{myDef}

\begin{myDef}
A \textbf{snapshot} $\mathcal{O}_k$
is the set of trajectories that are active at time step  $\textit{dt}_k$.
\end{myDef}

Figure~\ref{fig:motivation} shows three snapshots $\mathcal{O}_1$, $\mathcal{O}_2$, and $\mathcal{O}_3$, each of which contains twelve trajectories. Given the start time 00:00:00 and $\Delta t=10$, $(o_7.l,o_7.t)$ arrives at $\textit{dt}_1$ because 00:00:12 is mapped to 1. For simplicity, we use $o$ in figures to denote $o.l$.
The interval duration $\Delta t$ is the default sample interval of the dataset. Since deviations between the default sample interval and the actual intervals are small~\cite{li2020compression}, we can assume that each trajectory $o$ has at most one
GPS record at each time step $\textit{dt}_k$. If this is not the case for a trajectory $o$,  we simply keep $o$'s earliest GPS  at the time step. This simplifies the subsequent clustering. Thus, the GPS record of $o$ at $\textit{dt}_k$ is denoted as $(o.l_{k}, o.t_{k})$. If a trajectory $o$ is active at both $\textit{dt}_{k-1}$ and $\textit{dt}_k$ and the current time step is $\textit{dt}_k$, $o.l_k$ and $o.t_k$ are simplified as $o.l$ and $o.t$, and $o.l_{k-1}$ and $o.t_{k-1}$ are simplified as $o.\tilde{l}$ and $o.\tilde{t}$. At time step $\textit{dt}_2$  ($k=2$) in Figure~\ref{fig:motivation}, $o_7.\tilde{l}=o_7.{l}_1=(30.35, 120.2)$, $o_7.\tilde{t}=o_7.{t}_1=\,\,$00:00:12, $o_7.l=o_7.{l}_2=(30.36, 120.2)$, and $o_7.t=o_7.{t}_2=\,\,$00:00:22.

\begin{myDef}\label{def:local_density}
A $\bm{\theta}$\textbf{-neighbor set} of a streaming trajectory $o\,(\in \mathcal{O}_k)$ at the time step $\textit{dt}_k$ is $\mathcal{N}_{\theta}(o)=\{o'| o'\in \mathcal{O}_k\wedge d(o.l, o'.l)\leq \theta \}$,where $d(\cdot)$ is Euclidean distance and $\theta$ is a distance threshold. $\lvert \mathcal{N}_{\theta}(o) \rvert$ is called the \textbf{local density} of $o$ w.r.t. ${\theta}$ at $\textit{dt}_k$.
\end{myDef}
Figure~\ref{fig:t1} plots $o_i\,(1
\leq 1 \leq 6)$ at $\textit{dt}_1$ from Figure~\ref{fig:motivation}, where 
$\mathcal{N}_{\delta}(o_1)=\{o_1, o_2, o_3\}$.

\begin{figure} \centering
\subfigcapskip=-5pt
 \subfigure[Core points 
 $o_i$\,($1\leq i \leq 6\wedge i\neq 3$)
 ($\textit{minPts}=3)$]{      \includegraphics[scale=.4]{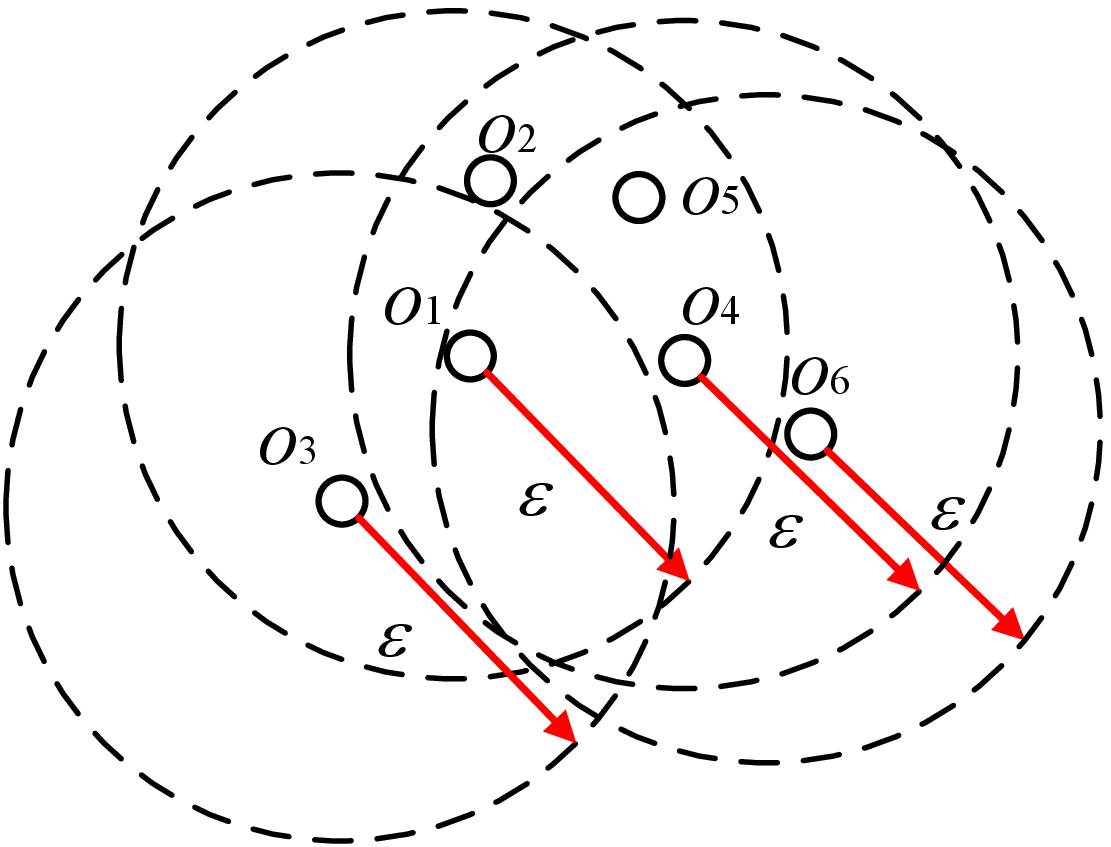}} 
 \subfigure[Seed points $o_1$ and $o_4$\,$(\rho=3)$ ]{
\includegraphics[scale=.5]{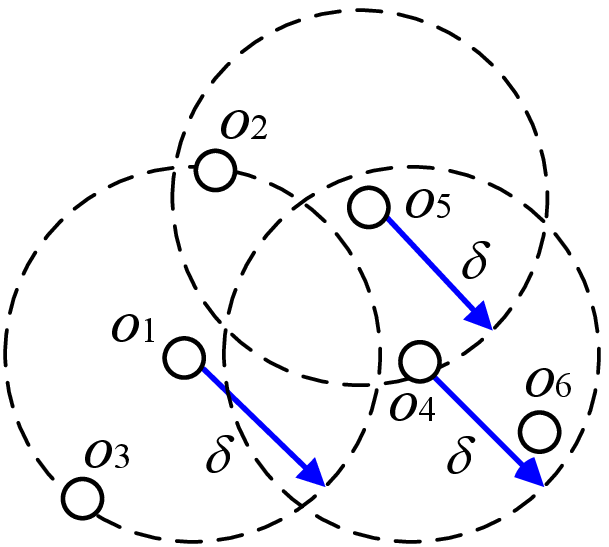}}\vspace{-5mm}
\caption{${o_i\,(1\leq i \leq 6)}$ at ${\textit{dt}_1}$ in Figure~\ref{fig:motivation}}
\label{fig:t1}
\vspace{-3mm}
\end{figure}

\subsection{DBSCAN}\label{sec:DBSCAN}
We adopt a well-known density-based clustering approach, DBSCAN~\cite{ester1996density}, for clustering. 
DBSCAN relies on two parameters to characterize density or sparsity, i.e., positive values $\varepsilon$ and \textit{minPts}.

\begin{myDef}\label{def:core_point}
A trajectory $o\in \mathcal{O}_k$ is a \textbf{core point} w.r.t. $\varepsilon$ and \textit{minPts}, if $\mathcal{N}_{\varepsilon}(o) \geq \textit{minPts}$.
\end{myDef}

\begin{myDef}\label{def:density_reachable}
A trajectory $o\in \mathcal{O}_k$ is \textbf{density reachable} from another trajectory $o'\in \mathcal{O}_k$, if a sequence of trajectories $o_{1}, o_{2}, \cdots o_{n}~(n\geq 2)$ exists such that (i) $o_1=o'$ and $o_n=o$; (ii) $o_{w}~(1 \leq w <n)$ are core points; and (iii) $d(o_{w},o_{w+1} )\leq \varepsilon~(1 \leq w < n)$. 
\end{myDef}

\begin{myDef}\label{def:connect}
A trajectory $o\in \mathcal{O}_k$ is \textbf{connected} to another trajectory $o'$ if a trajectory $o''$ exists such that both $o$ and $o'$ are density reachable from $o''$.
\end{myDef}

\begin{myDef}\label{def:cluster}
A non-empty subset of trajectories of $\mathcal{O}_k$ is called a \textbf{cluster} $c$, if $c$ satisfies the following conditions:
\begin{itemize}
\item Connectivity: $\forall o, o' \in c$,  $o$ is connected to $o'$;
\item Maximality: $\forall o, o' \in \mathcal{O}_k$, if $o\in c$ and $o'$ is density reachable from $o$, then $o'\in c$. 
\end{itemize}
\end{myDef}
Definition \ref{def:cluster} indicates that a cluster is formed by a set of core points and their density reachable points. {Given $\varepsilon$ and \textit{minPts}, $o\in \mathcal{O}_k$ is an \textbf{\textit{outlier}}, if it is not in any cluster; $o\in \mathcal{O}_k$ is a \textbf{\textit{border point}}, if $\mathcal{N}_{\varepsilon}(o) < \textit{minPts}$ and $d(o, o')\leq \epsilon$, where $o'$ is a core point.}
 
\begin{myDef}
A \textbf{clustering result} $\mathcal{C}_k=\{c_1, c_2, \cdots, c_n\}$ is a set of clusters obtained from the snapshot $\mathcal{O}_k$.
\end{myDef}

\begin{myEx}\label{ex:cluster}
In Figure~\ref{fig:motivation}, $\mathcal{C}_1$ has two clusters $c_1=\{o_1, o_2, o_3, o_4, o_5, $\\$o_6\}$ and $c_2=\{o_7, o_8, o_9, o_{10}, o_{11}, o_{12}\}$. {Further, $o_i$\,($1\leq i \leq 6\wedge i\neq 3$) in Figure~\ref{fig:t1}a are core points.}
\end{myEx}

\subsection{Evolutionary Clustering}\label{sec:evolutionary_clustering}
Evolutionary clustering is the problem producing a sequence of clusterings from streaming data; that is, clustering for each snapshot. It takes into account the smoothness characteristics of streaming data to obtain high-quality clusterings~\cite{chakrabarti2006evolutionary}. Specifically, two quality aspects are considered:
\begin{itemize}
    \item High historical quality: clustering  $\mathcal{C}_k$ should be similar to the previous clustering $\mathcal{C}_{k-1}$;
    \item High snapshot quality: $\mathcal{C}_k$ should reflect the true concept shift of clustering, i.e., remain faithful to the data at each time step.
\end{itemize}
Evolutionary clustering uses a cost function $\mathcal{F}_k$ that enables trade-offs between historical quality and snapshot quality at each time step $\textit{dt}_k$~\cite{chakrabarti2006evolutionary},
\begin{equation}\label{f:cost}
\mathcal{F}_k=\mathcal{SC}_k(\mathcal{C}_o, \mathcal{C}_k)+\alpha \cdot \mathcal{TC}_k(\mathcal{C}_{k-1}, \mathcal{C}_k)
\end{equation}
$\mathcal{F}_k$ is the sum of two terms: a snapshot cost ($\mathcal{SC}_k$) and a historical cost ($\mathcal{TC}_k$). The snapshot cost $\mathcal{SC}_k$ captures the similarity between clustering $\mathcal{C}_k$ and clustering $\mathcal{C}_o$ that is obtained without smoothing. The smaller $\mathcal{SC}_k$ is, the better the snapshot quality is. The historical cost $\mathcal{TC}_k$ measures how similar clustering $\mathcal{C}_k$ and the previous clustering $\mathcal{C}_{k-1}$ are. The smaller $\mathcal{TC}_k$ is, the better the historical quality is. Parameter  $\alpha\,(>0)$ enables controlling the trades-off between snapshot quality and historical quality.

%% file: problem.tex
\section{Problem Statement}\label{sec:ec}
We start by presenting two observations, based on which, we define the problem of evolutionary clustering of streaming trajectories.

\subsection{Observations}\label{sec:observation}
\paragraph{\textbf{Gradual evolutions of travel companions}} As pointed out in a previous study~\cite{tang2012discovery}, movement trajectories represent continuous and gradual location changes, rather than abrupt changes, implying that co-movements among trajectories also change only gradually over time.
Co-movement may be caused by (i) physical constraints of both road networks and vehicles, and  vehicles may have close relationships, e.g., they may belong to the same fleet or may target the same general destination~\cite{tang2012discovery}.

\paragraph{\textbf{Uncertainty of "border" points}} 
{Even with the observation that movements captured by trajectories are not dramatic during a short time, border points are relatively more likely to leave their current cluster at the next time step than core points. This is validated by statistics from two real-life datasets. Specifically, among the trajectories shifting to another cluster or becoming an outlier during the next time steps, 75.0\% and 61.5\% are border points in the two real-life datasets.}

\subsection{Problem Definition}\label{sec:problem_definition}

\paragraph{\textbf{Cost embedding}}
Existing evolutionary clustering studies generally perform temporal smoothing on the clustering result~\cite{folino2013evolutionary,chakrabarti2006evolutionary,chi2007evolutionary,yin2021multi}. Specifically, they adjust $\mathcal{C}_k$ iteratively so as to minimize Formula~\ref{f:cost}, which incurs very high cost. We adopt cost embedding~\cite{kim2009particle}, which pushes down the cost formula from the clustering result level to the data level, thus enabling flexible and efficient temporal smoothing. 
However, the existing cost embedding technique~\cite{kim2009particle} targets dynamic networks only.
To apply cost embedding to trajectories, we propose a minimal group structure and snapshot and historical cost functions.

\paragraph{\textbf{Snapshot cost} $\bm{\mathcal{SC}_k}$}
We first define the notion of an "adjustment" of a trajectory.
\begin{myDef}\label{def:repair}
An \textbf{adjustment} $r_k(o)$ is a
 location of a trajectory $o$ obtained through smoothing at $\textit{dt}_k$. Here, $r_k(o) \neq r_k(o')$ if $o\neq o'$.
The set of adjustments in $\mathcal{O}_k$ is denoted as $\mathcal{R}_k$.
\end{myDef}
We simplify $r_k(o)$ to $r(o)$ if the context is clear.
In Figure~\ref{fig:motivation}, $r(o_6)$ is an adjustment of $o_6$ at $\textit{dt}_2$. According to Formula~\ref{f:cost}, the snapshot cost measures how similar the current clustering result $\mathcal{C}_k$ is to the original clustering result $\mathcal{C}_o$. Since we adopt cost embedding that smooths trajectories at the data level, the snapshot cost of a trajectory $o$ w.r.t. its adjustment $r(o)$ at $\textit{dt}_k$ (denoted as $\mathcal{SC}_k(r(o))$) is formulated as the deviation between $o$ and $r(o)$ at $\textit{dt}_k$:  
\begin{equation}\label{f:SC0}
    \mathcal{SC}_k(r(o))=d(r(o), o.l)^2 \quad  s.t. \quad d(r(o), o.\tilde{l})\leq \mu \cdot (o.t-o.\tilde{t}),
\end{equation}
where $\mu$ is a speed constraint of the road network. Formula~\ref{f:SC0} requires that any adjustment $r(o)$ must follow the speed constraint. 
Obviously, the larger the distance between $o.l$ and its adjustment $r(o)$, the higher the snapshot cost.

\paragraph{\textbf{Historical cost} $\bm{\mathcal{TC}_k}$} 
As discussed in Section~\ref{sec:evolutionary_clustering}, one of the goals of evolutionary clustering is smoothing the change of clustering results during adjacent time steps.
Since we push down the smoothing from the cluster level to trajectory level, the problem becomes one of ensuring that each trajectory represents a smooth movement. According to the first observation in Section~\ref{sec:observation}, gradual location changes lead to stable co-movement relationships among trajectories during short periods of time. Thus, similar to neighbor-based smoothing in dynamic communities~\cite{kim2009particle}, it is reasonable to smooth the location of each trajectory in the current time step using its neighbours at the previous time step. However, the previous study~\cite{kim2009particle} smooths the distance between each pair of neighboring nodes. Simply applying this to trajectories may degrade the performance of smoothing if a "border" point is involved. Recall the second observation of Section~\ref{sec:observation} and assume that $o_1.l$ is smoothed according to $o_3.l$ at $\textit{dt}_2$ in Figures~\ref{fig:motivation} and~\ref{fig:t1} . As $o_3$ is a border point at $\textit{dt}_1$ with a higher probability to leave the cluster $c_1$ at $\textit{dt}_2$, using $o_3$ to smooth $o_1$ may result in $o_1$ also leaving $c_1$ or being located at the border of $c_1$ at $\textit{dt}_2$. The first case may incur an abrupt change to the clustering while the second case may degrade the intra-density of $c_1\,(\in \mathcal{C}_2)$ and increase the inter-density of clusters in $\mathcal{C}_2$.
To tackle this problem, 
we model neighboring trajectories as minimal groups summarized by seed points.

\begin{myDef}\label{def:seed_point}
A \textbf{seed point} $s\,(s\in \mathcal{S}_k)$ summarizes a \textbf{minimal group} $\mathcal{M}_k(s) = \{o \in  \mathcal{O}_k |\, d(o, s)\leq \delta\wedge \forall s'\in \mathcal{S}_k \, (s'\neq s\Rightarrow d(o,s)\leq d(o, s'))\}$ at $\textit{dt}_k$, where $\delta \,(0 < \delta \leq \varepsilon)$ is a given parameter, and $\mathcal{S}_k\,(\mathcal{S}_k \subset \mathcal{O}_k)$ is a seed point set at $\textit{dt}_k$. The cardinality of $\mathcal{M}_k(s)$, $\lvert \mathcal{M}_k(s) \rvert$, exceeds a parameter $\rho$.
Any trajectory o in $\mathcal{M}_k(s)$ that is different from $s$ is called a \textbf{non-seed point}. Note that, $\mathcal{M}_k(s)\cap \mathcal{M}_k(s')=\emptyset$ if $s\neq s'$.
\end{myDef}

Given the current time step $\textit{dt}_k$, we use $s$ to denote the seed point of $o$ at $\textit{dt}_k$ (i.e., $o\in\mathcal{M}_k(s)$), while use $\tilde{s}$ to denote that at $\textit{dt}_{k-1}$ (i.e., $o\in\mathcal{M}_{k-1}(\tilde{s})$).

\begin{figure}[t]
\begin{center}
\subfigcapskip=-20pt
\includegraphics[width=0.4\textwidth]{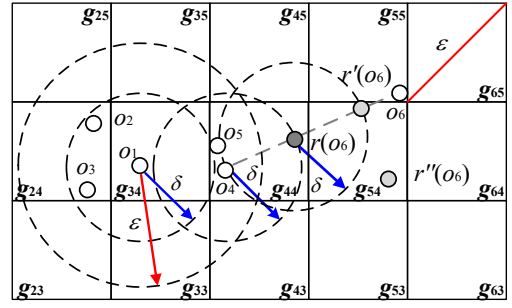}\vspace{-3mm}
\caption{$o_i\,(1\leq i \leq 6)$ at $\textit{dt}_2$ in Figure~\ref{fig:motivation}\,($\rho=3$)}
\vspace{-7mm}
\label{fig:t2}
\end{center}
\end{figure}

\begin{myEx}\label{ex:minimal group}
In Figure~\ref{fig:t1}b, there are two minimal groups, i.e., $\mathcal{M}_1({o_1})$\\$=\{o_1, o_2,  o_3\}$ and $\mathcal{M}_1({o_4})=\{o_4, o_5,  o_6\}$. In Figure~\ref{fig:t2}, there is only one minimal group before smoothing, i.e., $\mathcal{M}_2({o_1})=\{o_1, o_2,  o_3\}$. 
Further, given the current $k=2$, both $s$ and $\tilde{s}$ of $o_2$ is $o_1$ and $\tilde{s}.\tilde{l}=o_1.l_1$.

\end{myEx}
We propose to use the location of a seed point $s$ to smooth the location of a non-seed point $o\,(o\in \mathcal{M}_{k-1}(\tilde{s}))$ at $\textit{dt}_k$.
In order to guarantee the effectiveness of smoothing, 
Definition~\ref{def:seed_point} gives two constraints when generating minimal groups: (i) $d(o,s)<\delta\,(\delta\leq\varepsilon)$ and (ii) $\lvert \mathcal{M}_k(s) \rvert \geq \rho$. 
Setting $\delta$ to a small value, the first constraint ensures that $o\in \mathcal{M}_k(s)$ are close neighbors at $\textit{dt}_k$. Specifically, we require $\delta \leq \varepsilon$ because this makes it very likely that trajectories in the same minimal group are in the same cluster.
The second constraint avoids small neighbor sets $\mathcal{N}_{\delta}(s)$. Specifically, using an "uncertain border" point as a "pivot" to smooth the movement of other trajectories may lead to an abrupt change between clusterings or a low-quality clustering (according to the quality metrics of traditional clustering).
We present the algorithm for generating minimal groups in Section~\ref{sec:forming_micro_groups}.

Based on the above analysis, we formalize the historical cost of $o$ w.r.t. its adjustment $r(o)$ at $\textit{dt}_k$, denoted as $\mathcal{TC}_k(r(o))$, as follows.
\begin{equation}\label{f:TC0}
\begin{aligned}
&
\mathcal{TC}_k(r(o))=\left(\lceil\frac{d(r(o), \tilde{s}.l)}{\delta}\rceil-1 \right)^2\\
 &s.t.\, 
 d(r(o), o.\tilde{l})\leq \mu \cdot (o.t-o.\tilde{t}),
\end{aligned}
\end{equation}
where $o\in \mathcal{M}_{k-1}(\tilde{s})\backslash \{\tilde{s}\}$.
Given the threshold $\delta$, the larger the distance between $r(o)$ and $\tilde{s}.l$, the higher the historical cost. 
Here, we use the degree of closeness (i.e.,  $\lceil \frac{d(r(o), \tilde{s}.l)}{\delta}\rceil-1 $) instead of $d(r(o),\tilde{s}.l)$ to evaluate the historical cost, due to two reasons.
First, constraining the exact relative distance between any two trajectories during a time interval may be too restrictive, as it varies over time in most cases. Second, using the degree of closeness to constrain the historical cost is sufficient to obtain a smooth evolution of clusterings.

\paragraph{\textbf{Total cost $\bm{\mathcal{F}_k}$}}
Formulas~\ref{f:SC0} and~\ref{f:TC0} give the snapshot cost and historical cost for each trajectory $o$ w.r.t. its adjustment $r(o)$, respectively. However, the first measures the distance while the latter evaluates the degree of proximity. Thus, we normalize them to a specific range $[0,1)$:
\begin{equation}\label{f:sc}
    \mathcal{SC}_k(r(o))=\left(\frac{d(r(o), o.l)}{4\mu\cdot \Delta t+\delta}\right)^2  s.t. \,  d(r(o), o.\tilde{l})\leq \mu \cdot (o.t-o.\tilde{t})
\end{equation}
\begin{equation}\label{f:tc}
\begin{aligned}
      &\mathcal{TC}_k(r(o))= \left(\frac{\lceil\frac{d(r(o), \tilde{s}.l)}{\delta}\rceil-1}{\frac{4\mu\cdot \Delta t +\delta}{\delta}}\right)^2\\
      &s.t.\, 
 d(r(o), o.\tilde{l})\leq \mu \cdot (o.t-o.\tilde{t}),
\end{aligned}
\end{equation}
where $o\in \mathcal{M}_{k-1}(\tilde{s})\backslash \{\tilde{s}\}$ and $\Delta t$ is the duration of a time step.
Clearly, $\mathcal{SC}_k(r(o))\geq 0$ and $\mathcal{TC}_k(r(o))\geq 0$. Thus, we only need to prove  $\mathcal{SC}_k(r(o))<1$ and $\mathcal{TC}_k(r(o))<1$.

\begin{myLe}\label{le:speed1}
If $d(o.l, o.\tilde{l})\leq (o.t-o.\tilde{t}) \cdot \mu$ then
${d(r(o), o.l)} \leq 4\mu\cdot 
\Delta t$.
\end{myLe}
\begin{proof}
According to our strategy of mapping original timestamps (Section~\ref{sec:data_model}), $o.t-o.\tilde{t}\leq 2\Delta t$. Considering the speed constraint of the road network, $d(r(o),o.\tilde{l})\leq \mu \cdot(o.t-o.\tilde{t})$. Further, 
$d(r(o), o.l)\leq d(o.l, o.\tilde{l})+d(r(o), o.\tilde{l})$ due to the triangle inequality. Thus, we get $d(r(o), o.l) \leq  d(o.l, o.\tilde{l})+2\mu \cdot \Delta t$. Since $d(o.l, o.\tilde{l})\leq (o.t-o.\tilde{t}) \cdot \mu$, ${d(r(o), o.l)} \leq 4\mu\cdot 
\Delta t$.
\end{proof}
It follows from Lemma~\ref{le:speed1} that $\mathcal{SC}_k(r(o)) < 1$ if $d(o.l, o.\tilde{l})\leq (o.t-o.\tilde{t}) \cdot \mu$. 
However, $d(o.l, o.\tilde{l})\leq (o.t-o.\tilde{t}) \cdot \mu$ does not necessarily hold. 
To address this problem, we pre-process $o.l$ according to $o.\tilde{l}$ so that it follows the speed constraint before conducting evolutionary clustering. The details are given in Section~\ref{sec:speed_adjust}.
\begin{myLe}\label{le:normalize_tc}
If $d(\tilde{s}.l, \tilde{s}.\tilde{l})\leq  (\tilde{s}.t-\tilde{s}.\tilde{t})\cdot \mu$ then
$d(r(o), \tilde{s}.l)\leq 4\mu\cdot \Delta t+\delta$.
\end{myLe}
\begin{proof}
We have
$d(r(o), o.\tilde{l})\leq 2\mu\cdot \Delta t$ according to Lemma~\ref{le:speed1}. Further,
$d(o.\tilde{l}, \tilde{s}.\tilde{l})\leq \delta$ according to Definition~\ref{def:seed_point}. Since $d(r(o), \tilde{s}.l)\leq d(r(o), o.\tilde{l})+ d(o.\tilde{l}, \tilde{s}.l)\leq d(r(o), o.\tilde{l})+ d(o.\tilde{l}, \tilde{s}.\tilde{l})+ d(\tilde{s}.\tilde{l}, \tilde{s}.l)$, we get  $d(r(o), \tilde{s}.l)\leq 4\mu\cdot \Delta t+\delta$.
\end{proof}

According to Lemma~\ref{le:normalize_tc}, we can derive $\lceil\frac{4\mu\cdot \Delta t+\delta}{\delta}\rceil-1 < \frac{4\mu\cdot \Delta t+\delta}{\delta}$ and thus $\mathcal{TC}_k(r(o))<1$.
Letting $4\mu\cdot \Delta t +\delta=\pi$, 
the total cost $\mathcal{F}_k$ is:
\begin{equation}\label{f:F}
\begin{aligned}
    \mathcal{F}_k= &\sum_{o,\tilde{s}\in \Theta_k\wedge o\neq \tilde{s}}\frac{1}{\pi^2}\left( d(r(o),o.l)^2+\alpha\cdot\left(\delta\cdot \left(\lceil\frac{d(r(o), \tilde{s}.l)}{\delta}\rceil-1\right)\right)^2
    \right)\\
    &s.t.\, \forall o\in\Theta_k \, (d(r(o), o.\tilde{l})\leq \mu \cdot (o.t-o.\tilde{t})),
\end{aligned}
\end{equation}
where ${\Theta}_k=\mathcal{O}_k\cap (\bigcup_{\tilde{s}\in \mathcal{S}_{k-1}
}\mathcal{M}_{k-1}(\tilde{s}))$. 
Formula~\ref{f:F} indicates that we do not smooth the location of $o$ at $\textit{dt}_k$ if $o$ is not summarized in any minimal group at $\textit{dt}_{k-1}$. This is in accordance with the basic idea that we conduct smoothing by exploring the  neighboring trajectories.
We can now formulate our problem.
\begin{myDef}\label{def:eco}
Given a snapshot $\mathcal{O}_k$, a set of previous minimal groups $\bigcup_{\tilde{s}\in \mathcal{S}_{k-1}
}\mathcal{M}_{k-1}(\tilde{s})$, a time duration $\Delta t$, a speed constraint $\mu$, and parameters $\alpha$, $\delta$, $\varepsilon$, $\textit{minPts}$ and $\rho$, \textbf{\underline{e}volutionary \underline{c}lustering of streaming traject\underline{o}ries} (\textbf{ECO}) is to 
\begin{itemize}
\item find a set of adjustments $\mathcal{R}_{{k}_{opt}}$, such that $\mathcal{R}_{k_{opt}}=\arg \min_{\mathcal{R}_k} \mathcal{F}_k$;
\item compute a set of clusters $\mathcal{C}_k$ over $\mathcal{R}_{k_{opt}}$.
\end{itemize}
\end{myDef}
Specifically, each adjustment of $o.l \in \mathcal{R}_{k_{opt}}$ is denoted as $r_{opt}(o)$ and is then used as the previous location of $o$ (i.e. $o.\tilde{l}$) at $\textit{dt}_{k+1}$ for evolutionary clustering.

\begin{myEx}\label{ex:evolutionary_clustering}
Following Example~\ref{ex:minimal group},
ECO first finds a set of adjustments $\mathcal{R}_{{2}_{opt}}=\{r_{opt}(o_i)| 1 \leq i \leq 12\}$ at $\textit{dt}_2$. Then, it performs clustering over $\mathcal{R}_{{2}_{opt}}$ and gets $\mathcal{C}_2=\{c_1, c_2\}$, where $c_1=\{o_i|1\leq i \leq 6 \}$ and $c_2=\{o_i|7\leq i \leq 12 \}$.
Note that we only show $r_{opt}(o_6)(\,=r(o_6))$ and $r_{opt}(o_{10})(\,=r(o_{10}))$ in Figures~\ref{fig:motivation} and~\ref{fig:t2} because $r_{opt}(o_i)=o_i.l\,(1 \leq i \leq 12 \wedge i\neq 6\wedge i\neq 10)$ at $\textit{dt}_2$. 
\end{myEx}
%不连续也不可微
Clearly, the objective function in Formula~\ref{f:F} is neither continuous nor differentiable. Thus, computing the optimal adjustments using existing solvers involves iterative processes~\cite{song2015turn} that are too expensive for online scenarios. We thus prove that Formula~\ref{f:F} can be  solved approximately in linear time in Section~\ref{sec: Quadratic_time_solution}.

%% file: solution.tex
\section{Computation of Adjustments}\label{sec: Quadratic_time_solution}
Given the current time step $\textit{dt}_k$,
we start by decomposing $\mathcal{F}_k$ at the unit of minimal groups as follows,
\begin{equation}\label{f:F_decompose}
\begin{aligned}
   \mathcal{F}_k&=\sum_{\tilde{s}\in \mathcal{S}_{k-1}} f_k(\tilde{s}.l) \\
    &=\sum_{\tilde{s}\in \mathcal{S}_{k-1}}\sum_{o\in\Omega} \left( d(r(o),o.l)^2+\alpha \cdot\left(\delta\cdot \left(\lceil\frac{d(r(o), \tilde{s}.l)}{\delta}\rceil-1\right)\right)^2
    \right)\\
     &s.t.\, \forall o \in \Theta_k \, (d(r(o), o.\tilde{l})\leq \mu \cdot (o.t-o.\tilde{t})),
\end{aligned}
\end{equation}
where $\Theta_k=\mathcal{O}_k\cap (\bigcup_{\tilde{s}\in \mathcal{S}_{k-1}
}\mathcal{M}_{k-1}(\tilde{s}))$, $\Omega=\mathcal{M}_{k-1}(\tilde{s})\backslash \{\tilde{s}\}$,
$r(o)$ is the adjustment of $o.l$ at $\textit{dt}_k$, $\tilde{s}$ is the seed point of $o$ at $\textit{dt}_{k-1}$, and $\tilde{s}.l$ is the location of $\tilde{s}$ at $\textit{dt}_{k}$. We omit the multiplier $\frac{1}{\pi^2}$ from Formula~\ref{f:F} because $\Delta t$, $\mu$, and $\delta$ are constants and do not affect the results. 

\subsection{Linear Time Solution}\label{sec:smooth_non_seed_point}
We show that Formula~\ref{f:F_decompose} can be solved approximately in linear time. However,
Formula~\ref{f:F_decompose} uses each previous seed point $\tilde{s}$ for smoothing, and such points may also exhibit unusual behaviors from $\textit{dt}_{k-1}$ to $\textit{dt}_k$.
Moreover, $\tilde{s}$ may not be in $\mathcal{O}_k$.
We address these problems in Section~\ref{sec:smooth_seed_point} by proposing a seed point shifting strategy, and we assume here that  $\tilde{s}\in \mathcal{O}_k$ has already been smoothed, i.e., $r(\tilde{s})=\tilde{s}.l$.

\begin{myLe}\label{le:F_decompose}
$\mathcal{F}_k$ achieves the minimum value if each $f_k(\tilde{s}.l)\,(\tilde{s}\in \mathcal{S}_{k-1})$ achieves the minimum value.
\end{myLe}
\begin{proof}
To prove this, we only need to prove that $f_k(\tilde{s}.l)$ and ${f}_k(\tilde{s}'.\tilde{l})$\,$(\tilde{s}\neq \tilde{s}'\wedge \tilde{s}, \tilde{s}'\in \mathcal{S}_{k-1})$ do not affect each other. This can be established easily, as we require $\mathcal{M}_{k-1}(\tilde{s})\cap \mathcal{M}_{k-1}(\tilde{s}')=\emptyset$. We thus omit the details due to space limitation.
\end{proof}

Lemma~\ref{le:F_decompose} implies that Formula~\ref{f:F_decompose} can be solved by minimizing each $f_k(\tilde{s}.l)$ ($\tilde{s}\in \mathcal{S}_{k-1}$).
Next, we further "push down" the cost shown in Formula~\ref{f:F_decompose} to each pair of $o\,(o\in \mathcal{M}_{k-1}(\tilde{s})\backslash \{\tilde{s}\})$ and $\tilde{s}$.

\begin{eqnarray}\label{f:F_fixedcp}
 f_k(r(o), \tilde{s}.l)&=&\left(d(r(o),o.l)^2+\alpha \cdot\left(\delta\cdot \left(\lceil\frac{d(r(o), \tilde{s}.l)}{\delta}\rceil-1\right)\right)^2
    \right) \nonumber \\
 \quad &s.t.&\, d(r(o), o.\tilde{l})\leq \mu \cdot (o.t-o.\tilde{t})
\end{eqnarray}

\begin{myLe}\label{le:f_decompose}
$f_k(\tilde{s}.l)$ achieves the minimum value if each  $f_k(r(o),\tilde{s}.l)\,$\\$(o \in \mathcal{M}_{k-1}(\tilde{s})\cap \mathcal{O}_k\backslash \{\tilde{s}\})$ achieves the minimum value.
\end{myLe}
\begin{proof}
The proof is straightforward, because $f_k(r(o), \tilde{s}.l)$ and $f_k(r'(o), \tilde{s}.l)$ $(o, o' \in \mathcal{M}_{k-1}(\tilde{s})\cap \mathcal{O}_k\backslash \{\tilde{s}\}\wedge o\neq o')$ are independent of each other.
\end{proof}

According to Lemma~\ref{le:f_decompose}, the problem is simplified to computing $r_{opt}(o)=\arg \min_{r(o)} f_k(r(o), \tilde{s}.l)$ ($o \in \mathcal{M}_{k-1}(\tilde{s})\cap \mathcal{O}_k\backslash \{\tilde{s}\}$) given $\tilde{s}$. However, Formula~\ref{f:F_fixedcp} is still intractable as its objective function is not continuous. We thus aim to transform it into a continuous function.
Before doing so, we cover the case where the computation of $r_{opt}(o)$ w.r.t a trajectory $o$ can be skipped.

\begin{myLe}\label{le:no_adjust}
If $d(o.l, \tilde{s}.l) \leq \delta$ then $o.l=\arg \min_{r(o)} f_k(r(o), \tilde{s}.l)$. 
\end{myLe}
\begin{proof}
Let $r(o)\,(r(o)\neq o.l)$ be an adjustment of $o.l$.
Given $d(o.l, \tilde{s}.l) \leq \delta$, $\mathcal{TC}_k(o.l)=0
\leq \mathcal{TC}_k(r(o))$. On the other hand, as $d(r(o),o.l)>d(o.l,o.l)=0$, the snapshot cost $\mathcal{SC}_k(o.l)=0< \mathcal{SC}_k(r(o))$. Thus, $r_{opt}(o)=o.l$ if $d(o.l, \tilde{s}.l) \leq \delta$. 
\end{proof}
A previous study~\cite{kim2009particle} smooths the distance between each pair of neighboring nodes no matter their relative distances. In contrast, Lemma~\ref{le:no_adjust} suggests that if a non-seed point remains close to its previous seed point at the current time step, smoothing can be ignored.
This avoids over-smoothing close trajectories. 
Following Example~\ref{ex:evolutionary_clustering}, $o_2.l=\arg \min_{r(o_2)} f_2(r(o_2), o_1.l)$. 

\begin{myDef}
A \textbf{circle} is given by $\mathcal{Q}(e, x)$, where $e$ is the center and $x$ is the radius. 
\end{myDef}

\begin{myDef}
A \textbf{segment} connecting two locations $l$ and $l'$ is denoted as $\textit{se}(l,l')$.
The \textbf{intersection} of a circle $\mathcal{Q}(e, x)$ and a segment $\textit{se}(l, l')$ is denoted as $\textit{se}(l,l')  \oplus \mathcal{Q}(e, r)$.
\end{myDef}

Figure~\ref{fig:t2} shows a circle $\mathcal{Q}(o_1.l, \delta)$ that contains $o_1.l$, $o_2.l$, and $o_3.l$. Further, $r(o_6)=\textit{se}(o_6.l,o_4.l) \oplus \mathcal{Q}(o_4.l, \delta)$.

\begin{myLe}\label{le:notempty}
$\textit{se}(o.l, \tilde{s}.l) \cap \mathcal{Q}(o.\tilde{l},\mu \cdot (o.t-o.\tilde{t}))\neq \emptyset$.  
\end{myLe}
\begin{proof}
In Section~\ref{sec:problem_definition}, we constrain $o.t-o.\tilde{t} \leq \mu \cdot \Delta t$ before smoothing, which implies that $o.l\in \mathcal{Q}(o.\tilde{l}, \mu\cdot (o.t-o.\tilde{t}))$.  
Hence, $\textit{se}(o.l, \tilde{s}.l) \cap \mathcal{Q}(o.\tilde{l},\mu \cdot (o.t-o.\tilde{t}))\neq \emptyset$.
\end{proof}
In Figure~\ref{fig:t2_1}, given $o_6.t-o_6.\tilde{t}=3$, $o_6.l\in \textit{se}(o_6.l,o_4.l)\cap\mathcal{Q}(o_6.\tilde{l}, 3\mu)$.

\paragraph{\textbf{Omitting the speed constraint}}
We first show that without utilizing the speed constraint, an optimal adjustment $r_{opt'}(o)$ of $o.l$ that minimizes $f(r(o), \tilde{s}.l)$ can be derived in constant time. 
Based on this, we explain how to compute $r_{opt}$ based on $r_{opt'}(o)$.

\begin{myLe}\label{le:r_given_y}
$\forall r'(o) \notin \textit{se}(o.l,\tilde{s}.l) (\exists r(o) \in \textit{se}(o.l,\tilde{s}.l)  ( f_k(r(o), 
\tilde{s}.l)\leq f_k(r'(o), \tilde{s}.l)))$.

\end{myLe}
\begin{proof}
Let $d(o.l,\tilde{s}.l)=y$. 
First, we prove that $\forall r'(o) \notin \mathcal{Q}(\tilde{s}.l, y)$\\$ ( \exists r(o) \in \textit{se}(o.l,\tilde{s}.l) (f_k(r'(o), \tilde{s}.l)\geq f_k(r(o), \tilde{s}.l)))$.
Two cases are considered, i.e., (i) $d(r'(o), o.l) \leq y$ and (ii) $d(r'(o), o.l) > y$. For the first case, we can always find an adjustment $r(o)\in  \textit{se}(o.l,\tilde{s}.l)$, such that $d(r'(o), o.l)=d(r(o), o.l)$. Hence, $\mathcal{SC}_k(r'(o)) =\mathcal{SC}_k(r(o))$. However, we have $\mathcal{TC}_k(r'(o)) \geq \mathcal{TC}_k(r(o))$ due to $d(r'(o), \tilde{s}.l) > d(r(o), \tilde{s}.l)$. Thus, $f_k(r'(o), \tilde{s}.l) \geq f_k(r(o), \tilde{s}.l)$.
For the second case, it is clear that  $\forall r(o)\in \textit{se}(o.l,\tilde{s}.l)( \mathcal{SC}_k(r'(o)) >\mathcal{SC}_k(r(o))\wedge \mathcal{TC}_k(r'(o)) \geq \mathcal{TC}_k(r(o)))$. Thus, $f_k(r'(o), \tilde{s}.l)>f_k(r(o), \tilde{s}.l)$.

Second, we prove that $\forall r'(o)\in \mathcal{Q}(\tilde{s}.l, y)\backslash \textit{se}(o.l,\tilde{s}.l)  ( \exists r(o) \in \textit{se}(o.l,\tilde{s}.l) (f_k(r'(o), \tilde{s}.l)\geq f_k(r(o), \tilde{s}.l)))$. We can always find $r(o)\in  \textit{se}(o.l,\tilde{s}.l)$, such that $d(r'(o), \tilde{s}.l)=d(r(o), \tilde{s}.l)$. Hence, $\mathcal{TC}_k(r''(o)) $\\$=\mathcal{TC}_k(r(o))$. However, in this case $\mathcal{SC}_k(r'(o))$ $>\mathcal{SC}_k(r(o))$ due to $r(o)\in \textit{se}(o.l,\tilde{s}.l) \wedge r'(o)\notin \textit{se}(o.l,\tilde{s}.l)$.
Thus, we have $f_k(r'(o), \tilde{s}.l)$\\$>f_k(r(o), \tilde{s}.l)$.
\end{proof}
In Figure~\ref{fig:t2}, $f(r(o_6), o_4.l)\leq  f(r''(o_6), o_4.l)$ and $f(r'(o_6), o_4.l)\leq  f(r''(o_6), o_4.l)$ due to $r''(o_6)\notin \textit{se}(o_4.l, o_6.l)$.  Lemma~\ref{le:r_given_y} indicates that if we ignore the speed constraint in Formula~\ref{f:F_fixedcp}, we can search $r_{opt'}(o)$ just on $\textit{se}(o.l,\tilde{s}.l)$ without missing any result. 

\begin{figure}[t]
\begin{center}
\subfigcapskip=-20pt
\includegraphics[width=0.38\textwidth]{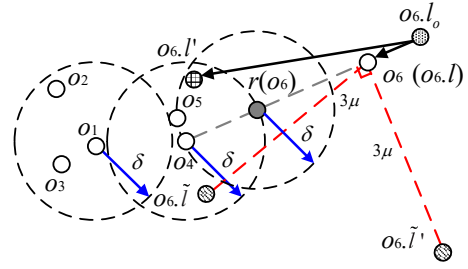}\vspace{-3mm}
\caption{An example of speed-based pre-processing, i.e., $\bm{o_6.l_o\rightarrow o_6}$, in Figure~\ref{fig:t2} ($\bm{\rho=3}$)}
\vspace{-5mm}
\label{fig:t2_1}
\end{center}
\end{figure}
\begin{myLe}\label{le:solution}
Let $d(r_{opt'}(o), \tilde{s}.l)=b_{opt'} \cdot \delta$. If $d(o.l, \tilde{s}.l)>\delta$ then $b_{opt'}\in\left\{{\bm{\mathrm{N}}^*}\cap [\lambda_1, \lambda_2]\right\}\cup \lambda_2$, where $\lambda_1=\frac{d(\tilde{s}.l, o.\tilde{l})-\mu \cdot (o.t-{o}.\tilde{t})}{\delta}$, $\lambda_2=\frac{d(o.l, \tilde{s}.l)}{\delta}$ 
and ${\bm{\mathrm{N}}}^*$ is the natural numbers.
\end{myLe}
\begin{proof}
We start by proving $\max\{\lambda_1,1\}  \leq b_{opt'} \leq  \lambda_2$. 
First, we have  $r_{opt'}(o)\in \textit{se}(o.l, \tilde{s}.l)$ according to Lemma~\ref{le:r_given_y}. Thus, $b_{opt'}\cdot \delta \leq d(o.l, \tilde{s}.l)$, i.e., $b_{opt'}\leq \lambda_2$.
Further, $\forall r(o)\in \textit{se}(o.l, \tilde{s}.l) (d(r(o),\tilde{s}.l)+ d(r(o), o.\tilde{l}) \geq d(\tilde{s}.l, o.\tilde{l}))$ due to the triangle inequality. 
Thus, $b_{opt'}\cdot \delta + \mu \cdot (o.t-{o}.\tilde{t}) \geq d(\tilde{s}.l, o.\tilde{l})$, i.e., $b_{opt'} \geq  \lambda_1$.
Moreover, $\forall r'(o)\in \textit{se}(o.l, \tilde{s}.l) ((0< d(r'(o), \tilde{s}.l) < \delta \wedge d (r (o), \tilde{s}.l ) = \delta) \Rightarrow (\mathcal{TC}(r'(o))=\mathcal{TC}(r(o)) \wedge \mathcal{SC}(r'(o))>\mathcal{SC}(r(o))))$. Thus, we get $b_{opt'}\geq 1$.

Next, we prove $\forall r'(o),r(o)\in \textit{se}(o.l, \tilde{s}.l) ((b-1)\cdot \delta <d(r'(o), \tilde{s}.l)$\\$<b\cdot \delta \wedge d(r(o), \tilde{s}.l)=b\cdot \delta  \wedge 1 \leq b\leq  \lfloor\lambda_2 \rfloor \wedge b\in \bm{\mathrm{N}}^*)\Rightarrow 
(f(r(o), \tilde{s}.l)<f(r'(o), \tilde{s}.l)))$.
According to Formula~\ref{f:tc}, $\mathcal{TC}(r(o))=\mathcal{TC}(r'(o))$. Further, $d(o.l, r'(o))=d(o.l,\tilde{s}.l)-d(r'(o), \tilde{s}.l)$ and $d(o.l, r(o))=d(o.l,\tilde{s}.l)$ $-b\cdot \delta$. As $d(r'(o), \tilde{s}.l)< b\cdot \delta$ we have $\mathcal{SC}(r(o)) <\mathcal{SC}(r'(o))$. Thus, %if $r_{opt'}(o)=b\cdot \delta$ then 
$b_{opt'}\notin [1, \lfloor \lambda_2 \rfloor ]\backslash \bm{\mathrm{N}}^*$.

Finally, we prove $\forall r'(o)\in \textit{se}(o.l, \tilde{s}.l) (b\cdot \delta \leq d(r'(o), \tilde{s}.l)  \wedge \lfloor\lambda_2 \rfloor < b \leq \lambda_2)\Rightarrow 
(f(o.l, \tilde{s}.l)<f(r'(o), \tilde{s}.l)))$.
Similar to the above proof, in this case $\mathcal{TC}(r'(o))=\mathcal{TC}(o.l) \wedge \mathcal{SC}(r'(o))>\mathcal{SC}(o.l)$. Thus, %if $r_{opt'}(o)=b\cdot \delta$ then 
$b_{opt'}\notin (\lfloor \lambda_2 \rfloor, \lambda_2]\backslash \{\lambda_2\}$.
\end{proof}
In Figure~\ref{fig:t2}, we have $r_{opt'}(o_6)\in\{o_6.l, r(o_6), r'(o_6)\}$.
Based on Lemmas~\ref{le:no_adjust} to~\ref{le:solution}, 
we let $d(r(o), \tilde{s}.l)=b\cdot \delta$ and simplify Formula~\ref{f:F_fixedcp} to the following function:
\begin{equation}\label{f:solution}
 \begin{aligned} 
  \quad f_k(b, \tilde{s}.l)&= \left(  d(o.l,\tilde{s}.l)-b\cdot \delta\right)^2+\alpha \cdot\left(\delta\cdot \left(b-1\right)\right)^2
  \\
 &s.t. \,\,
 b\in\left\{{\bm{\mathrm{N}}^*}\cap [\lambda_1, \lambda_2]\right\}\cup \lambda_2,
\end{aligned}
\end{equation}
where $\lambda_1=\frac{d(\tilde{s}.l, o.\tilde{l})-\mu \cdot (o.t-o.\tilde{t})}{\delta}$ and $\lambda_2=\frac{d(o.l, \tilde{s}.l)}{\delta}$.
The snapshot cost $\left(d(o.l,\tilde{s}.l)-b\cdot \delta\right)^2$ is derived  according to  Lemma~\ref{le:r_given_y}, i.e., $o.l, r_{opt'}(o)$ and $\tilde{s}.l$ are on the same line segment; 
while the historical cost  $\left(\delta\cdot \left(b-1\right)\right)^2$ is obtained by simply plugging $d(r(o),\tilde{s}.l)=b \cdot \delta$ into Formula~\ref{f:F_fixedcp}. 
The objective function in Formula~\ref{f:solution} is a continuous. Thus, the $b_{opt'}
\,(\in\left\{{\bm{\mathrm{N}}^*}\cap [\lambda_1, \lambda_2]\right\}\cup \lambda_2)$ that minimizes the function
can be obtained in constant time without sacrificing accuracy.

\begin{myEx}\label{ex:solution}
Continuing Example~\ref{ex:evolutionary_clustering} and given $d(o_6.l, o_4.l)= 25$, $\alpha=2.1$ and $\delta=10$, we get $b_{opt'}=1$ and $r_{opt'}(o_6)=r(o_6)$.
\end{myEx}

\paragraph{\textbf{Introducing the speed constraint}} Recall that $r_{opt'}(o)$ is the optimal adjustment of $o.l$ without taking the speed constraint in Formula~\ref{f:F_fixedcp} into account, while $r_{opt}(o)$ takes the constraint into account.
We have narrowed the range of $r_{opt'}$ to a set of discrete locations on $\textit{se}(o.l, \tilde{s}.l)$ without sacrificing any accuracy. Further, if $r_{opt'}\in \mathcal{Q}(o.\tilde{l},\mu \cdot \Delta t)$ then $r_{opt}(o)=r_{opt'}(o)$. However, if $r_{opt'}\notin \mathcal{Q}(o.\tilde{l},\mu \cdot \Delta t)$, $r_{opt'}(o)$ is an invalid adjustment. In this case, 
letting $d(r_{opt}(o),\tilde{s}.l)=b_{opt}\cdot\delta$, we propose to approximate $r_{opt}$ by 
searching only in the narrowed range of $r_{opt'}$, i.e., we propose to compute $b_{opt}$ approximately as follows.
\begin{equation}\label{f:optimal}
\begin{aligned}
 &
     b_{opt}=\arg \min_{b\in \left\{{\rm{N}^*}\cap [\lambda_1, \lambda_2]\right\}\cup \lambda_2}\lvert b-b_{opt'}\rvert\\
 &s.t. \,\, \mathcal{Q}(o.\tilde{l},\mu \cdot (o.t-o.\tilde{t}))\cap \mathcal{Q}(\tilde{s}.l, b\cdot \delta)\neq \emptyset,
\end{aligned}
%\vspace{-1mm}
\end{equation}
where $\mathcal{Q}(o.\tilde{l},\mu \cdot (o.t-o.\tilde{t}))\cap \mathcal{Q}(\tilde{s}.l, b\cdot \delta)\neq \emptyset$ indicates that $r_{opt}(o)\in \mathcal{Q}(o.\tilde{l},\mu \cdot (o.t-o.\tilde{t}))$ must hold due to $d(r_{opt}(o), \tilde{s}.l)=b_{opt}\cdot \delta$. After getting $b_{opt'}$, $b_{opt}$ can be located according to $\textit{se}(o.l, \tilde{s}.l)\oplus \mathcal{Q}(o.\tilde{l},\mu \cdot (o.t-o.\tilde{t}))$ in constant time.
Following Example~\ref{ex:solution} and given $o_6.t-o_6.\tilde{t}=3$ and $\mu=9$, $r_{opt}(o_6)=r(o_6)$ if $o_6.l_1=o_6.\tilde{l}$,
while $r_{opt}(o_6)=o_6.l$ if $o_6.l_1=o_6.\tilde{l}'$ (shown in Figure~\ref{fig:t2_1}). Specifically, in the latter case, $o_6.l$ is the only feasible solution of $r_{opt}(o_6)$ according to Formula~\ref{f:optimal}, as $\textit{se}(o_6.l,o_4.l)\oplus \mathcal{Q}(o_6.\tilde{l}', 3\mu)=o_6.l$.
Note that computing $r_{opt}$ using Formula~\ref{f:optimal} may not yield an optimal value that minimizes $f_k(r(o), \tilde{s}.l)$. This is because we approximate the feasible region of $r_{opt}$ by the narrowed range of $r_{opt'}$ and may miss an $r(o) \,(r(o)\in \mathcal{Q}(o.\tilde{l},\mu \cdot (o.t-o.\tilde{t}))\backslash \textit{se}(o.l, \tilde{s}.l))$ that minimizes Formula~\ref{f:F_fixedcp}.
However, experiments show that $r_{opt'}=r_{opt}$ in most case. The underlying reasons are that the maximum distance a trajectory can move under the speed limitation 
is generally far larger than the distance a trajectory actually moves between any two time steps and that we constrain $d(o.l, o.\tilde{l}) \leq \mu \cdot (o.t-o.\tilde{t})$ before smoothing, which "repairs" large noise to some extent. 
So far, the efficiency of computing $r_{opt}(o)$ using Formula~\ref{f:F_fixedcp} has been improved to $O(1)$ time complexity.

\subsection{Shifting of Seed Points}\label{sec:smooth_seed_point}
Section~\ref{sec:smooth_non_seed_point} assumes that the previous seed point $\tilde{s}.l$ evolves gradually when smoothing $o\in \mathcal{M}_{k-1}(\tilde{s})\cap\mathcal{O}_k\backslash \{\tilde{s}\}$ at $\textit{dt}_k$, which is not always true. Thus, $\tilde{s}.l$ may also need to be smoothed.
We first select a "pivot" for smoothing $\tilde{s}.l$.
An existing method~\cite{song2015turn} maps the noise point to the accurate point that is closest to it in a batch mode. Inspired by this, we smooth $\tilde{s}.l$ using $o.l$\,($o\in \mathcal{M}_{k-1}(\tilde{s})\cap\mathcal{O}_k$), which is a set of discrete locations. This is based on the observation that the travel companions of each trajectory evolves gradually due to the smooth movement of trajectories. Next, we determine which trajectory $o$ should be selected as a "pivot" to smooth $\tilde{s}.l$.

Evolutionary clustering assigns a low cost (cf. Formula~\ref{f:cost}) if the clusterings change smoothly during a short time period.
Since we use cost embedding, we consider the location of a trajectory $o$ as evolving smoothly if the distance between $o$ and $o'$ ($o'\in \mathcal{M}_{k-1}(\tilde{s})\cap\mathcal{O}_k\backslash\{ o\}$) varies only little between two adjacent time steps. This is essentially evaluated by $f_k(o)$ (cf. Formula~\ref{f:F_decompose}), which measures the cost of smoothing $o'$ ($o'\in \mathcal{M}_{k-1}(\tilde{s})\cap\mathcal{O}_k\backslash\{ o\}$) according to $o$.
Hence, we select the "pivot" for smoothing $\tilde{s}.l$ using the following formula:
\begin{equation}\label{f:shift}
   \tilde{s}_{\mathit{new}}=\arg\min_{o\in \mathcal{M}_{k-1}(\tilde{s})\cap\mathcal{O}_k} f_k(o)
\end{equation}
After obtaining a "pivot" $\tilde{s}_{\mathit{new}}$, instead of first smoothing $\tilde{s}.l$ according to $\tilde{s}_{\mathit{new}}.l$ and then smoothing $o.l (o\in \mathcal{M}_{k-1}(\tilde{s})\backslash \{\tilde{s}\})$ by $\tilde{s}.l$, we shift the seed point of $o\in \mathcal{M}_{k-1}(\tilde{s})\backslash \{\tilde{s}\}$ from $\tilde{s}$ to $\tilde{s}_{\mathit{new}}$ and use $\tilde{s}_{\mathit{new}}$ to smooth other trajectories $o\,( o\neq \tilde{s}_{\mathit{new}})$.
The reasons are: (i) by Formula~\ref{f:shift},   $\tilde{s}_{\mathit{new}}$ is the trajectory with the smoothest movement from $\textit{dt}_{k-1}$ to $\textit{dt}_{k}$ among trajectories in $\mathcal{M}_{k-1}(\tilde{s})$, and thus it is less important to smooth it; (ii) we can save $\lvert \mathcal{M}_{k-1}(\tilde{s})\rvert-1$ computations in Formula~\ref{f:F_decompose}. Formula~\ref{f:shift} suggests that the seed point may not be shifted, i.e., $\tilde{s}_{\mathit{new}}$ may be $\tilde{s}$.
Intuitively, when computing $\tilde{s}_{\mathit{new}}$, the locations of all trajectories in their corresponding minimal group are smoothed; and with the seed point shifting strategy, smoothing does not require that the previous seed point is active at the current time step.
\begin{myEx}
Continuing Example~\ref{ex:solution}, given $f_2(o_6.l)=\min\{f_2(o_4.l),$\\$ f_2(o_5.l), f_2(o_6.l)\}$, $\tilde{s}_{\mathit{new}}=\tilde{s}=o_6$. When calculating $\tilde{s}_{\mathit{new}}$, we get $r_{opt}(o_4)=o_4.l$, $r_{opt}(o_5)=o_5.l$, and $r_{opt}(o_6)=r(o_6)$.
\end{myEx}
The time complexity of smoothing a minimal group $\mathcal{M}_{k-1}(\tilde{s})$ is $O(\lvert \mathcal{M}_{k-1}(\tilde{s})\rvert^2)$.

\subsection{Speed-based Pre-processing}\label{sec:speed_adjust}
We present the pre-processing that forces each to-be-smoothed trajectory $o \in \mathcal{O}_k\cap (\bigcup_{\tilde{s}\in \mathcal{S}_{k-1}
}\mathcal{M}_{k-1}(\tilde{s}))$ to observe the speed constraint. The pre-processing guarantees the correctness of  the normalization of the snapshot and historical costs and can repairs large noise to some extent.
We denote the location of $o$ before  pre-processing as $o.l_o$ and the possible location after as $o.l_p$. 

A naive pre-processing strategy is to map $o.l_o$ to a random location on or inside $\mathcal{Q}(o.\tilde{l}, \mu \cdot \Delta t)$. However, this random strategy may make the smoothing less reasonable.

\begin{myEx}\label{ex:speed_adjust_motivation}
Continuing Example~\ref{ex:solution}, 
Figure~\ref{fig:t2_1} shows two possible locations $o_6.l$ and $o_6.l'$ of $o_6.l_o$, both of which are chosen at random while observing the speed constraint, i.e., they are located on and inside $\mathcal{Q}(o_6.\tilde{l}, 3\mu)$, respectively. Since $d(o_6.l', o_4.l)<\delta<d(o_6.l, o_4.l)$, the adjustment of $o_6.l'$ is $o_6.l'$ itself while that of $o_6.l$ is $r(o_6)$. 
\end{myEx}

In this example, $o_6.l'$ is less reasonable than $r(o_6)$. Specifically, according to the minimum change principle~\cite{song2015turn}, the changes to the data distribution made by the speed-based pre-processing and the neighbor-based smoothing should be as small as possible. However, considering $d(o_6.l_o, o_4.l)$, $o_6.l'$ is too close to $o_4.l$ compared with $o_6.l$ and $r(o_6)$. Given a pre-processed location $o.l$, its change due to smoothing has already been minimized through Formula~\ref{f:F_fixedcp}. Thus, to satisfy the minimum change principle, we just need to make the impact of speed-based pre-processing on neighbor-based smoothing as small as possible. 
Hence, we find the pre-processed $o.l$ via the speed constraint as follows.
\begin{equation}\label{f:speed_adjust}
\begin{aligned}
    &o.l=\arg\min_{o.l_p}\left\lvert d(o.l_p,\tilde{s}.l)-  d(o.l_o,\tilde{s}.l)\right\rvert
\\
    & s.t. \, d(o.l_p-o.\tilde{l})\leq \mu \cdot (o.t-o.\tilde{t})
\end{aligned}
\end{equation}
This suggests that the difference between $d(o.l,\tilde{s}.l)$ and $d(o.l_o,\tilde{s}.l))$ is expected to be as small as possible, in order to mitigate the effect of speed-based pre-processing on computing historical cost. Before applying Formula~\ref{f:speed_adjust}, we pre-process $\tilde{s}.l$ so that it also follows the speed constraint:
\begin{equation}\label{f:speed_adjust_seed}
\begin{aligned}
    &\tilde{s}.l=\arg\min_{\tilde{s}.l_p} d(\tilde{s}.l_o,\tilde{s}.l_p)
\\
    & s.t. \, d(\tilde{s}.l_p-\tilde{s}.\tilde{l})\leq \mu \cdot (\tilde{s}.t-\tilde{s}.\tilde{t})
\end{aligned}
\end{equation}
As $\tilde{s}$ is not smoothed by any trajectories in $\mathcal{M}_{k-1}(\tilde{s})$ (cf. Section~\ref{sec:smooth_non_seed_point}),
Formula~\ref{f:speed_adjust_seed} lets the closest location to $\tilde{s}.l_p$ satisfying the speed constraint be $\tilde{s}.l$. This is also in accordance with the minimum change principle~\cite{song2015turn}.
According to the seed point shifting strategy, we examine each $o\in \mathcal{M}_{k-1}(\tilde{s})\cap\mathcal{O}_k$ to identify the most smoothly moving trajectory as $\tilde{s}_{\mathit{new}}$. Thus, before this process, we have to force each $o\in \mathcal{M}_{k-1}(\tilde{s})\cap\mathcal{O}_k$ to follow the speed constraint w.r.t. the current to-be-examined seed point $\tilde{s}$, i.e., computing $o.l$ w.r.t. $\tilde{s}.l$ according to Formula~\ref{f:speed_adjust}.
Obviously, a speed-based pre-processing is only needed when $d(o.l_o-o.\tilde{l})> \mu \cdot (o.t-o.\tilde{t})$; otherwise, $o.l=o.l_o$.
Formulas~\ref{f:speed_adjust} and~\ref{f:speed_adjust_seed} can be computed in constant time.  
%Following Example~\ref{ex:speed_adjust_motivation},
%$o_6.l$ is the final result w.r.t $o_6.l_o$ after speed-based pre-processing.

%% file: algorithm.tex
\section{Algorithms}\label{sec:algorithm}
We first introduce a grid index and then present the algorithms for generating minimal groups, smoothing locations, and performing the clustering, together with a set of optimization techniques.
\begin{algorithm} [tb]
\small
\LinesNumbered
\caption{Generating minimal groups}\label{alg:micro_group}
	\KwIn{a set of trajectories $\mathcal{O}_k$, a threshold $\delta$}
	\KwOut{a set of seed point $\mathcal{S}_k$ and minimal groups $\bigcup_{s\in \mathcal{S}_k}\mathcal{M}_k(s)$}
   \For{ each $o \in \mathcal{O}_k $}{
    $\mathcal{S}_{k}\leftarrow \mathcal{S}_{k}\cup o$ if $\forall s\in \mathcal{S}_{k} (d(o.l,s.l)>\delta$)
    }\For{ each $o \in \mathcal{O}_k\backslash \mathcal{S}_k$}{
    $\mathcal{M}_{k}(s)\leftarrow\mathcal{M}_{k}(s)\cup o$ \\ $s\leftarrow \arg\min_{\{s\in \mathcal{I}_{\delta}(g(o)) \wedge d(o.l,s.l)\leq \delta\}} d(o.l,s.l)$\\
    }
    \Return {$\mathcal{S}_k$ and $\bigcup_{s\in \mathcal{S}_k}\mathcal{M}_k(s)$}
\end{algorithm}
\subsection{Grid Index}
We use a grid index~\cite{gan2017dynamic} to accelerate our algorithms. Figure~\ref{fig:t2} shows an example  index.
Specifically, the diagonal
of each grid cell (denoted as $g$) has length ${\varepsilon}$, which is the parameter used in DBSCAN ~\cite{gan2017dynamic}. This accelerates the process of finding core points.
The number of trajectories that fall into $g$ is denoted as $\lvert g \rvert$. 
Given $o\in g$ at $\textit{dt}_k$ and $o\in \mathcal{M}_{k}(s)$, $\mathcal{G}(g)$ is the collection of grid cells $g'$, such that $o'\in \mathcal{M}_{k}(s)\wedge o'\in g'$.
Following Example~\ref{ex:minimal group}, $\mathcal{G}(g_{34})=\{g_{24}, g_{34}\}$, as shown in Figure~\ref{fig:t2}. 
The smallest distance between the boundaries of two grid cells, $g$ and $g'$, is denoted as $\min(g,g')$. Clearly, $\min(g,g)=0$.
For example in Figure~\ref{fig:t2}, $\min(g_{24}, g_{44})=\varepsilon$.
Next, we introduce the concept of $h$-closeness~\cite{gan2017dynamic}.

\begin{myDef}\label{def:e_close}
Two grid cells $g$ and $g'$ are  $\bm{h}$-\textbf{\textit{close}}, if $\min(g,g')\leq h$.
The set of the $\bm{h}$-\textbf{\textit{close}} grid cells of $g$ is denoted as $\mathcal{I}_h(g)$.
\end{myDef}

\begin{myLe}\label{le:grid}
For $o\in g$, we have $d(o.l, o'.l)>h$ if $o'\in g'\wedge g'\notin \mathcal{I}_h(g)$.
\end{myLe}
The proof is straightforward. We utilize two distance parameters, i.e., $\varepsilon$ for clustering (cf. Definition~\ref{def:core_point}) and $\delta$ for finding minimal groups (cf. Definition~\ref{def:seed_point}). Thus, we only need to consider $\mathcal{I}_h(g)$, where $h=\varepsilon, \delta$.  Following again existing work~\cite{gan2017dynamic}, we define 
$\mathcal{I}_\varepsilon(g_{ij})=\Omega\backslash (g_{i_1j_1}\cup g_{i_1j_2} \cup g_{i_2j_1} \cup g_{i_2j_2})$, where $\Omega=\{g_{i'j'}| i_1 \leq i'\leq  i_2\wedge j_1 \leq j'\leq j_2 \}$ and $i_1=i-2$, $i_2=i+2$, $j_1=j-2$ and $j_2=j+2$. For example in Figure~\ref{fig:t2},  $\mathcal{I}_{\varepsilon}(g_{44})=\Omega\backslash (g_{22} \cup g_{62} \cup g_{26} \cup g_{66})$, where $\Omega=\{g_{ij}| 2\leq i,j \leq 6 \}$.
Since we set $\delta<\varepsilon$, we only need to compute $\mathcal{I}_h(g)$, such that $h<\varepsilon$.

\begin{myLe}\label{le:grid_prune}
$\mathcal{I}_h(g_{ij})=\mathcal{I}_{\varepsilon}(g_{ij})$ if $\frac{\varepsilon}{\sqrt{2}} \leq h<\varepsilon$; otherwise $\mathcal{I}_h(g_{ij})=\{g_{i'j'}| i-1 \leq i'\leq i+1 \wedge j-1 \leq j' \leq j+1 \}$.
\end{myLe}

The proof of Lemma~\ref{le:grid_prune} follows from the grid cell width being $\frac{\varepsilon}{\sqrt{2}}$. In Figure~\ref{fig:t2}, given $\frac{\varepsilon}{\sqrt{2}} \leq h<\varepsilon$,   $\mathcal{I}_{\delta}(g_{44})=\mathcal{I}_{\varepsilon}(g_{44})$.

\subsection{Generating Minimal Groups}\label{sec:forming_micro_groups}
Sections~\ref{sec:ec} and \ref{sec: Quadratic_time_solution} indicate that $o$ is smoothed at $\textit{dt}_k$ if $\exists \mathcal{M}_{k-1}(\tilde{s})(o\in \mathcal{M}_{k-1}(\tilde{s}))$; otherwise, $o$ is considered as an "outlier," to which neighbor-based smoothing cannot be applied.
Thus, we aim to include as many trajectories as possible in the minimal groups, in order to smooth as many trajectories as possible.

According to the above analysis, an optimal set of minimal groups should satisfy $\forall o\in \mathcal{O}_k\backslash (\bigcup_{s\in \mathcal{S}_{k}}\mathcal{M}_{k}(s)), \forall s\in \mathcal{S}_k$\, $(d(o,s)>\delta)$. Clearly,
$\bigcup_{s\in \mathcal{S}_{k}}\mathcal{M}_{k}(s)$, the set of trajectories in the minimal groups, is determined given $\mathcal{S}_k$. Definition~\ref{def:seed_point} implies that the local density of a seed point $s$ should be not small, i.e., $\lvert \mathcal{N}_{\delta}(s)\rvert \geq \rho$. It guarantees that there is at least one trajectory $o\,(=s)$ in a minimal group that is not located at the border of a cluster.
Considering the above requirement and constraint, we have to enumerate all the possible combinations to get the optimal set of $\mathcal{S}_k$, which is infeasible.

Therefore, we propose a greedy algorithm, shown in Algorithm~\ref{alg:micro_group}, for computing a set of minimal groups at $\textit{dt}_k$. Each trajectory $o$ is mapped to a grid cell $g$ before generating minimal groups.
We first greedily determine $\mathcal{S}_k$ and then generate minimal groups according to $\mathcal{S}_k$. This is because a non-seed point $o$ attached to a minimal group $\mathcal{M}_{k}(s)$ at the very beginning may turn out to be closer to another newly obtained seed point $s'$. This incurs repeated processes for finding a seed point for $o$. Instead, we compute the seed point $s$ for each $o\in \mathcal{O}_k$ exactly once.
According  to Lemma~\ref{le:grid}, we will not miss any possible seed point for $o$ by searching $\mathcal{I}_{\delta}(o)$ rather than $\mathcal{S}_k$ (Line 5). Note that Algorithm~\ref{alg:micro_group} generates minimal groups $\mathcal{M}_k(s) $ such that $\lvert \mathcal{M}_k(s) \rvert< \rho$. We simply ignore these during smoothing.

\begin{algorithm} [tb]
\small
\LinesNumbered
\caption{Smoothing}\label{alg:smoothing}
	\KwIn{a minimal group $\mathcal{M}_{k-1}(\tilde{s})$}
	\KwOut{a set of adjustments $\mathcal{R}_k(\mathcal{M}_{k-1}(\tilde{s}))$}
    $\textit{sum}\leftarrow 0$, $\mathcal{R}_k(\mathcal{M}_{k-1}(\tilde{s}))\leftarrow\emptyset$,
     $\textit{glp}\leftarrow\infty$, $\tilde{s}_{\mathit{new}}\leftarrow\textit{null}$, $\mathcal{A}\leftarrow\emptyset$\\
     \For{each $o' \in \mathcal{M}_{k-1}(\tilde{s})\cap \mathcal{O}_k$}{
      compute $o'.l$ according to Formula~\ref{f:speed_adjust_seed}, $\mathcal{A}\leftarrow o'.l$, $\textit{sum}\leftarrow0$\\ 
      \For{each $o \in \mathcal{M}_{k-1}(\tilde{s})\cap \mathcal{O}_k\backslash\{ o'\}$}{
      compute $o.l$ according to Formula~\ref{f:speed_adjust}\\
     compute  $r_{opt}(o)$ according to Formulas~\ref{f:solution} and ~\ref{f:optimal}\\
     $\textit{sum}\leftarrow \textit{sum}+f_k(r_{opt}(o), o'.l)$\\
      \If{$\textit{sum}\geq \textit{glp}$} 
     { break\qquad \qquad /*\,$o'$ must not be $\tilde{s}_{\mathit{new}}$\,*/}
       $\mathcal{A}\leftarrow \mathcal{A}\cup r_{opt}(o)$\\
     }
    
     \If{$\lvert \mathcal{A} \rvert=\lvert \mathcal{M}_{k-1}(\tilde{s})\rvert$}{
    $\textit{glp}\leftarrow \textit{sum}$, $\tilde{s}_{\mathit{new}}\leftarrow o'$, $\mathcal{R}_k(\mathcal{M}_{k-1}(\tilde{s}))\leftarrow \mathcal{A}$
    }
    }
    \Return{$\mathcal{R}_k(\mathcal{M}_{k-1}(\tilde{s}))$}
\end{algorithm}

\subsection{Evolutionary Clustering}\label{sec:evolutionary_clustering_alg}
\paragraph{\textbf{Smoothing}}
Algorithm~\ref{alg:smoothing} gives the pseudo-code of the smoothing algorithm. We maintain $\textit{glp}$ to record the current minimal $f_k(o'.l)=\sum_{o \in \mathcal{M}_{k-1}(\tilde{s})\cap \mathcal{O}_k\backslash \{o'\}}f_k(r(o), o'.l)\,(o'\in \mathcal{M}_{k-1}(\tilde{s})\cap \mathcal{O}_k)$
and maintain $\mathcal{A}$ to record adjustments w.r.t. $o'.l$ (Line 1). The computation of $f_k(o'.l)$ is terminated early if its current value exceeds $\textit{glp}$ (Lines 8--9).
As can be seen, if $o'\,(o'\in\mathcal{M}_{k-1}(\tilde{s}))\cap \mathcal{O}_k$ is identified as the trajectory with the smoothest movement in its minimal group, the set of adjustments $r(o)\, (o\in \mathcal{M}_{k-1}(\tilde{s})\cap \mathcal{O}_k)$ w.r.t $o'.l$ is returned, i.e., $\mathcal{R}_k(\mathcal{M}_{k-1}(\tilde{s}))$.

\paragraph{\textbf{Optimizing modularity}}
Modularity is a well-known quality measure for clustering~\cite{kim2009particle, yin2021multi}, which is computed as follows.
\begin{equation}\label{f:qs}
    QS=\sum_{c\in \mathcal{C}_k}\left(\frac{IS(c)}{TS}-
    \left(\frac{DS(c)}{TS}\right)^2\right)
\end{equation}
\noindent Here, $\textit{TS}$ is the sum of similarities of all pairs of trajectories, $\textit{IS}(c)$ is the sum of similarities of all pairs of trajectories in cluster $c$, $\textit{DS}(c)$ is the sum of similarities between a trajectory in cluster $c$ and any trajectory in cluster $c'\,(c'\in \mathcal{C}_k \backslash \{c\})$.
A high $\textit{QS}$ indicates a good clustering result. The similarity between any two trajectories $o$ and $o'$ is defined as $\frac{1}{d(o.l, o'.l)}$.

A previous study~\cite{kim2009particle} iteratively adjusts $\varepsilon$ to find the (local) optimal $\textit{QS}$ as well as the clustering result at each time step. Specifically, given an $\varepsilon$, a constant $\Delta\varepsilon$, and the current clustering result $\mathcal{C}_k$,
it calculates three modularity during each iteration: $\textit{QS}_{h}$, $\textit{QS}_{l}$, and $\textit{QS}_{m}$. $\textit{QS}_{m}=\textit{QS}$ is the modularity of $\mathcal{C}_k$.  $\textit{QS}_{h}$ is calculated from pairs in $\mathcal{C}_k$ with a similarity in the range $[\varepsilon, \varepsilon+\Delta\varepsilon]$, and $\textit{QS}_{l}$ is calculated from pairs in $\mathcal{C}_k$ with a similarity in the range $[0, \varepsilon-\Delta\varepsilon)$. Then $\varepsilon$ is adjusted as follows.
\begin{itemize}
    \item If $\textit{QS}_h=\max\{\textit{QS}_{h},\textit{QS}_{l}, \textit{QS}_{m}\}$, $\varepsilon$ increases by $\Delta\varepsilon$;
    \item If $\textit{QS}_l=\max\{\textit{QS}_{h},\textit{QS}_{l}, \textit{QS}_{m}\}$, $\varepsilon$ decreases by $\Delta\varepsilon$;
    \item If $\textit{QS}_m=\max\{\textit{QS}_{h},\textit{QS}_{l}, \textit{QS}_{m}\}$, $\varepsilon$ is unchanged.
\end{itemize}
The first two cases leads to another iteration of calculating $\textit{QS}_{h}$, $\textit{QS}_{l}$, and $\textit{QS}_{m}$ using the newly updated $\varepsilon$, while the last case terminates the processing.
This iterative optimization of modularity~\cite{kim2009particle} has a relatively high time cost. We improve the cost by only updating $\varepsilon$ at $\textit{dt}_{k}$ (denoted as $\varepsilon_{k}$) once to "approach" the (local) optimal modularity of $\mathcal{C}_k$.
Although $\varepsilon_k$ is then used for clustering at $\textit{dt}_{k+1}$ instead of at  $\textit{dt}_{k}$, the quality of the clustering is generally still improved, as the clustering result evolves gradually. Specifically, $\varepsilon$ is still obtained by the iterative optimization at the first time step.

\begin{algorithm} [tb]
\small
\LinesNumbered
\caption{Evolutionary clustering (ECO)}\label{alg:Evolutionary_Clustering}
	\KwIn{a snapshot $\mathcal{O}_k$, a clustering result $\mathcal{C}_{k-1}$, a set of minimal groups $\bigcup_{s\in \mathcal{S}_{k-1}}\mathcal{M}_{k-1}(\tilde{s})$, thresholds $\delta$, $\rho$, $\alpha$, $\varepsilon_{k-1}$, $\Delta\varepsilon$, $\textit{minPts}$}
	\KwOut{a clustering result ${C}_k$ 
	}
%build a grid index according to $\varepsilon_{k-1}$ \\
  smooth $\mathcal{O}_k$ and get the set of adjustments $\mathcal{R}_{k_\mathit{opt}}$  \, /* Algorithm~\ref{alg:smoothing} */\\
  build a grid index according to $\varepsilon_{k-1}$ and $\mathcal{R}_{k_\mathit{opt}}$\\
 generate minimal groups based on $\mathcal{R}_{k_\mathit{opt}}$ \, \quad\,\,  \quad\, /* Algorithm~\ref{alg:micro_group} */\\
 cluster $\mathcal{R}_{k_\mathit{opt}}$ to get $\mathcal{C}_k$ \qquad \qquad \qquad \qquad\qquad \, \,\,/* DBSCAN */\\
 update $\varepsilon_{k-1}$ to $\varepsilon_k$ according to $\mathcal{C}_k$    \qquad  \qquad \, \\
map $c\in \mathcal{C}_k$ to $c'\in \mathcal{C}_{k-1}$ \qquad \quad \quad \qquad \qquad \, \, /*\, literature~\cite{kim2009particle}\,*/\\ 
 \Return{$\mathcal{C}_k$}
\end{algorithm}

\paragraph{\textbf{Grid index and minimal group based accelerations}}
As we set the grid cell width to $\frac{\varepsilon}{\sqrt{2}}$, each $o\in g$ is a core point if $\lvert g \rvert\geq \textit{minPts}$~\cite{gan2017dynamic}. Similarly, if $\lvert \mathcal{M}_k(s) \rvert\geq \rho$, $s$ is a core point.
These efficient checks accelerate the search for core points as well as DBSCAN. 
In Example~\ref{ex:solution} and given \textit{minPts}=3, $o_4$, $o_5$ and $o_6$ are core points due to $\lvert g_{44} \rvert=3$ and $o_1$ is core point due to $\lvert \mathcal{M}_2(o_1) \rvert=3$.

\paragraph{\textbf{Evolutionary clustering  of streaming trajectories}}
All pieces are now in place to present the algorithm for evolutionary clustering of streaming trajectories (ECO), shown in Algorithm~\ref{alg:Evolutionary_Clustering}. The sub-procedures in lines 1--5 are detailed in the previous sections. Note that, as we perform evolutionary clustering at each time step with an updated $\varepsilon$, a grid index is built at each time step once locations arrive. The time cost of this is neglible ~\cite{li2021trace}. Also note that, the grid index is built after smoothing, as the locations of trajectories are changed. Finally, we connect clusters in adjacent
time steps with each other (Line 6) as proposed in the literature~\cite{kim2009particle}. This mapping aims to find the evolving, forming, and dissolving relationships between $c'_k\in \mathcal{C}_{k-1}$ and $c'_k\in \mathcal{C}_{k}$. Building on Examples~\ref{ex:cluster} and~\ref{ex:evolutionary_clustering}, $c_1 \in \mathcal{C}_1$ evolves to $c_1 \in \mathcal{C}_2$ while $c_2 \in \mathcal{C}_1$ evolves to $c_2 \in \mathcal{C}_2$, and no clusters form or dissolve. The details of the mapping are available elsewhere~\cite{kim2009particle}.
The time complexity of ECO at time step $\textit{dt}_k$ is $O(\lvert \mathcal{O}_k \rvert^2)$.

%% file: experiment.tex
\section{Experiments}\label{sec:experiments}
We report on extensive experiments aimed at achieving insight into the performance of ECO.
\subsection{Experimental Design}
\paragraph{\textbf{Datasets}.}
Two real-life datasets, Chengdu (CD) and Hangzhou (HZ), are used. The CD dataset is collected from 13,431 taxis over one day (Aug. 30, 2014) in Chengdu, China. It contains 30 million GPS records.
The HZ dataset is collected from 24,515 taxis over one month (Nov. 2011) in Hangzhou, China. It contains 107 million GPS records. The sample intervals of CD and HZ are 10s and 60s.

\paragraph{\textbf{Comparison algorithms and experimental settings.}} 
We compare with three methods: 
\begin{itemize}
    \item Kim-Han~\cite{kim2009particle} is a representative density-based evolutionary clustering method. It evaluates costs at the individual distance level to improve efficiency.
    \item DYN~\cite{yin2021multi} is the state-of-the-art evolutionary clustering. It adapts a particle swarm algorithm and random walks to improve result quality.
    \item OCluST~\cite{mao2018online} is the state-of-the-art for traditional clustering of streaming trajectories that disregards the temporal smoothness. It continuously absorbs newly arriving locations and updates representative trajectories maintained in a novel structure for density-based clustering.
\end{itemize}
In the experiments, we study the effect on performance of the parameters summarized in Table~\ref{tb:paramater}. $\Delta\varepsilon$ is set to 50 on both datasets.
The number of generations and the population size of DYN~\cite{yin2021multi} are both set to 20, in order to be able to process large-scale streaming data. Other parameters are set to their recommended values~\cite{yin2021multi, kim2009particle, mao2018online}. 
We compare with Kim-Han~\cite{kim2009particle} and DYN~\cite{yin2021multi} because, to the best of our knowledge, no other evolutionary clustering methods exist for trajectories. To adapt these two to work on clustering trajectories, we construct a graph on top of the GPS data by 
adding  an edge between two locations (nodes) $o.l$ and $o'.l$ if $d(o.l, o'.l)\leq \tau$, where $\tau=3000$ on CD and $\tau=1000$ on HZ. All algorithms are implemented in C++, and the experiments are run on a computer with an Intel Core i9-9880H CPU (2.30 GHz) and 32 GB memory.

\begin{table}
\centering\small
\caption{Parameter ranges and default values}\vspace{-4mm} 

\begin{tabular}{|c|c|}
 \bottomrule 
{ }\textbf{Parameter}&{ }\textbf{Range}\\ \hline
	{ }\textit{minPts} { }&{ }2, 4, 5, 6, 7, \textbf{8}, 10\\\hline
	{ }$\delta$ &{ } {200}, 300, \textbf{400}, 500, \textbf{600}, 700, 800 {}\\\hline
	{ } $\alpha$ &{ } 0.1, 0.3, 0.5, 0.7, \textbf{0.9}\\\hline
		{ } $\rho$ &{ } 4, 5, \textbf{6}, 7, 8\\
\bottomrule
\end{tabular}\label{tb:paramater}\vspace{-4mm}
\end{table}

\paragraph{\textbf{Performance metrics.}} 
We adopt \textit{modularity} $\textit{QS}$ (cf. Formula~\ref{f:qs}) to measure the quality of clustering.
We report $\textit{QS}$ as average values over all time steps.
The higher the $\textit{QS}$, the better the clustering. 
Moreover, we use \textit{normalized mutual information} $\textit{NMI}$~\cite{strehl2002cluster} to measure the similarity between two clustering results obtained at consecutive time steps.
\begin{equation}
    NMI= \frac{I(\mathcal{C}_{k-1}; \mathcal{C}_{k})}{\sqrt{H(\mathcal{C}_{k-1})\cdot H(\mathcal{C}_{k})}},
\end{equation}
where $I(\mathcal{C}_{k-1}; \mathcal{C}_{k})$ is the mutual information between clusters $\mathcal{C}_{k-1}$ and $\mathcal{C}_k$, $H(\mathcal{C}_{k})$ is the entropy of cluster $\mathcal{C}_k$.  
Specifically, the reported $\textit{NMI}$ values are averages over all time steps. Clusters evolve more smoothly if $\textit{NMI}$ is higher.
Finally, efficiency is measured as the \textit{average processing time} per record at each time step.

\begin{figure} \centering
\subfigcapskip=-5pt
  \subfigure[{CD dataset}]{      \includegraphics[scale=.17]{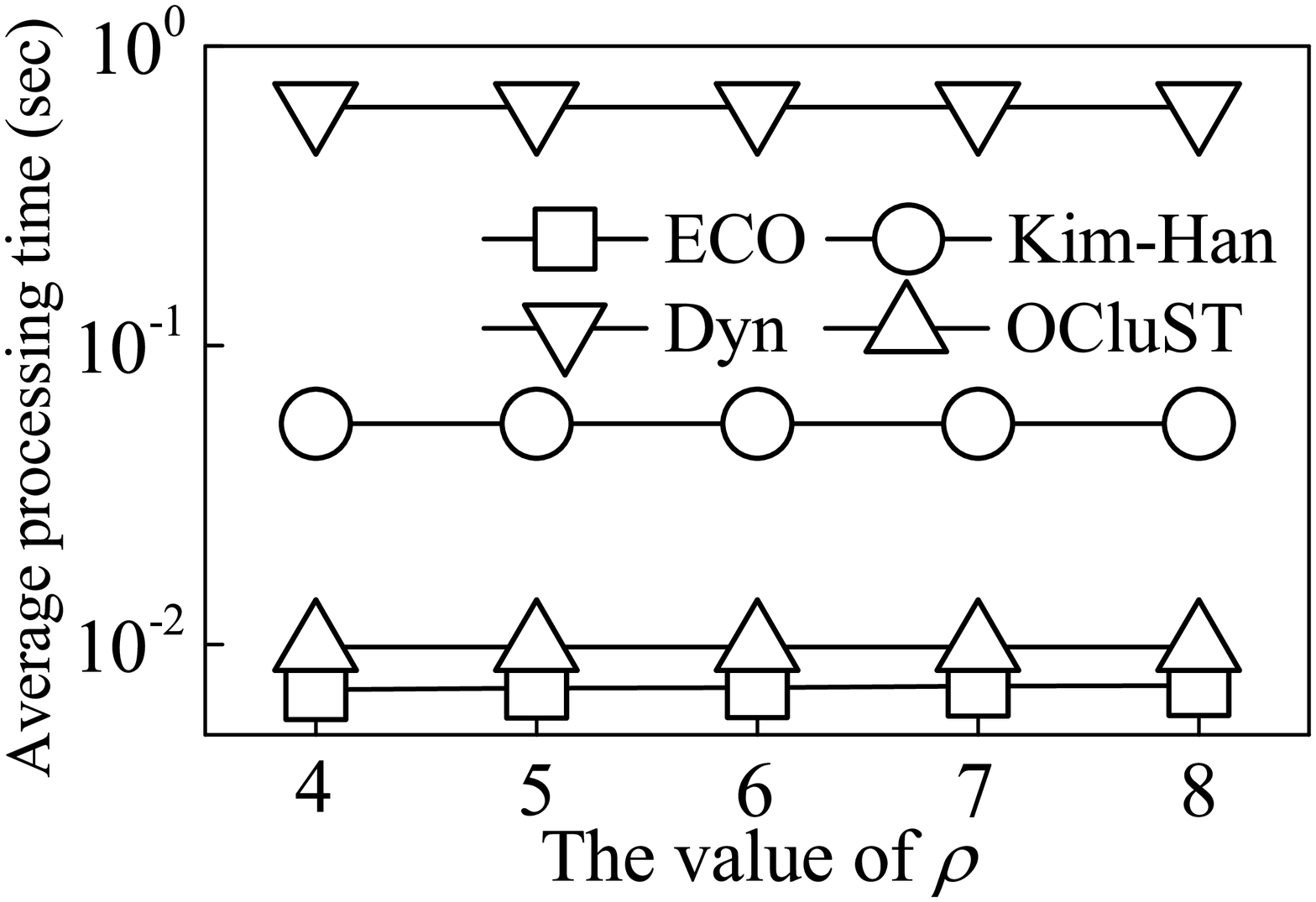}} 
 \subfigure[{HZ dataset}]{
\includegraphics[scale=.17]{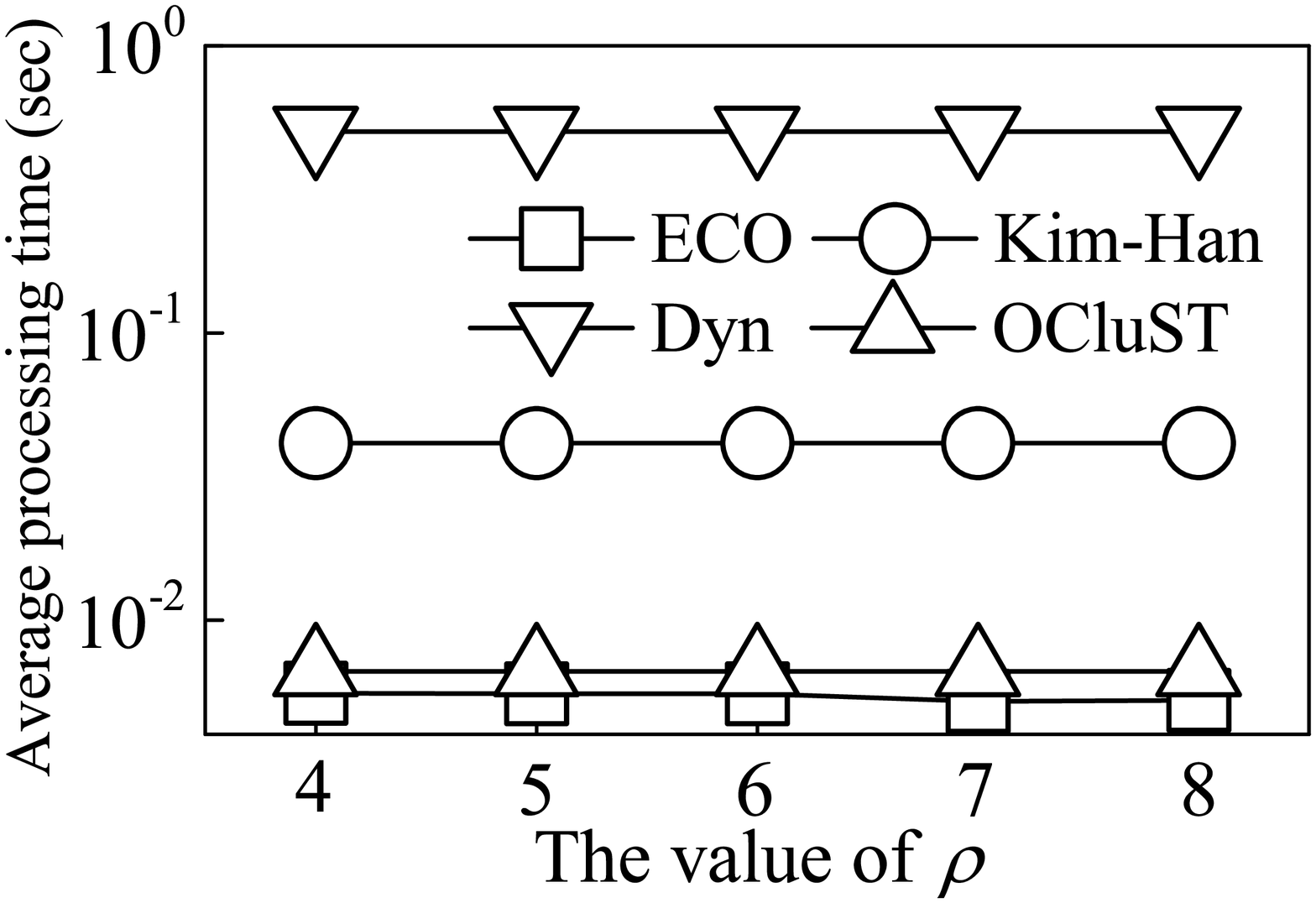}}
\subfigure[{CD dataset}]{      \includegraphics[scale=.17]{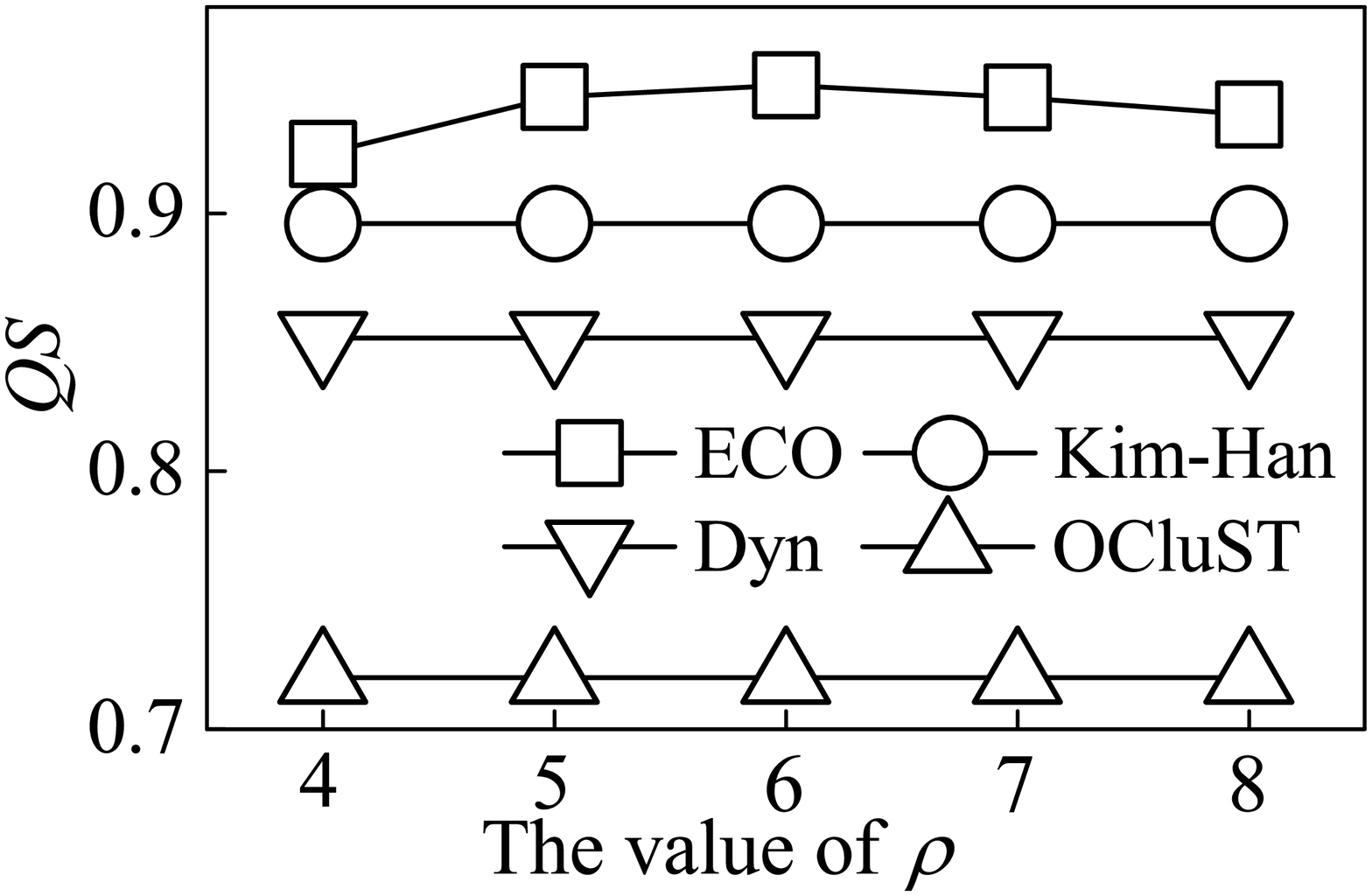}} 
 \subfigure[{HZ dataset}]{
\includegraphics[scale=.17]{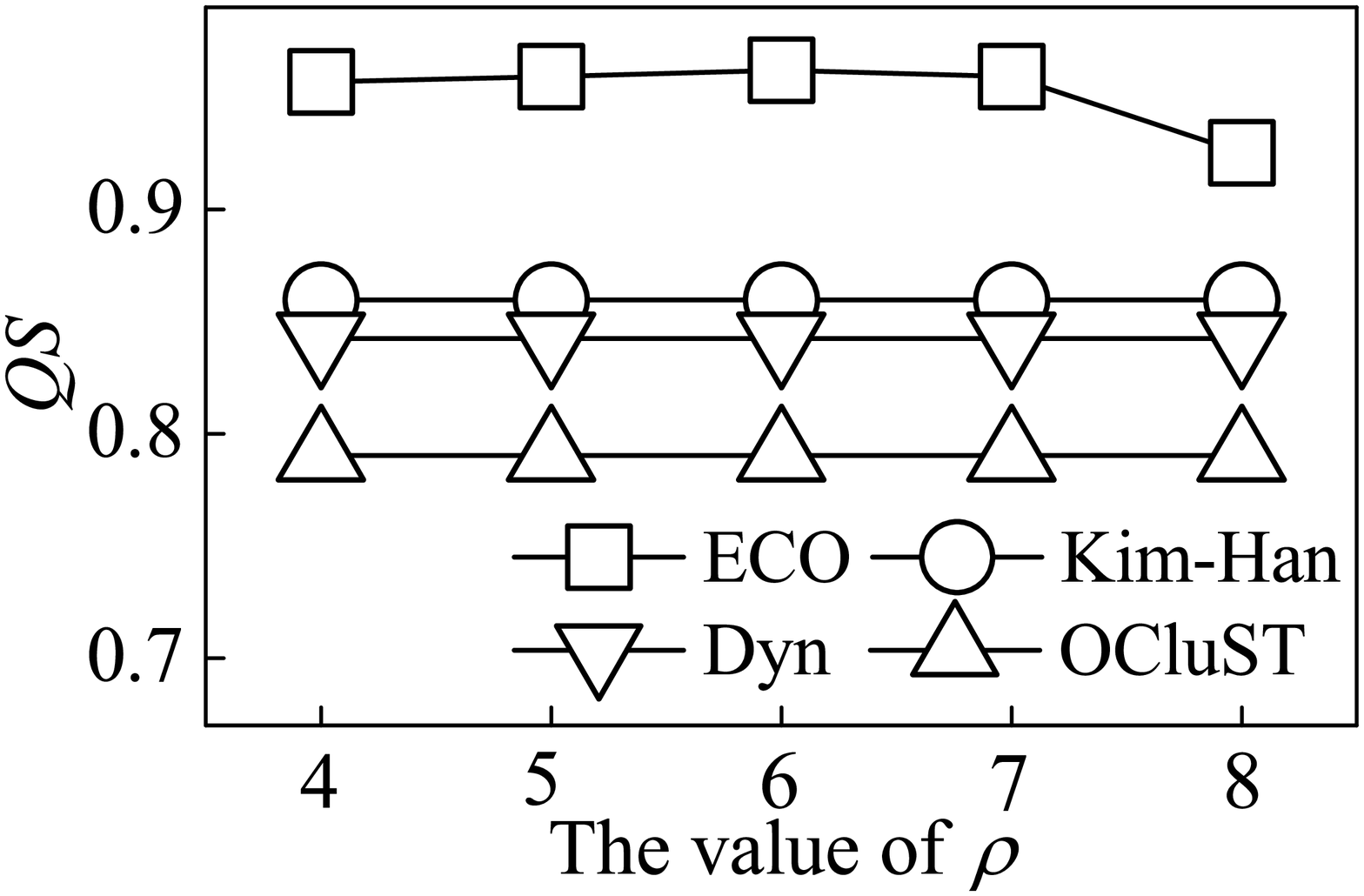}}
\subfigure[{CD dataset}]{      \includegraphics[scale=.17]{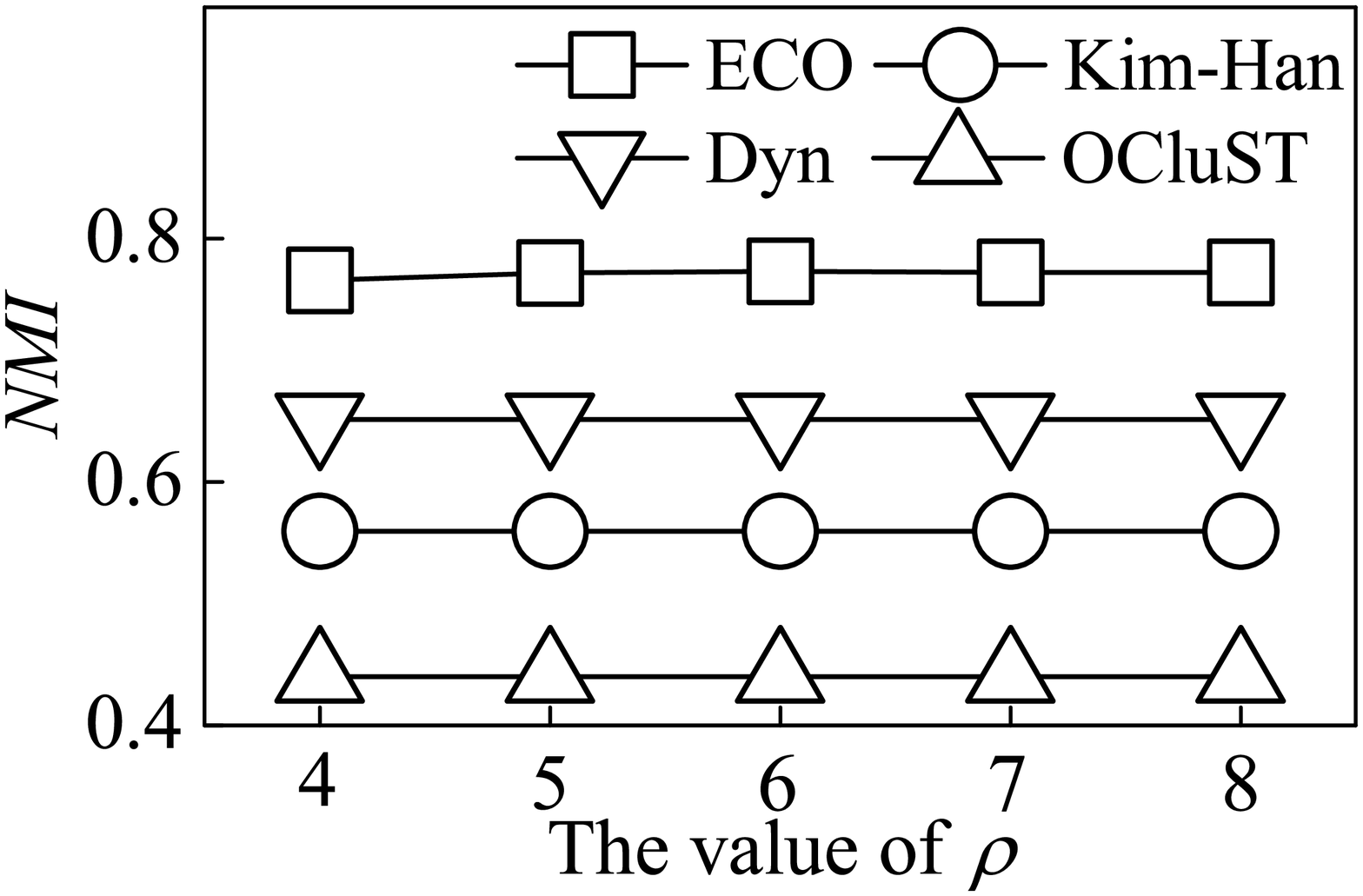}} 
 \subfigure[{HZ dataset}]{
\includegraphics[scale=.17]{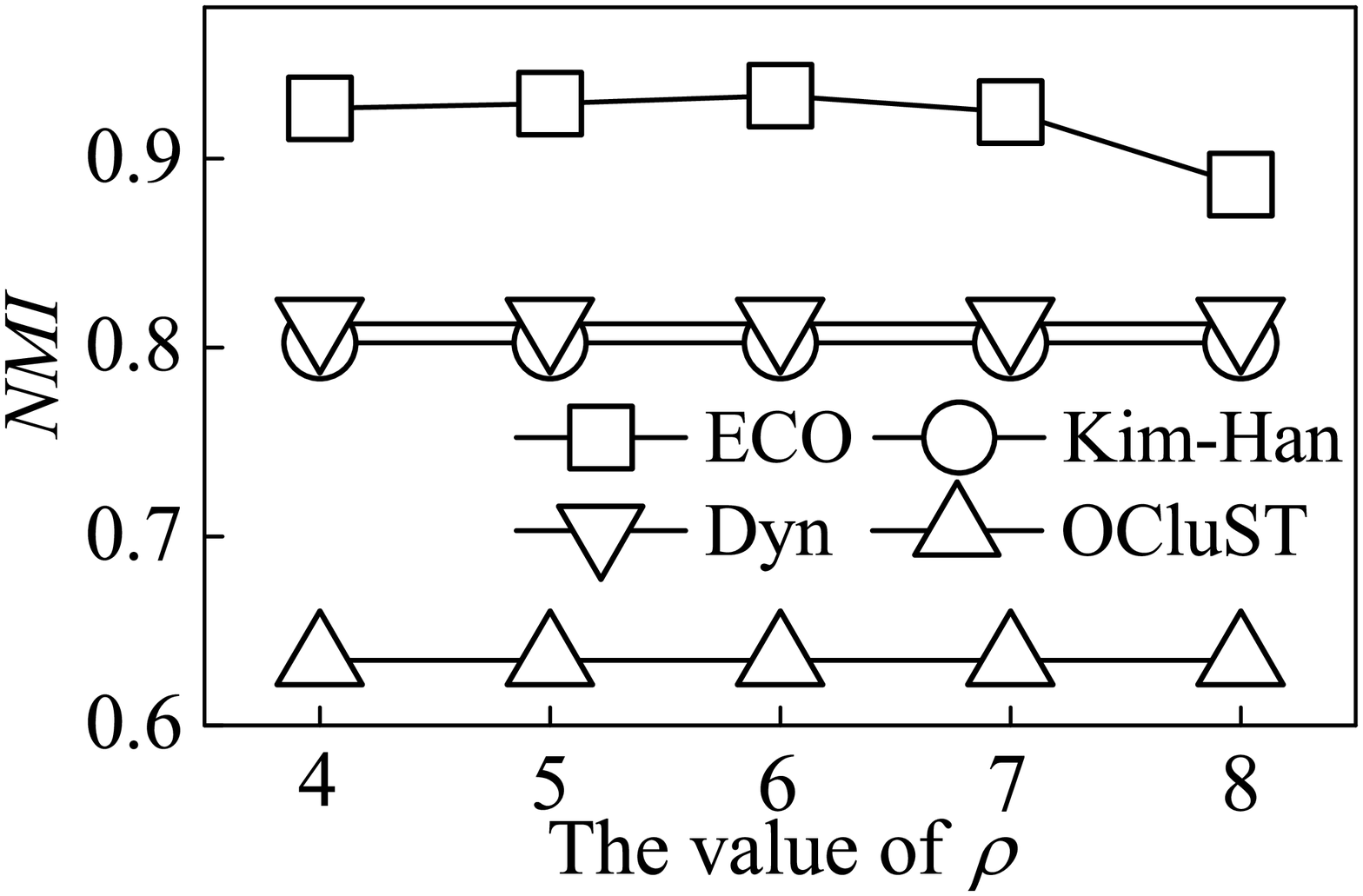}}
\vspace{-5mm}
\caption{Effects of varying $\mathbf{\rho}$}
\label{fig:rho}
\vspace{-5mm}
\end{figure}

\subsection{Comparison and Parameter Study}
We study the effect of parameters (summarized in Table~\ref{tb:paramater}) on the performance of the four methods.
\paragraph{\textbf{Effects of varying $\bm{\rho}$}}
Figure \ref{fig:rho} reports on the effect of varying $\delta$. First, ECO generally outperforms the baselines in terms of all performance metrics. In particular, ECO's average processing time is almost one order of magnitude lower than Kim-Han and almost two orders of magnitude better than  that of Dyn on both datasets. Moreover, ECO is even slightly more efficient than
OCluST, because the latter updates its data structure repeatedly for macro-clustering. 
The high efficiency and quality of ECO are mainly due to three reasons: (i) Except for the initialization, ECO excludes iterative processes and is accelerated by grid indexing and the proposed optimizing techniques; 
(ii) ECO takes into account temporal smoothness, which is designed specifically for trajectories (cf. Formula~\ref{f:F}); (iii) Locations with the potential to incur mutation of a clustering are adjusted to be closer to their neighbors that evolve smoothly, generally increasing the intra-density and decreasing the inter-density of clustering.

Second, we consider the effects of varying $\rho$. Figures~\ref{fig:rho}a and~\ref{fig:rho}b show that the average processing time is relatively stable. This is because the most time-consuming process in ECO is the clustering, which depends highly on the volume of data arriving at each time step. All four methods achieve higher efficiency on HZ than CD, due to CD's larger average data size of each time step. Figures~\ref{fig:rho}c--\ref{fig:rho}f show that as $\rho$ grows, $\textit{QS}$ and $\textit{NMI}$ first increase and then drop. 
On the one hand, trajectories with high local density are generally more stable, i.e., more likely to remain in the same cluster in adjacent time steps.
On the other hand, with a too large $\rho$, few minimal groups are generated, and thus few locations are smoothed. As the baselines  do not have parameter $\rho$, their performance is unaffected.

\paragraph{\textbf{Effects of varying $\bm{\delta}$}}
\begin{figure} \centering
\subfigcapskip=-5pt
  \subfigure[{CD dataset}]{      \includegraphics[scale=.17]{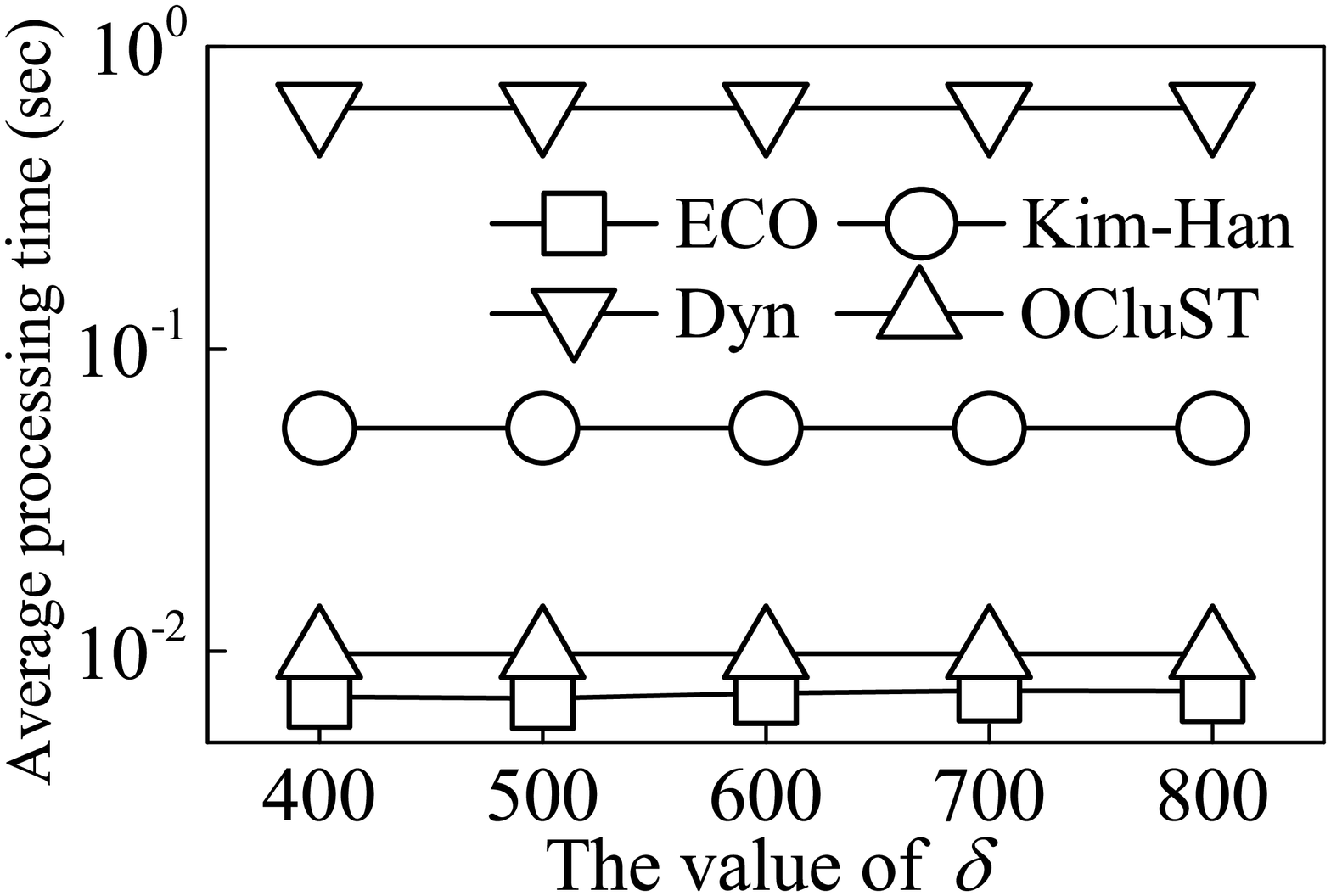}} 
 \subfigure[{HZ dataset}]{
         \includegraphics[scale=.17]{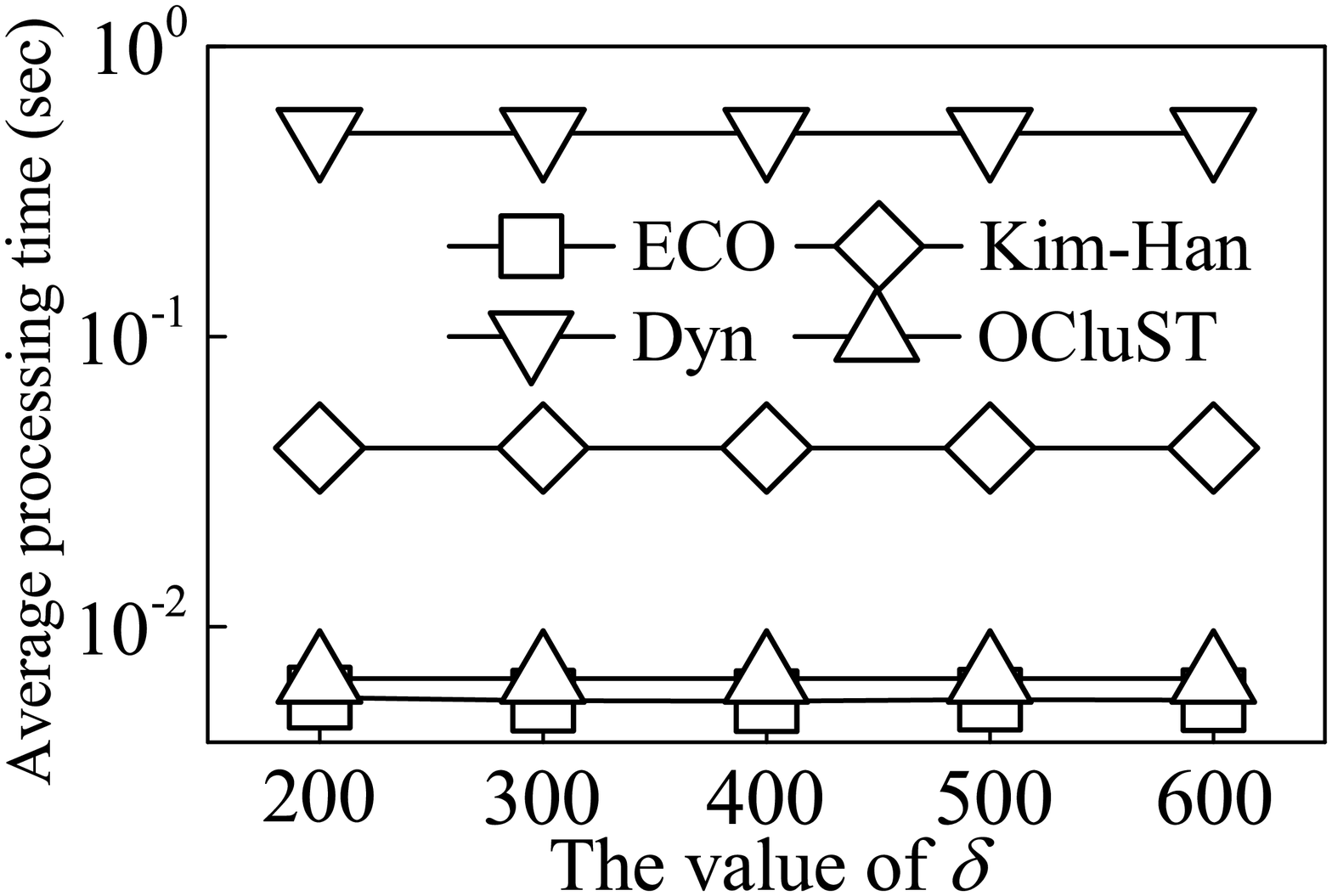}}
           \subfigure[{CD dataset}]{      \includegraphics[scale=.17]{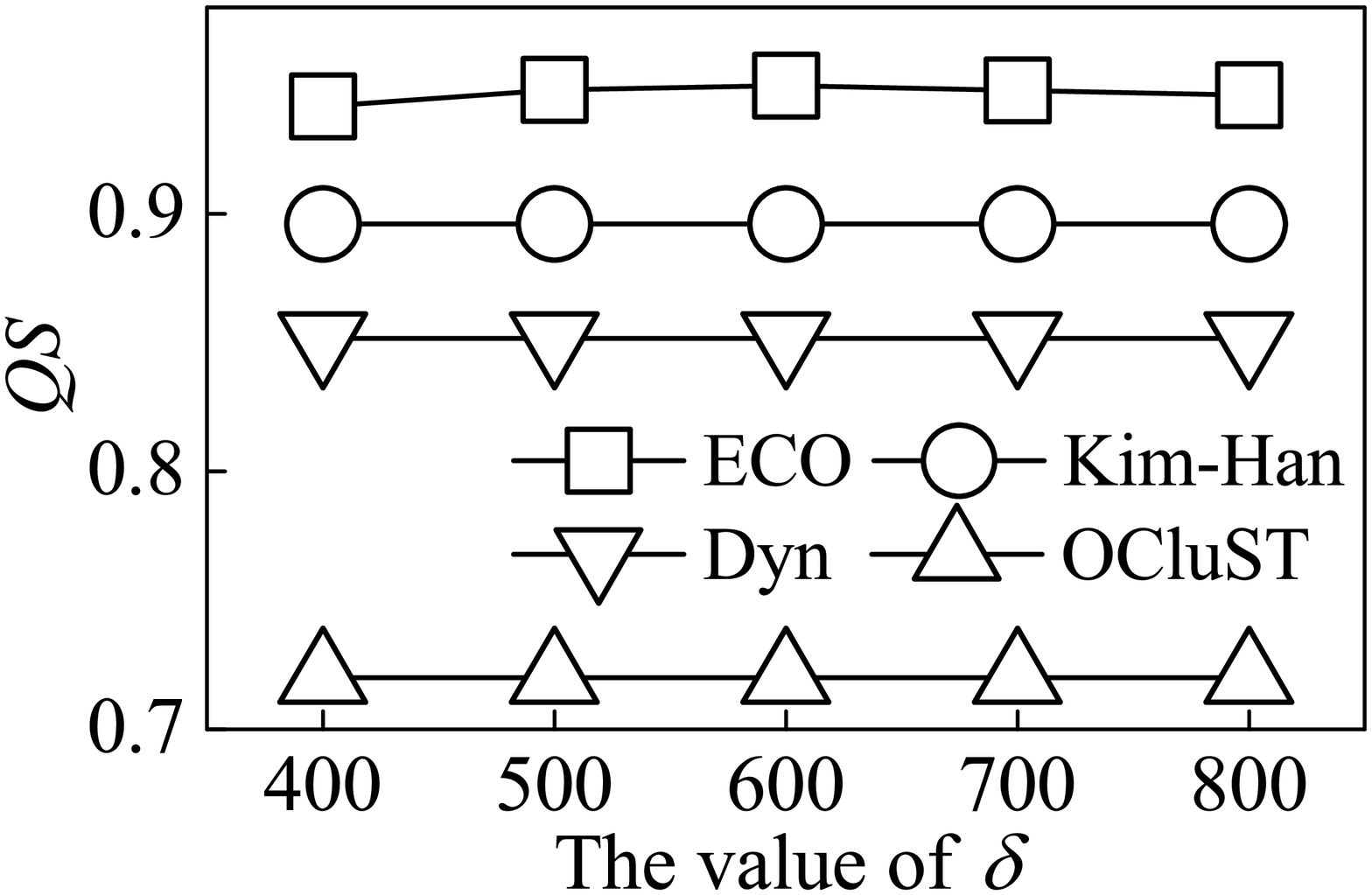}} 
 \subfigure[{HZ dataset}]{
         \includegraphics[scale=.17]{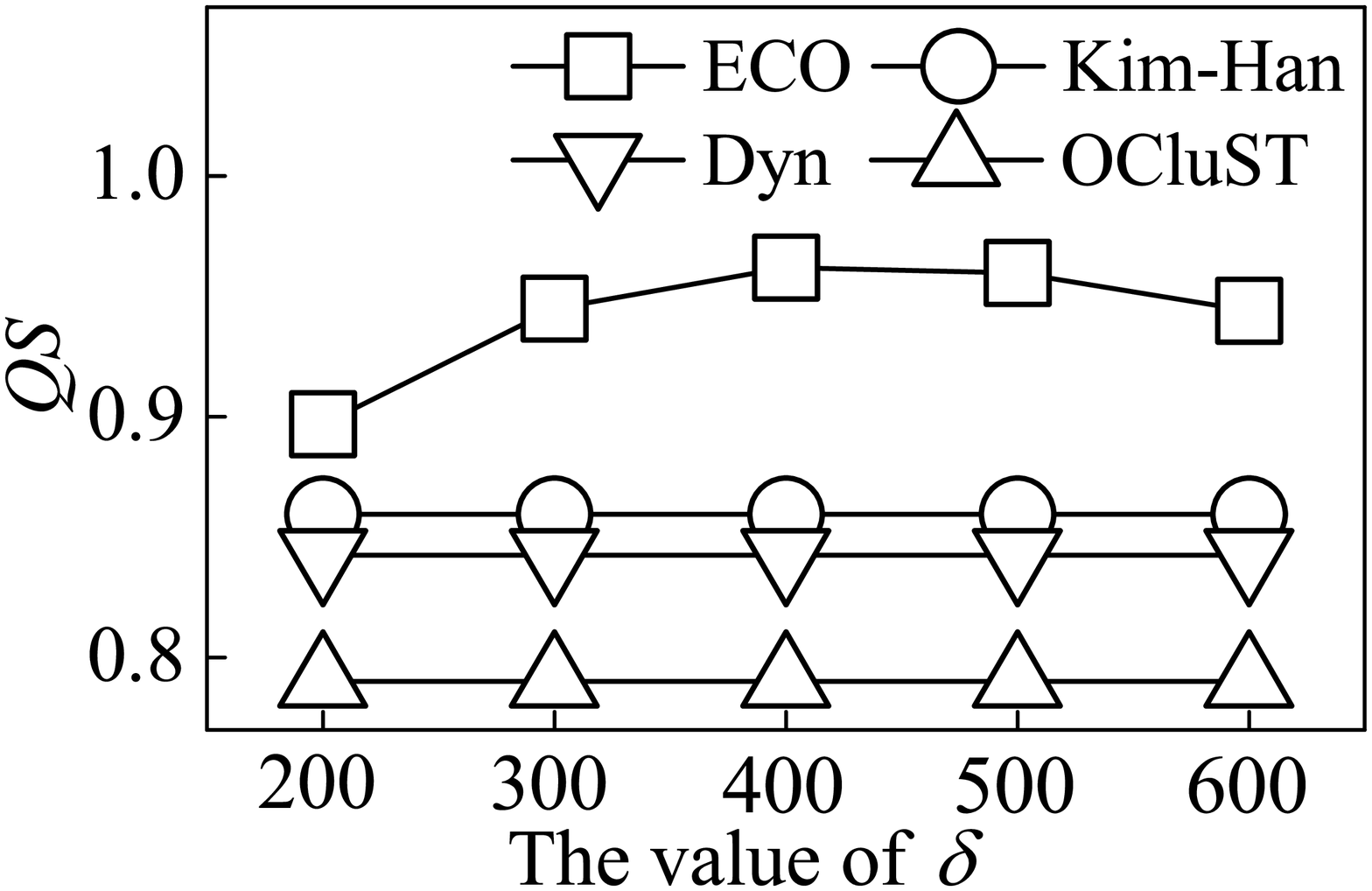}}
           \subfigure[{CD dataset}]{\includegraphics[scale=.17]{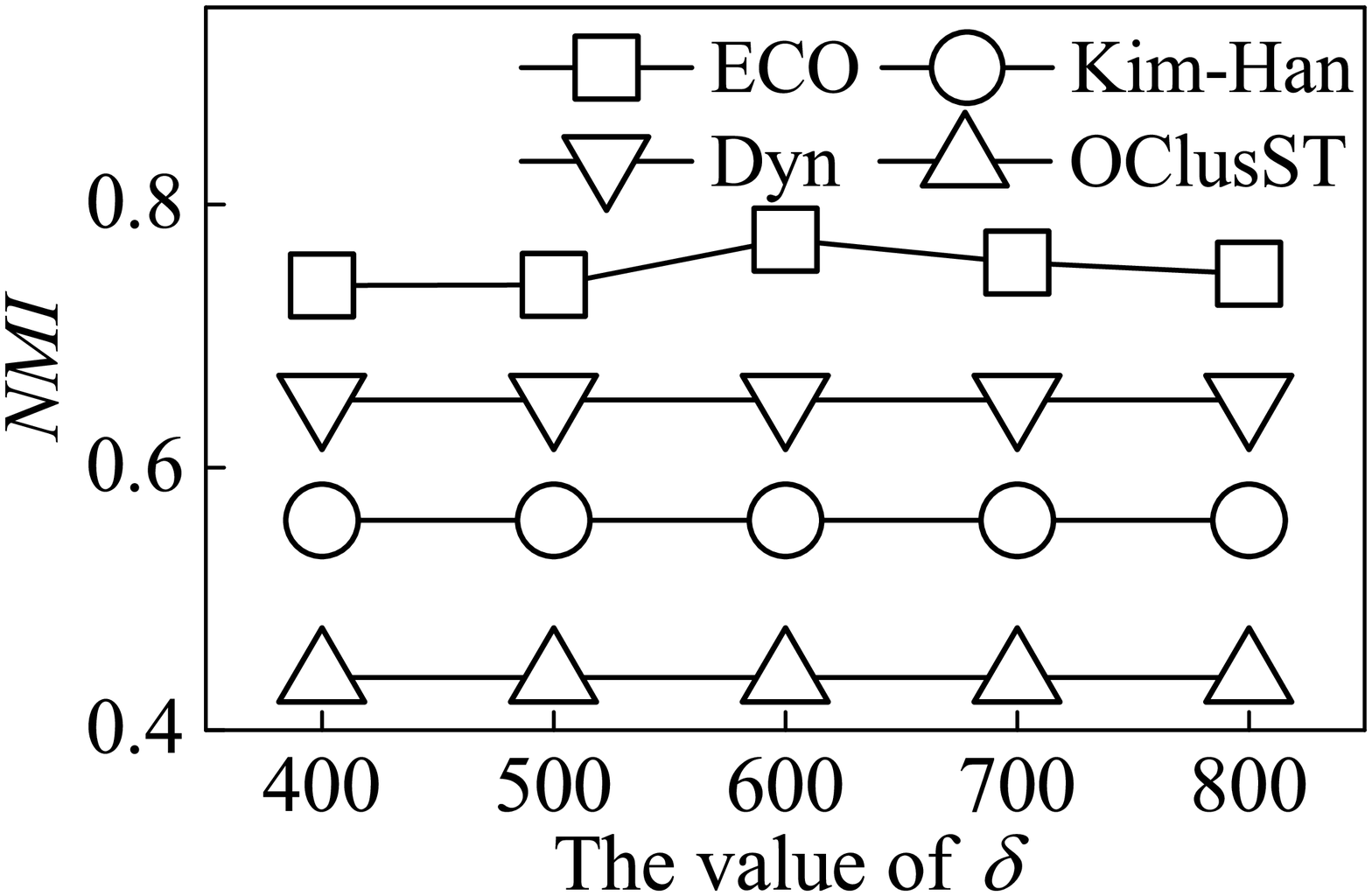}}
 \subfigure[{HZ dataset}]{
         \includegraphics[scale=.17]{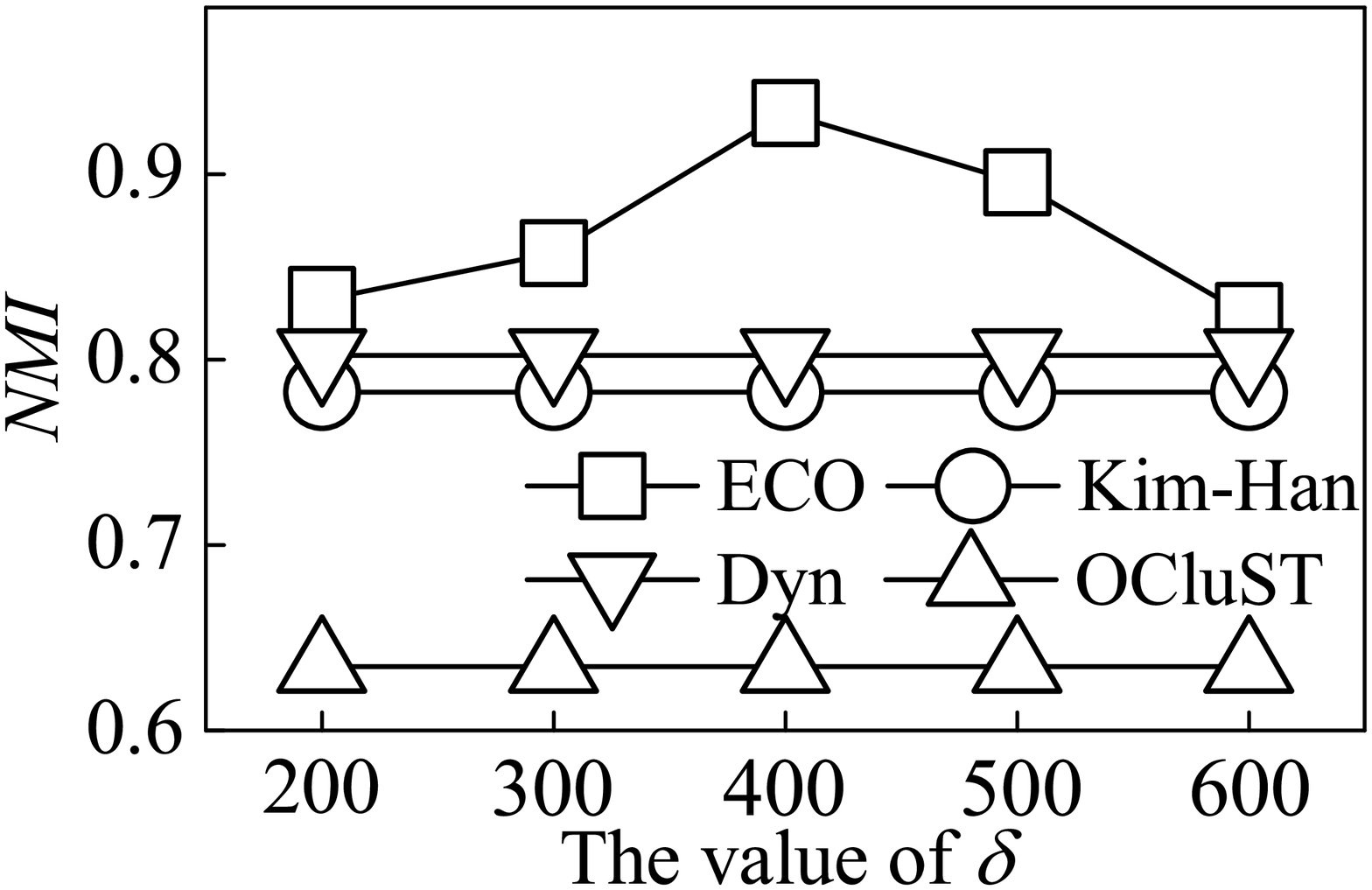}}
         \vspace{-5mm}
          \caption{Effects of varying $\mathbf{\delta}$}
	\label{fig:delta}
 \vspace{-6.5mm}
\end{figure}

Figure~\ref{fig:delta} reports in the effects of varying $\delta$. Specifically, ECO outperforms the baselines in terms of all metrics, and
the processing times of the methods remain stable, as shown in Figures~\ref{fig:delta}a and~\ref{fig:delta}b.  Figures~\ref{fig:delta}c--\ref{fig:delta}f indicate that as $\delta$ increases, both $\textit{QS}$ and $\textit{NMI}$ first increase and then drop. On the one hand, a too small $\delta$ leads to a small number of trajectories forming minimal groups and being smoothed; on the other hand, a too large $\delta$ also leads to few smoothing operations, as more pairs of trajectories $o$ and $o'$\,($o,o'\in \mathcal{M}_{k-1}(\tilde{s})$) satisfy $d(o.l,o'.l)\leq \delta$ at $\textit{dt}_k$.
 Since the baselines do not utilize parameter $\delta$, they are unaffected by variations in $\delta$.

\begin{figure} \centering
\subfigcapskip=-5pt
  \subfigure[{CD dataset}]{      \includegraphics[scale=.17]{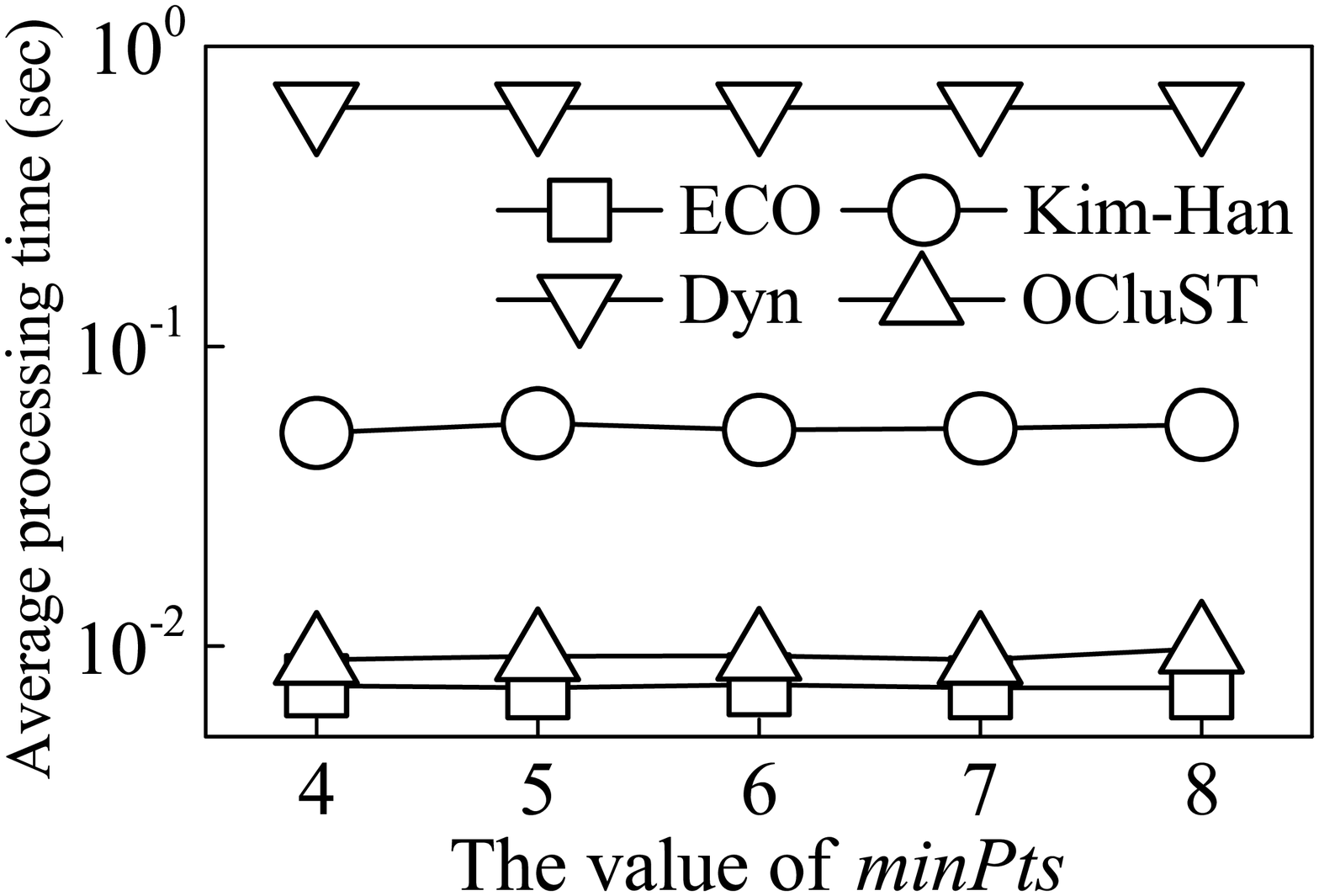}} 
 \subfigure[{HZ dataset}]{
\includegraphics[scale=.17]{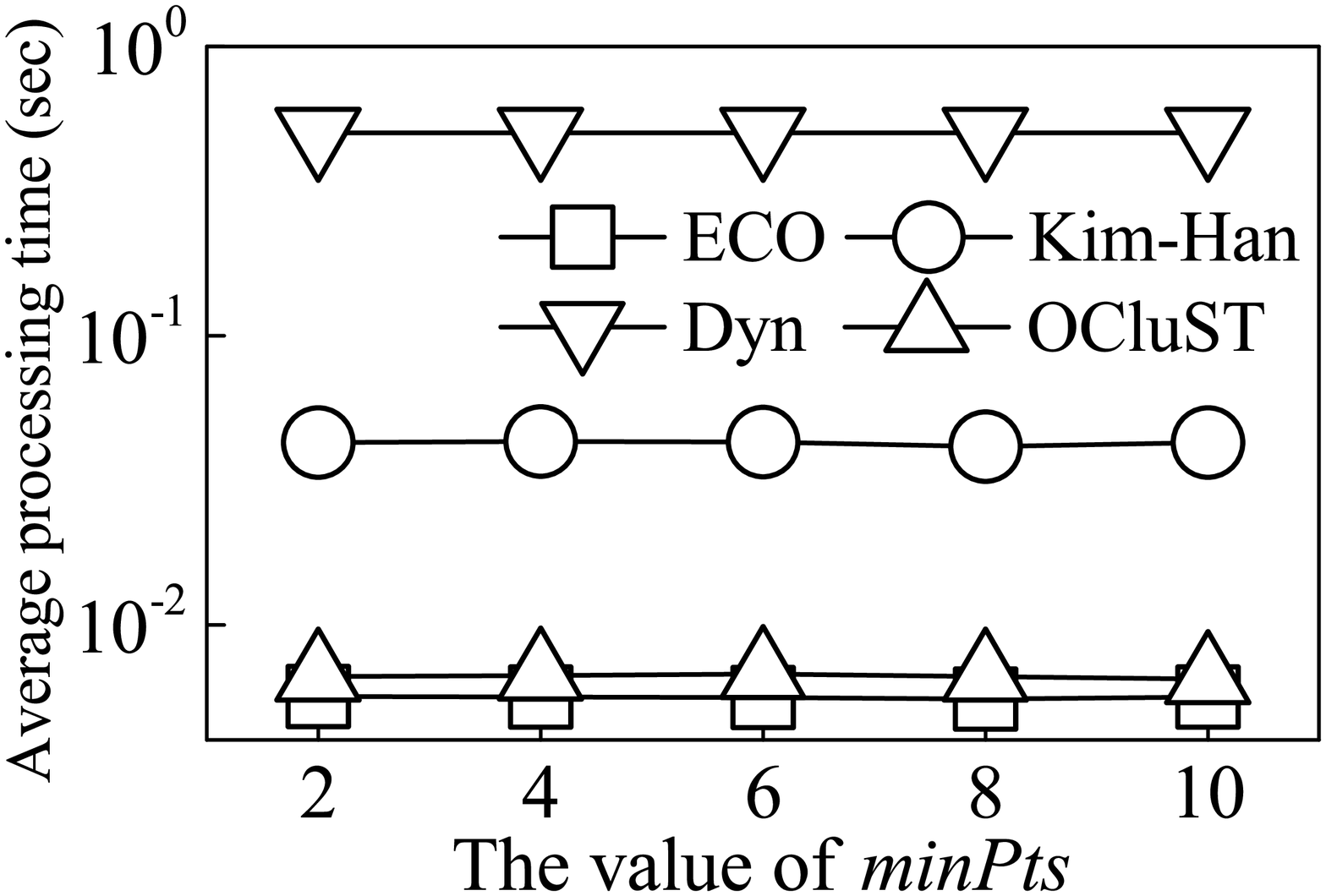}}
  \subfigure[{CD dataset}]{      \includegraphics[scale=.17]{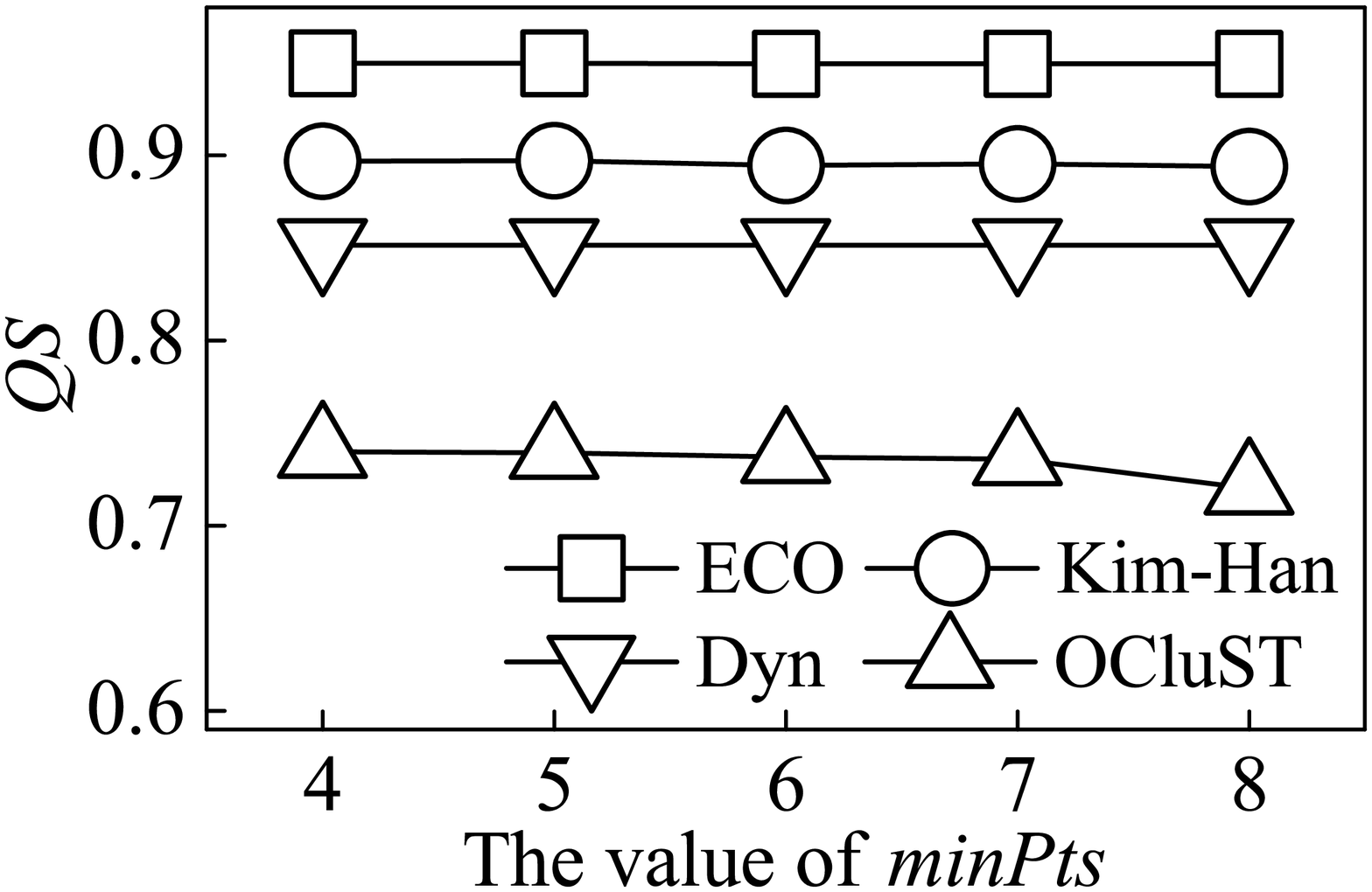}} 
 \subfigure[{HZ dataset}]{
\includegraphics[scale=.17]{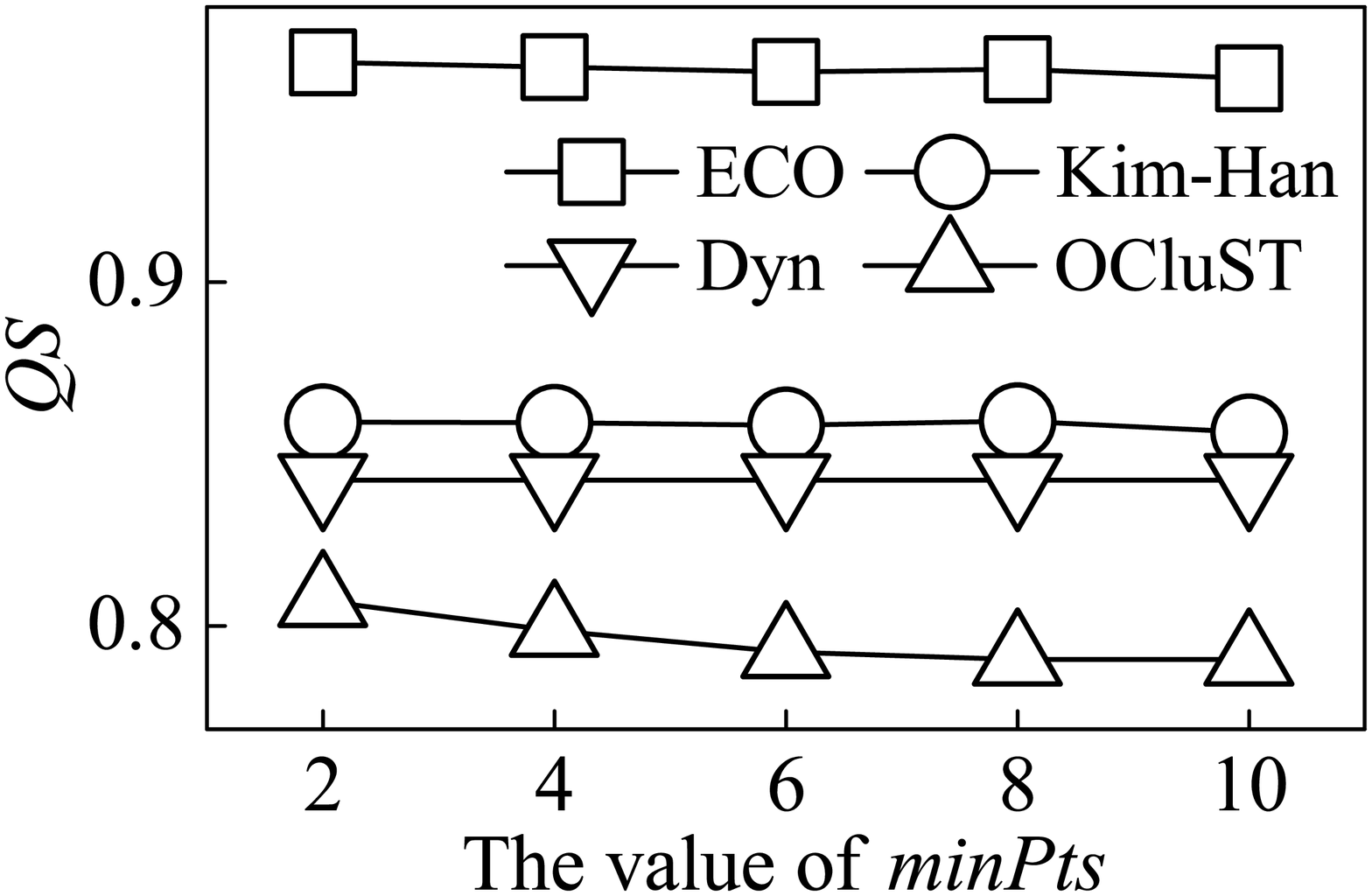}}
\subfigure[{CD dataset}]{      \includegraphics[scale=.17]{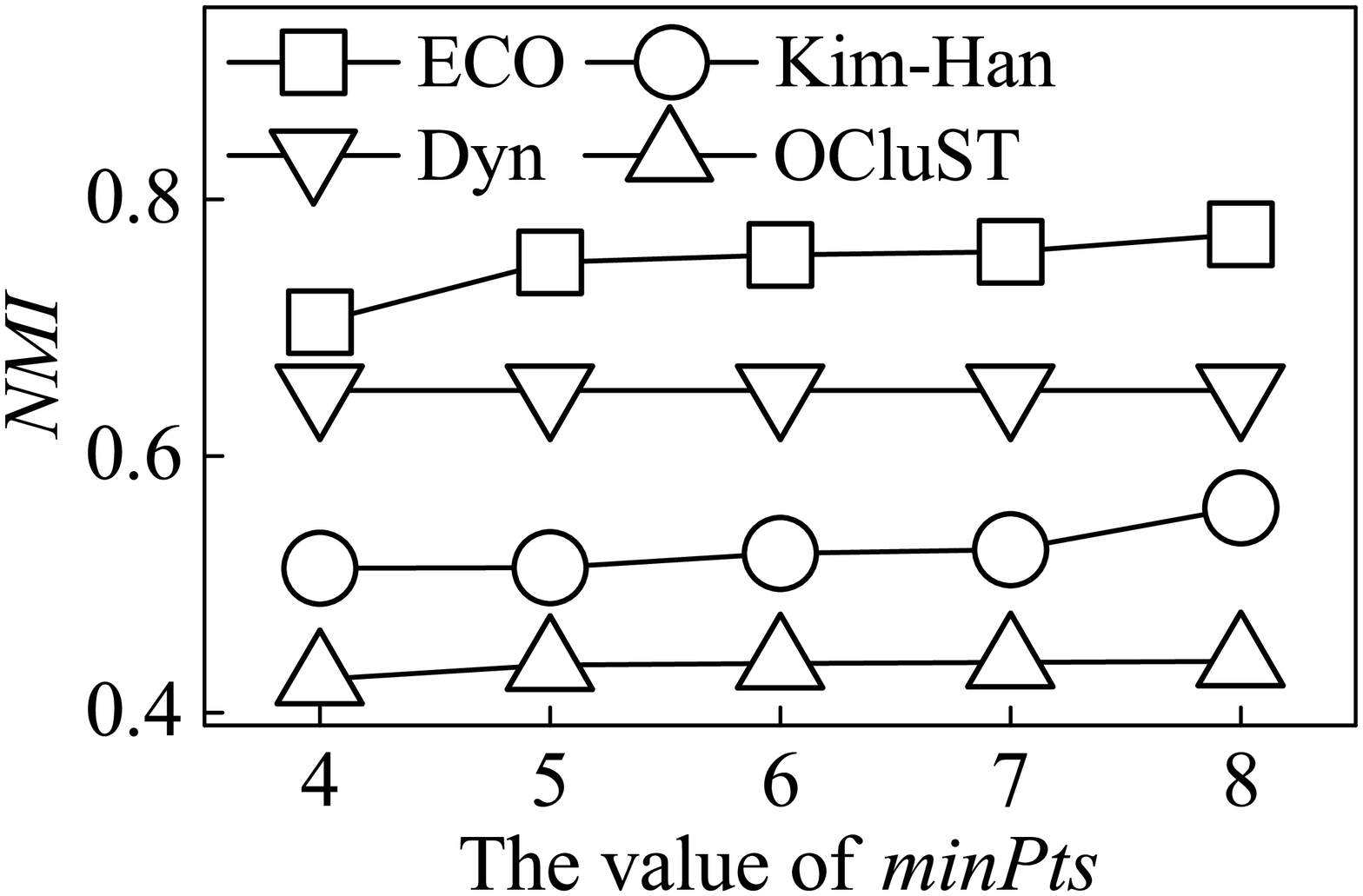}} 
 \subfigure[{HZ dataset}]{
\includegraphics[scale=.17]{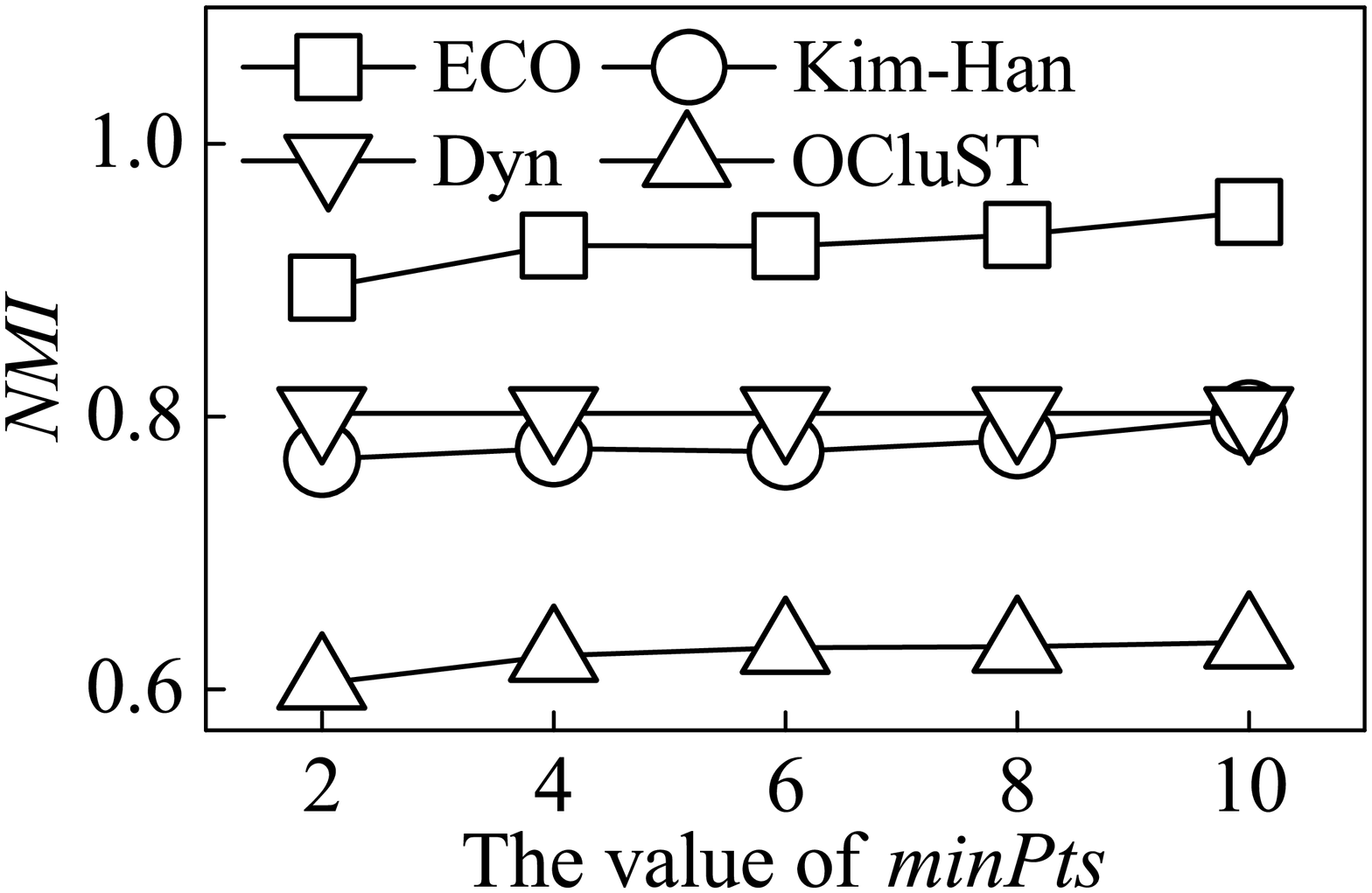}}
   \vspace{-5mm}
\caption{Effects of varying \textit{minPts}}	\label{fig:minPts}
   \vspace{-6mm}
\end{figure}

\paragraph{\textbf{Effects of varying \textit{minPts}}}
Figure~\ref{fig:minPts} shows the effect of $\textit{minPts}$ on clustering. When varying $\textit{minPts}$, the effects are similar to those seen when varying  $\rho$ and $\delta$. Figures~\ref{fig:minPts}c and \ref{fig:minPts}d show that $\textit{QS}$ of ECO, Dyn, and OCluST drop as \textit{minPts} increases. The findings indicate that the average distance between trajectories in different clusters decreases with \textit{minPts}. Assume that $o$ is a core point when \textit{minPts} is small and that $o'$ is a border point that is density reachable from $o$.
As \textit{minPts} increases, $o$ may no longer be a core point. In this case, $o'$ and $o$ may be density reachable from different core points and may thus be in different clusters, even if $d(o.l,o'.l)\leq \varepsilon$. As a result, the distances between trajectories in different clusters decrease. 
Figures~\ref{fig:minPts}e and~\ref{fig:minPts}f show that \textit{NMI} increases with \textit{minPts} for ECO, Dyn, and OCluST. The findings suggest that with a smaller \textit{minPts}, the trajectories at the "border" of a cluster are more likely to shift between being core points and being non-core points over time. In this case, clusters fluctuate more between consecutive time steps for a smaller \textit{minPts}. 
As Dyn adopts particle swarm clustering instead of density based clustering,  Dyn is unaffected by \textit{minPts}.

\begin{figure} \centering
\subfigcapskip=-5pt
  \subfigure[{CD dataset}]{      \includegraphics[scale=.17]{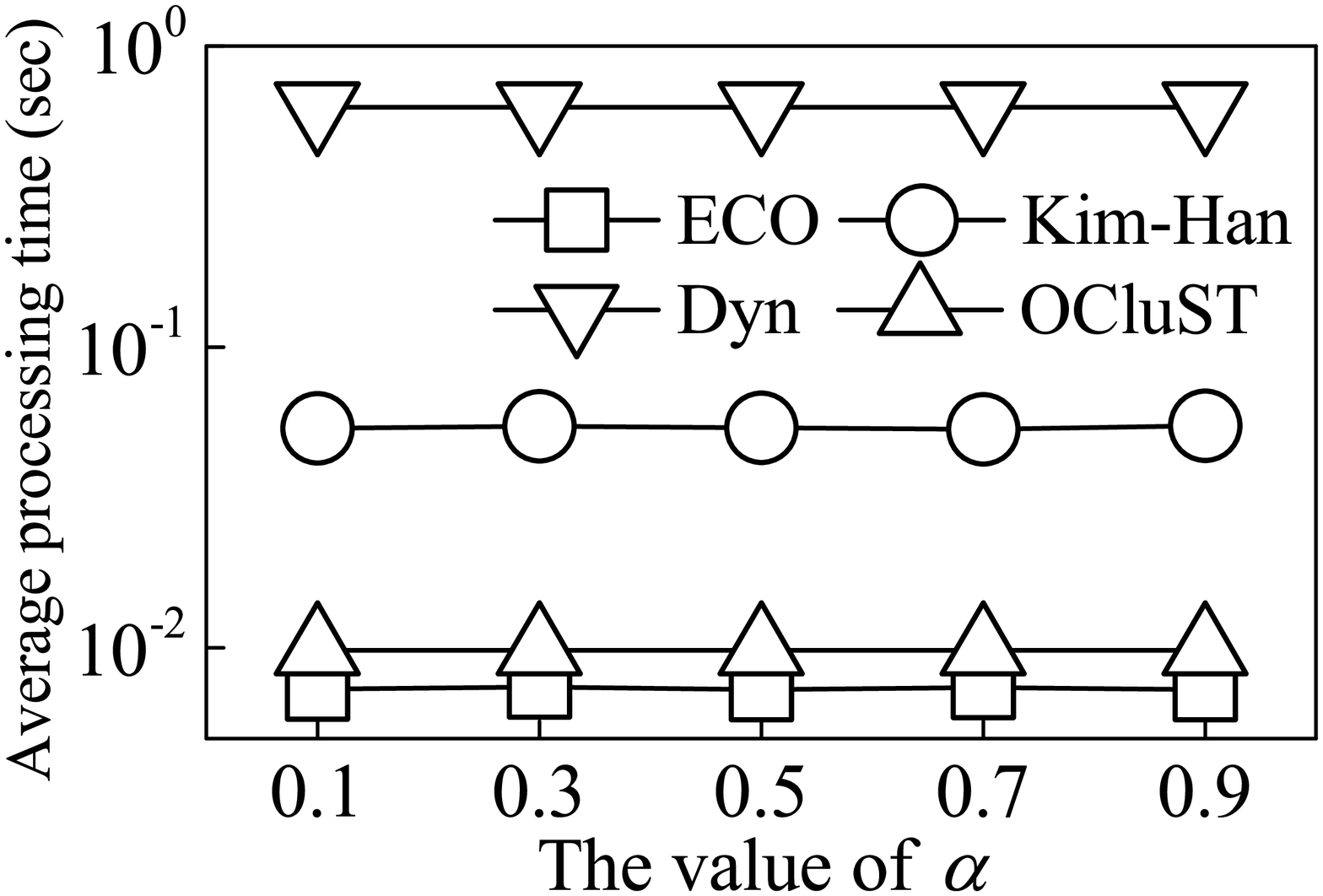}} 
 \subfigure[{HZ dataset}]{
\includegraphics[scale=.17]{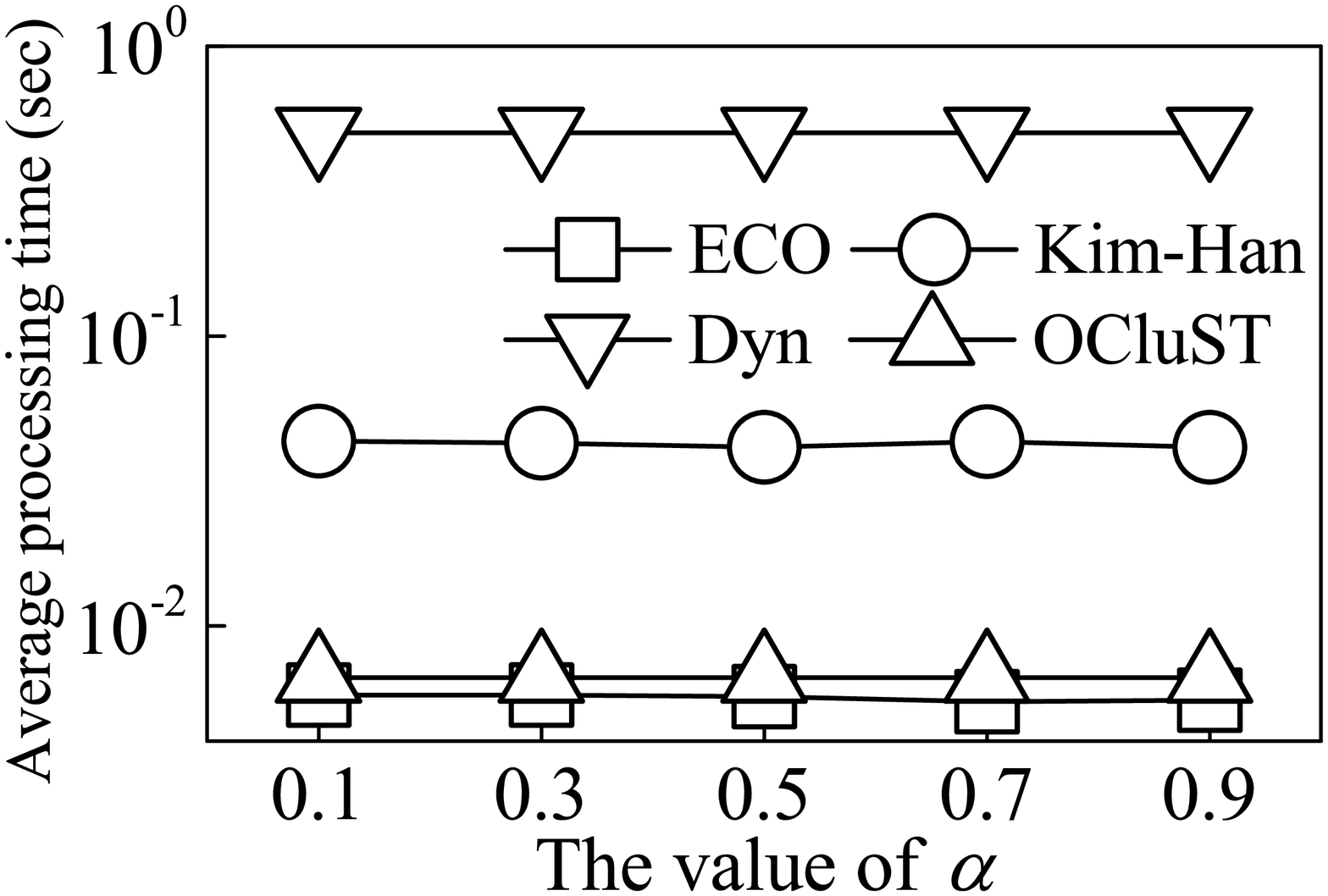}}
\subfigure[{CD dataset}]{      \includegraphics[scale=.17]{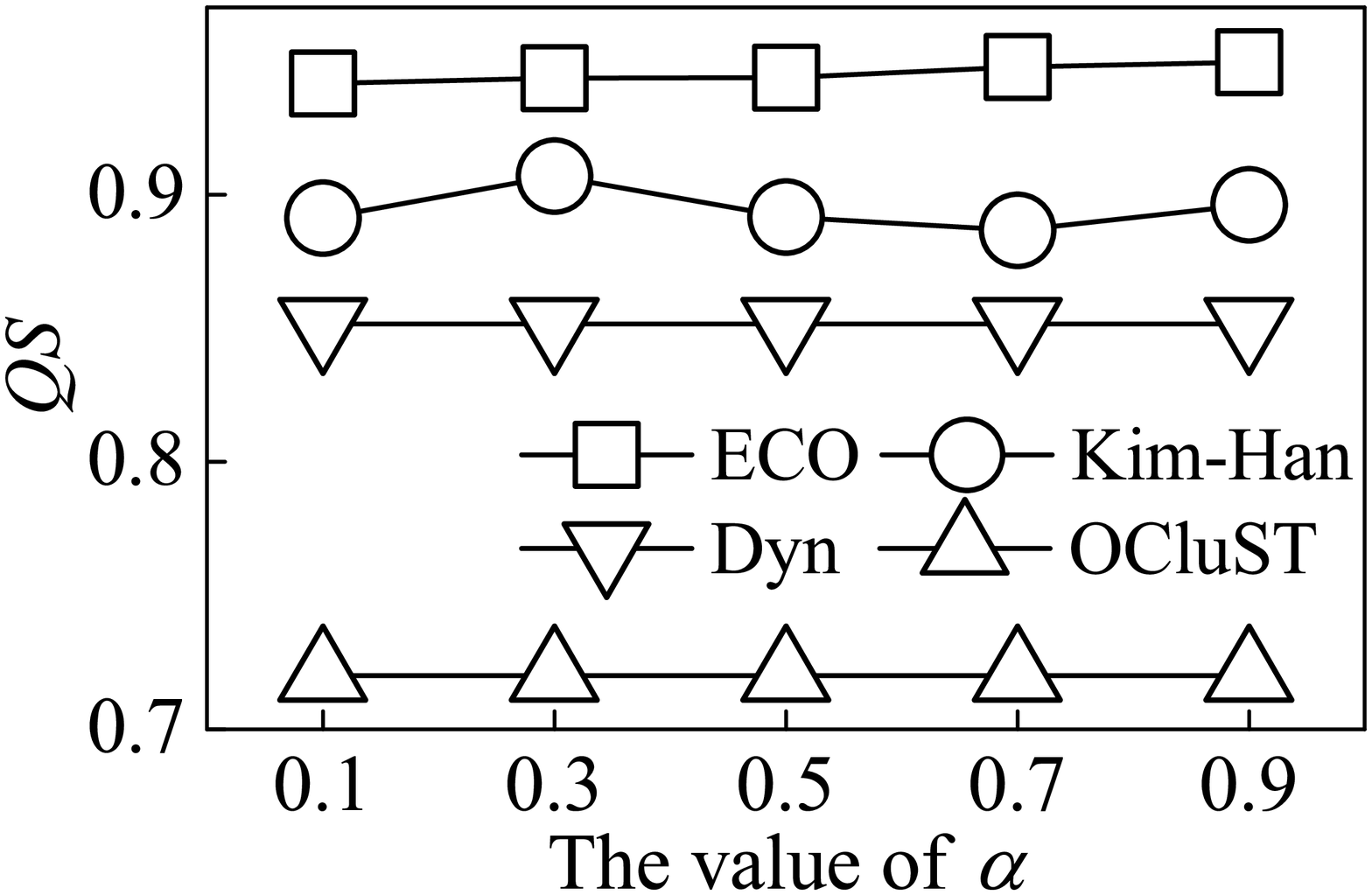}} 
 \subfigure[{HZ dataset}]{
\includegraphics[scale=.17]{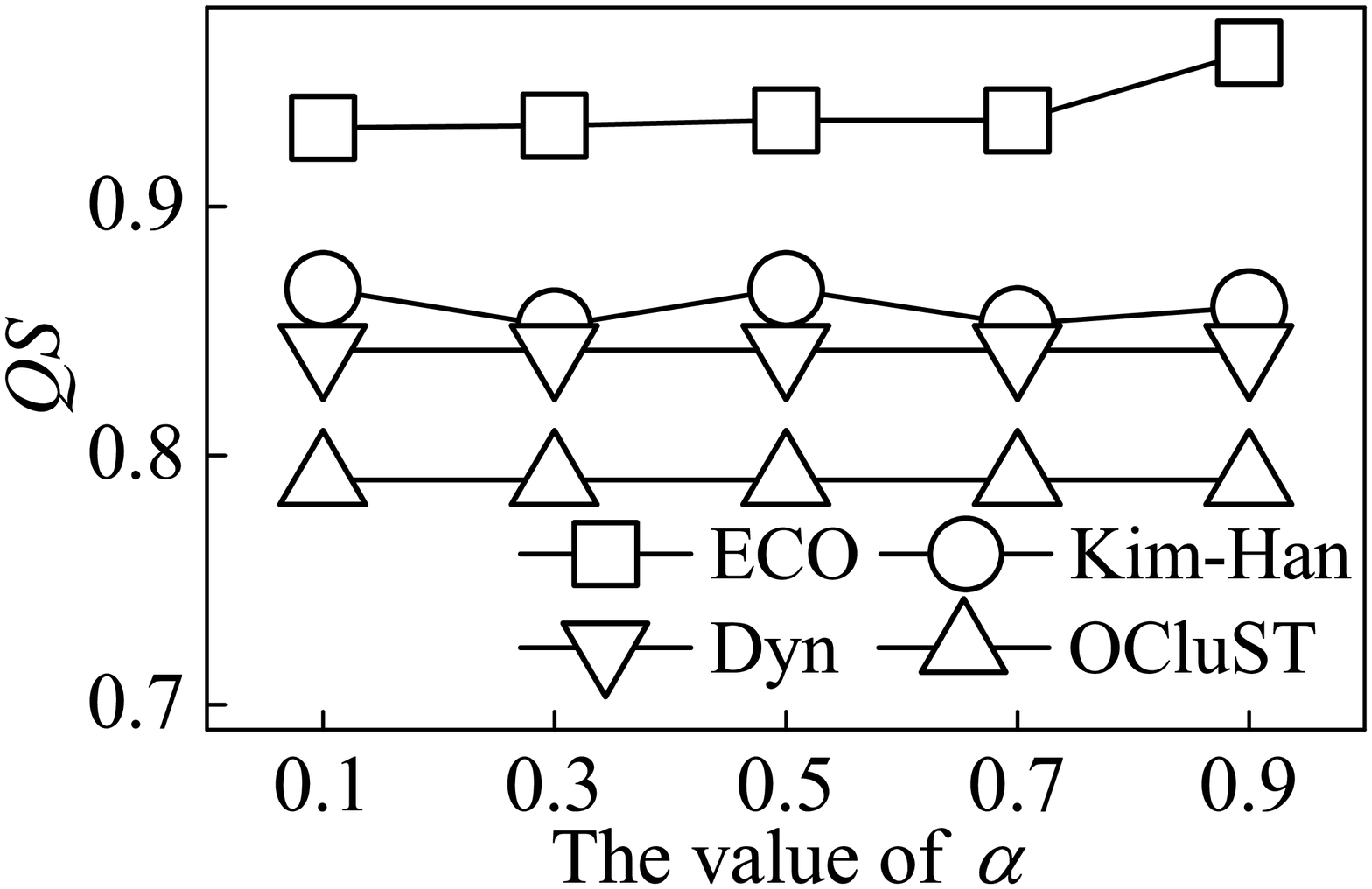}}
\subfigure[{CD dataset}]{      \includegraphics[scale=.17]{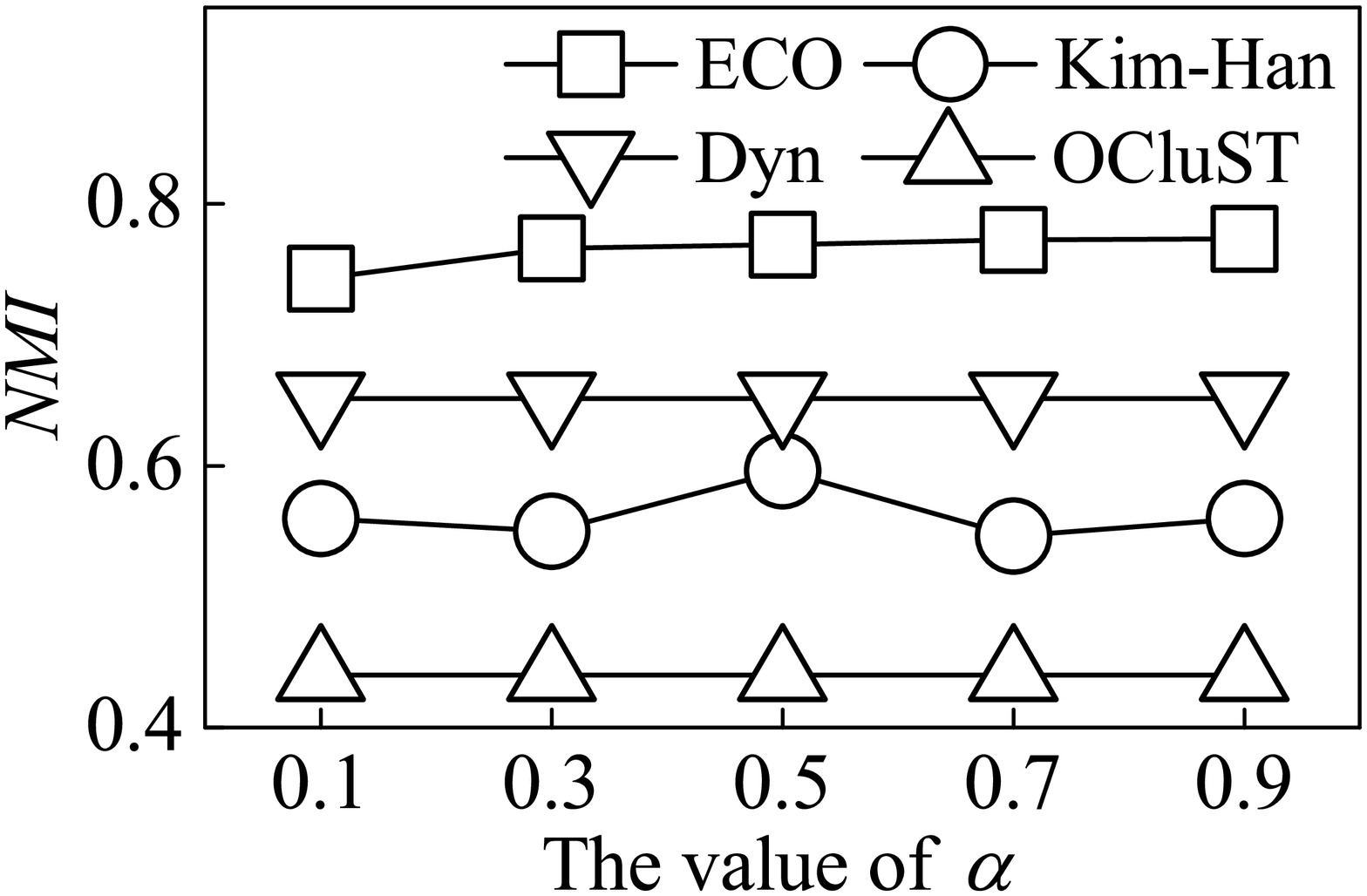}} 
 \subfigure[{HZ dataset}]{
\includegraphics[scale=.17]{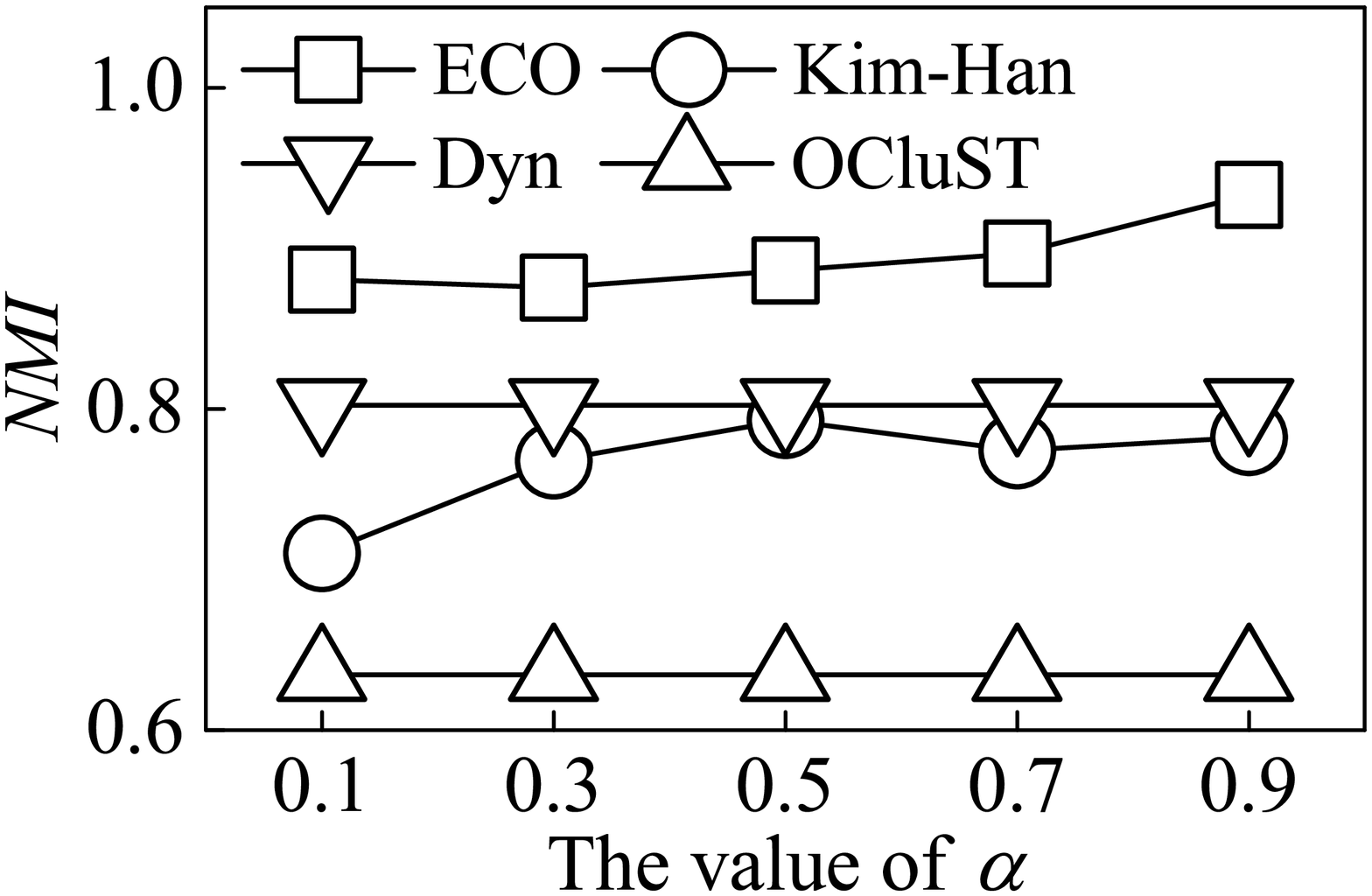}}
   \vspace{-5mm}
\caption{Effects of varying $\alpha$}	\label{fig:alpha}
   \vspace{-7mm}
\end{figure}

\paragraph{\textbf{Effects of varying $\bm{\alpha}$}}
Figure \ref{fig:alpha} shows the effects of varying $\alpha$. First, ECO always achieves the best performance among all methods. Second, all methods exhibit stable performance when varying $\alpha$. Third, both \textit{QS} and \textit{NMI} of ECO increase with $\alpha$. On the one hand, according to Formula~\ref{f:F_fixedcp}, a larger $\alpha$ generally leads to a smaller distance between two trajectories in the same cluster. On the other hand, Formula~\ref{f:F_fixedcp} reduces the historical cost as $\alpha$ increases, which renders the evolution of clusters more smooth.

\subsection{Scalability}
To study the scalability, we vary the data size from 20\% to 100\%, which is done by randomly sampling moving objects by their IDs. The results are reported in Figure~\ref{fig:scalability}. First,  ECO achieves the highest efficiency for large datasets, but is less efficient than OCluST for small datasets. This is because the locations of trajectories  are generally distributed uniformly when data size is small. In this case, the grid index becomes less useful. As expected, the processing times of all methods increase with the dataset size. 

Second, \textit{QS} improves with the data size for all methods, with ECO always being best. As illustrated above, data distribution is generally non-uniform  for a 100\% dataset and becomes increasingly uniform as the data size decreases. Thus, the average distances between any pair of trajectories in the same cluster become smaller as data size increases, resulting in a larger \textit{QS}.

Third, ECO achieves the highest \textit{NMI}, which increases with the data size. This is mainly because fewer trajectories are able to form minimal groups and be subjected to smoothing when the average distances between any pair of trajectories increases.

\begin{figure} \centering
\subfigcapskip=-5pt
  \subfigure[{CD dataset}]{      \includegraphics[scale=.17]{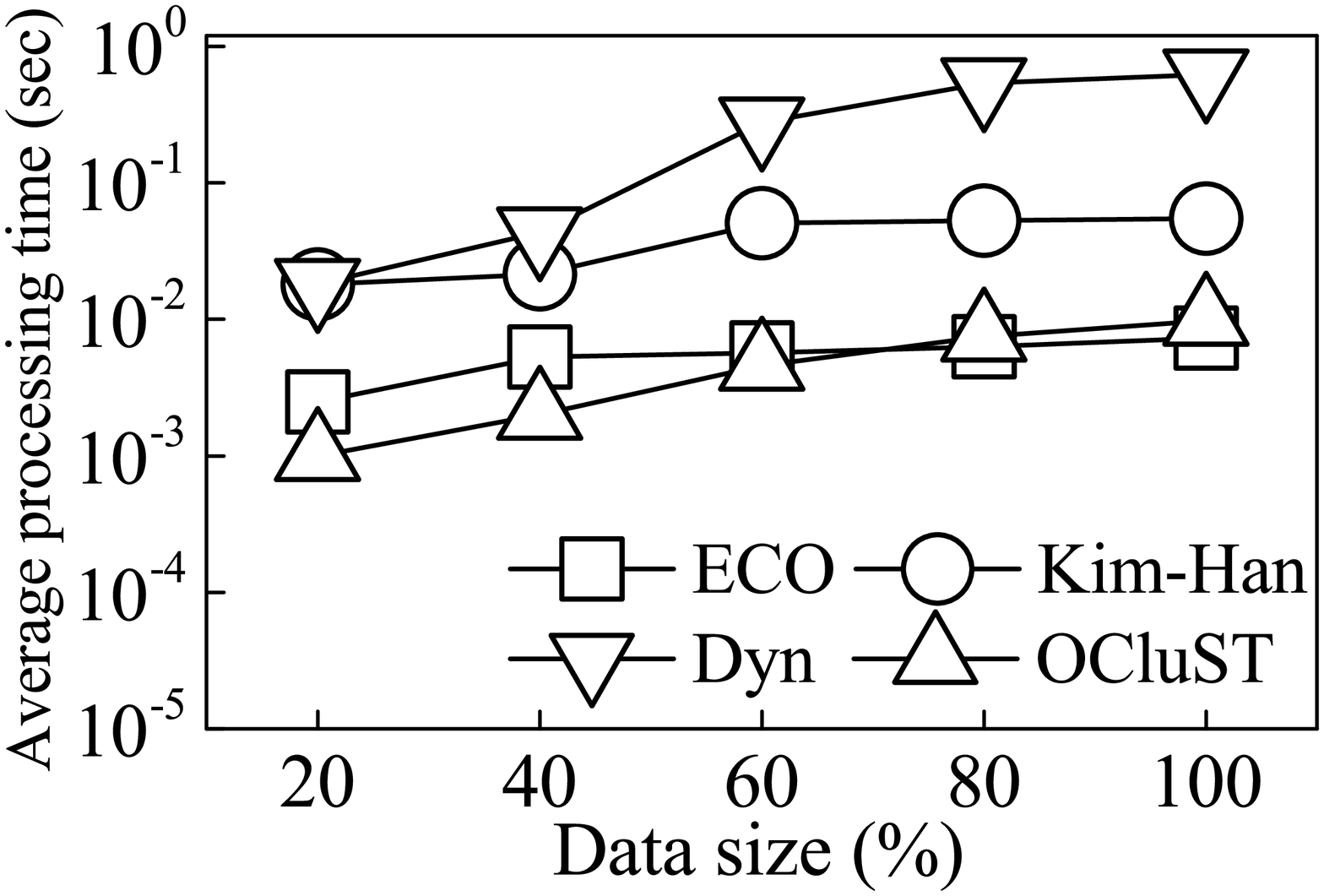}} 
 \subfigure[{HZ dataset}]{
\includegraphics[scale=.17]{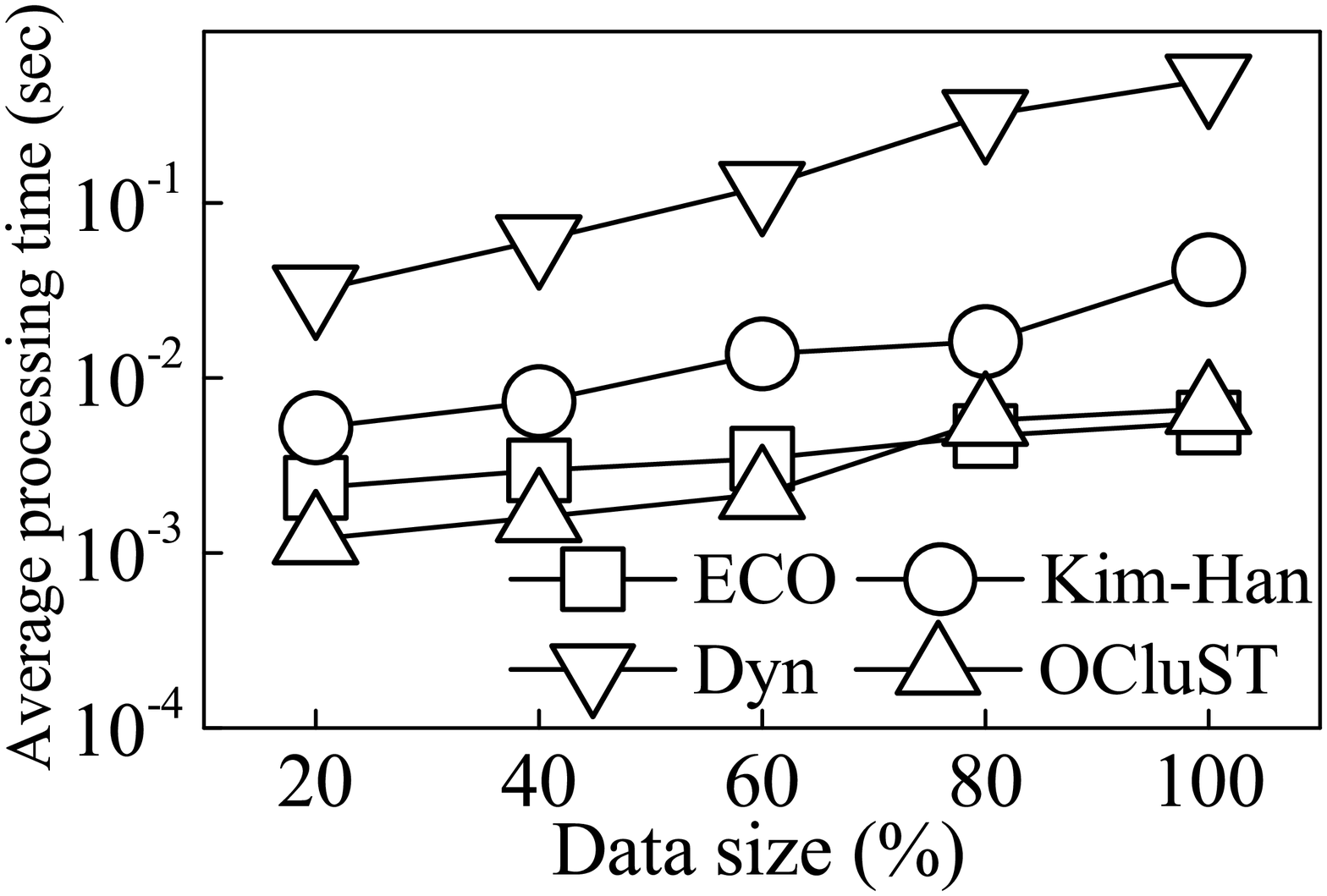}}
  \subfigure[{CD dataset}]{      \includegraphics[scale=.17]{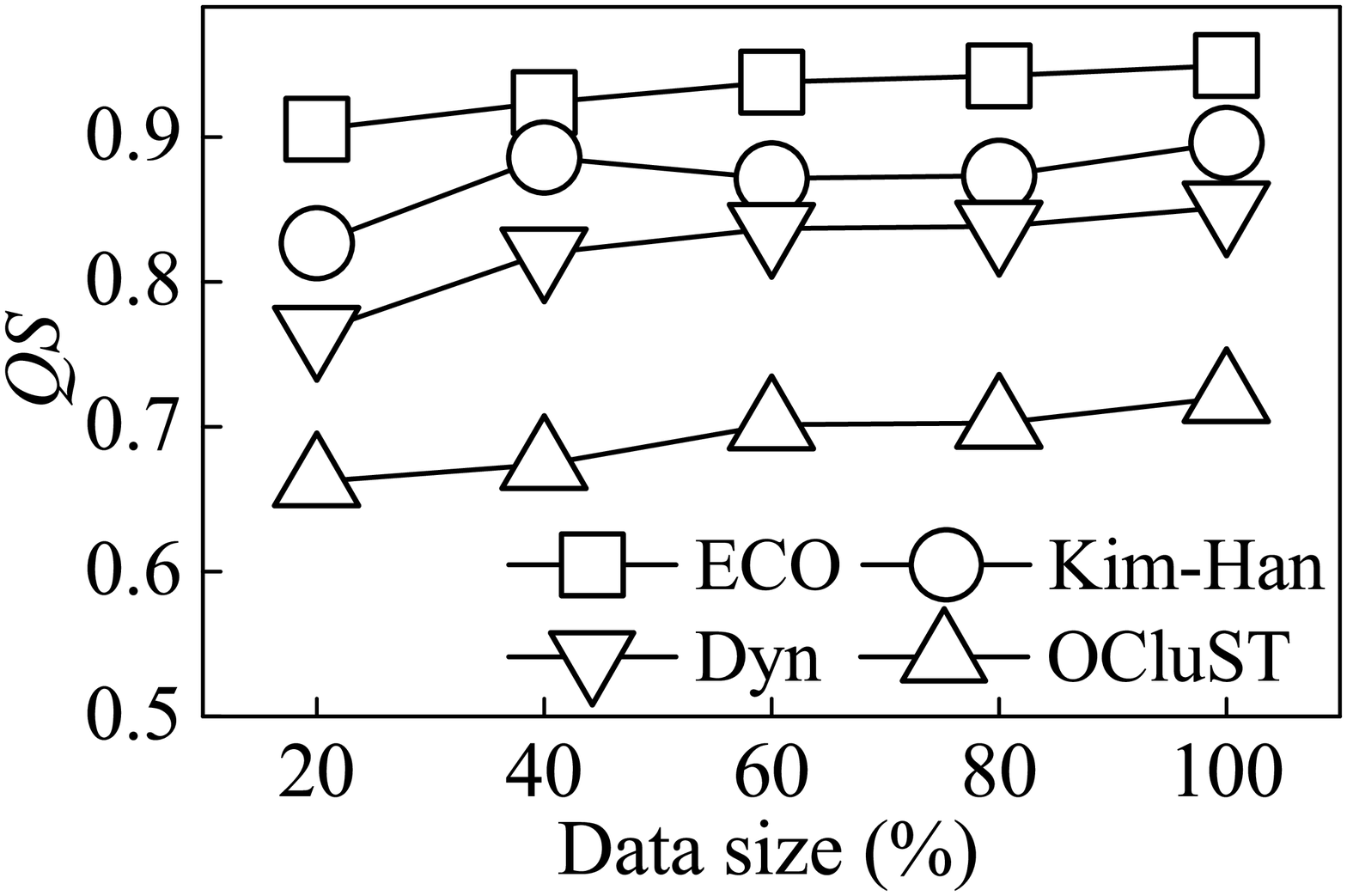}} 
 \subfigure[{HZ dataset}]{
\includegraphics[scale=.17]{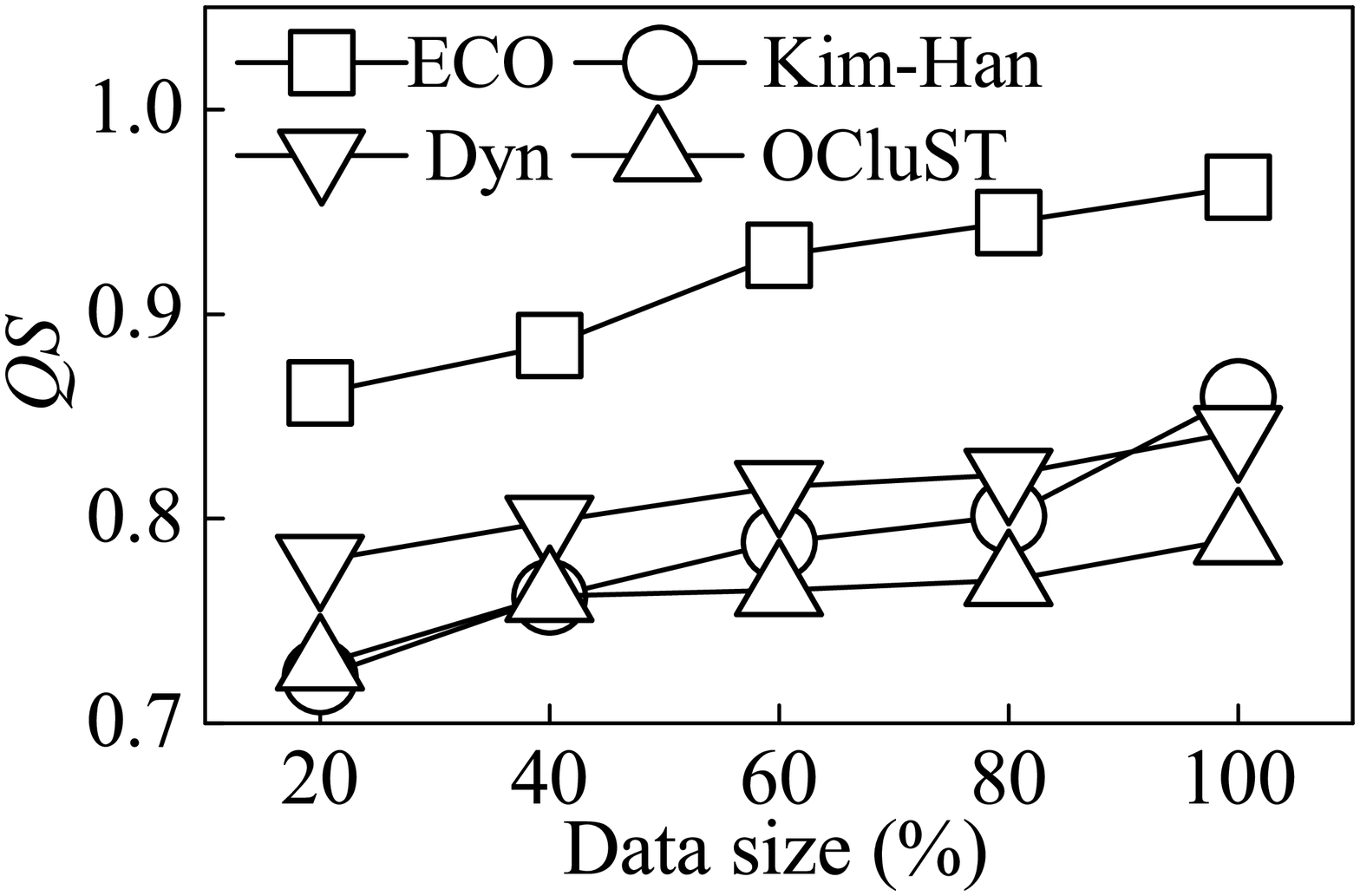}}
\subfigure[{CD dataset}]{      \includegraphics[scale=.17]{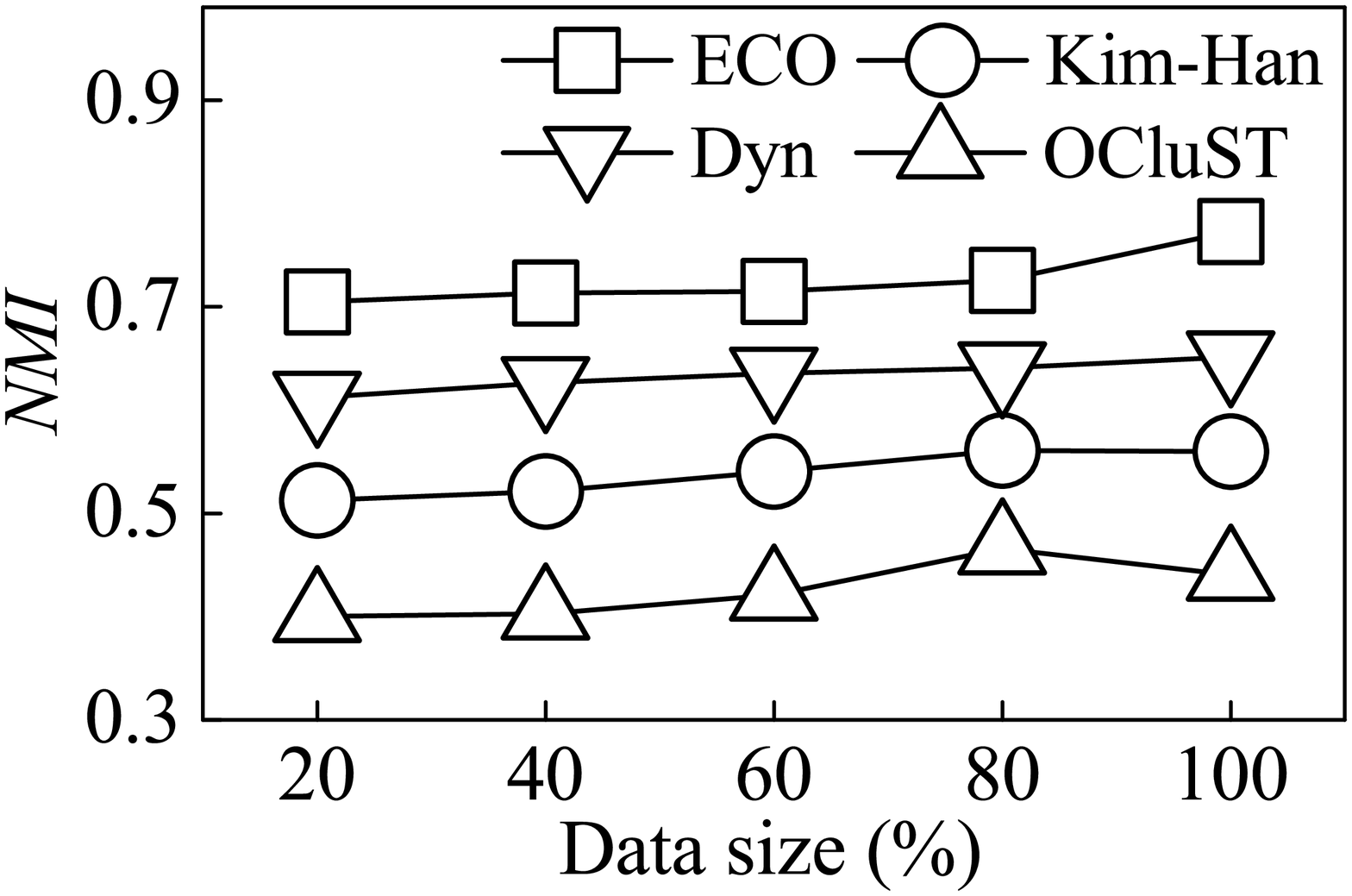}} 
 \subfigure[{HZ dataset}]{
\includegraphics[scale=.17]{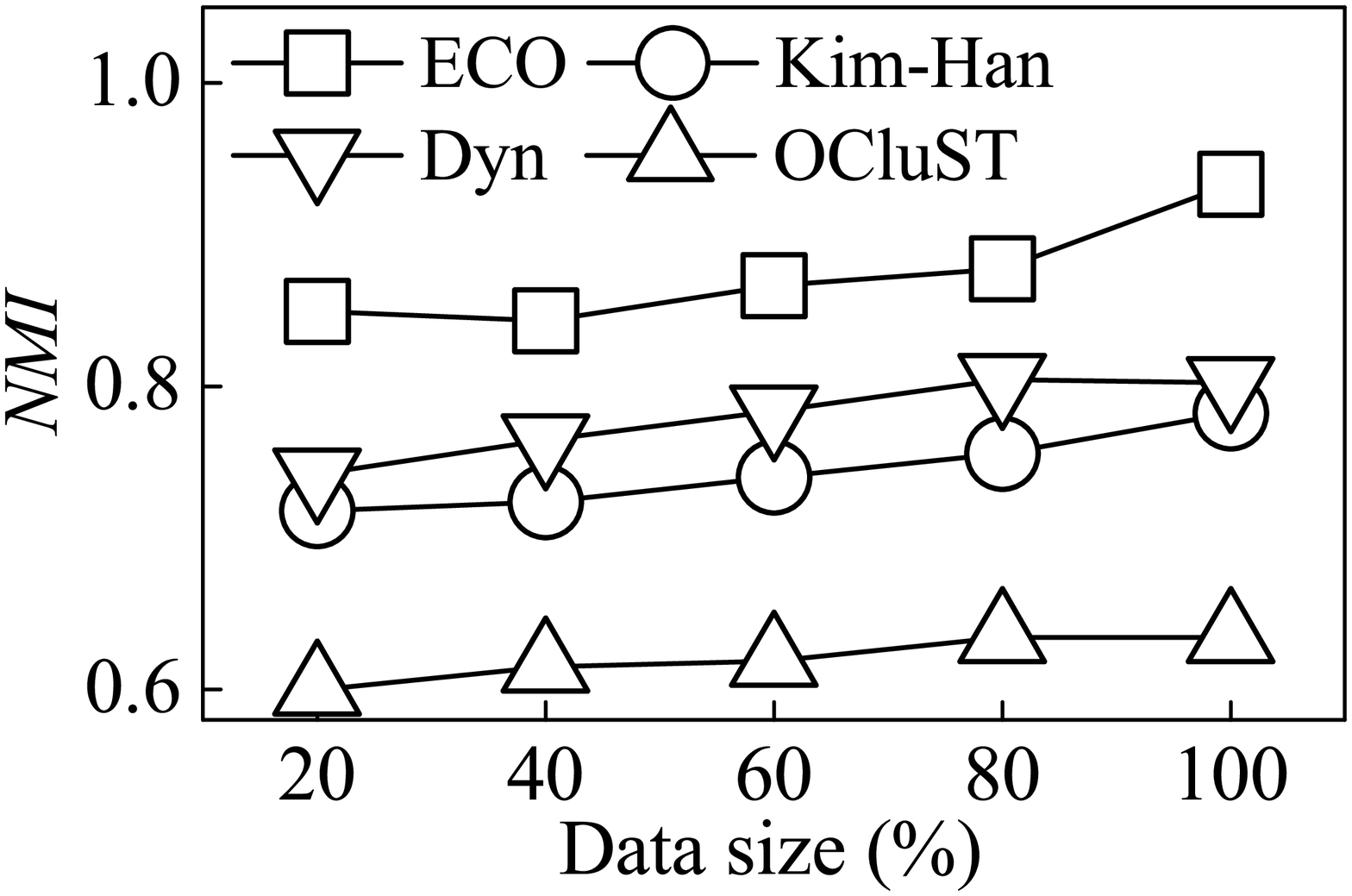}}
   \vspace{-5mm}
\caption{Scalability}\label{fig:scalability}
   \vspace{-6mm}
\end{figure}

%% file: related_work.tex
\section{Related Work}\label{sec:related_work}
We proceed to review related works on streaming trajectory clustering
and evolutionary clustering.
\subsection{Streaming Trajectory Clustering}
Streaming trajectory clustering finds  representative paths or
common movement trends among objects in real time.
The main target of existing studies of streaming trajectory clustering is to update clusters continuously with high efficiency.
Jensen et al.~\cite{jensen2007continuous} exploit an incrementally maintained clustering feature CF and propose a scheme for continuous clustering of moving objects.
Costa et al.~\cite{costa2014dealing} define a new metric for streaming trajectory clustering and process trajectories by means of non-separable Fourier transforms. Tra-POPTICS~\cite{deng2015scalable}  explores the use of graphics processing units to improve efficiency.
Similar to ECO, proposals in ~\cite{yu2013online,yu2013continuous} use an index structure, TC-tree, to facilitate efficient updates of clusters. Some studies maintain a micro-group structure, that stores compact summaries of trajectories to enable fast and flexible clustering~\cite{li2010incremental, da2016online, mao2018online,riyadh2017cc_trs}. CUTis~\cite{da2016online} continuously merges micro-groups into clusters, while CC\_TRS~\cite{riyadh2017cc_trs}, TCMM~\cite{li2010incremental}, and OCluST~\cite{mao2018online} generate macro groups on top of micro groups when requested by users. In the experimental study, we compare with the state-of-the-art streaming trajectory clustering method~\cite{mao2018online}.
Some studies of real-time co-movement pattern mining also involve
streaming trajectory clustering~\cite{chen2019real, tang2012discovery, li2012effective}. 
Comprehensive surveys of trajectory clustering are available~\cite{bian2018survey, yuan2017review}.

To the best of our knowledge, no existing studies of streaming trajectory clustering  exploit temporal smoothness to improve clustering quality. 

\subsection{Evolutionary Clustering}
Evolutionary clustering has been studied to discover evolving community structures in applications such as social~\cite{kim2009particle} and financial networks~\cite{fenn2009dynamic} and recommender systems~\cite{chen2020collaborative}. Most  studies target $k$-means, agglomerative hierarchical, and spectural clustering~\cite{chakrabarti2006evolutionary,xu2014adaptive,chi2007evolutionary,ma2017evolutionary,ma2019detecting}.
Chakrabarti et al.  ~\cite{chakrabarti2006evolutionary} propose a generic framework for evolutionary clustering.
Chi et al.~\cite{chi2007evolutionary} develop two functions for evaluating historical costs, PCQ (Preserving Cluster Quality) and PCM (Preserving Cluster Membership), to improve the stability of clustering. Xu et al.~\cite{xu2014adaptive} estimate the optimal smoothing parameter $\alpha$ of evolutionary clustering. Recent studies model evolutionary clustering as a multi-objective problem~\cite{folino2013evolutionary,yin2021multi, liu2020detecting,liu2019evolutionary} and use genetic algorithms to solve it, which is too expensive for online scenarios. %Folino et al.~\cite{folino2013evolutionary} adopt pareto optimality theory to
%derive the optimal solution of the multiple-objective function.
Dyn~\cite{yin2021multi}, the state-of-the-art proposal, features non-redundant random walk based population initialization and an improved particle swarm algorithm  to enhance clustering quality.  

The Kim-Han proposal~\cite{kim2009particle} is the  one that is closest to ECO. It uses neighbor-based smoothing and a cost embedding technique that smooths the similarity between each pair of nodes. However, ECO differs significantly from Kim-Han.
First, the cost functions used are fundamentally different. Kim-Han's cost function is designed  specifically for nodes in dynamic networks and is neither readily applicable to, or suitable for, trajectory data. In contrast, ECO's cost function is shaped according to the characteristics of the movements of trajectories.
Second, Kim-Han smooths the similarity between each pair of neighboring nodes; in contrast, ECO smooths only the locations of a trajectory with abrupt movements, and the smoothing is performed only according to its most smoothly moving neighbor. 
Finally, Kim-Han includes iterative processes that degrade its efficiency, while ECO achieves $O(\lvert \mathcal{O}_k\rvert^2)$ complexity at each time step in the worst case and adopts a grid index to improve efficiency.

%% file: conclusion.tex
%\vspace{-2mm}
\section{Conclusion and Future Work}\label{sec:conclusion}
We propose a new framework for evolutionary clustering of streaming trajectories that targets faster and better clustering.
Following existing studies, we propose so-called snapshot and historical costs for trajectories, and formalize the problem of evolutionary clustering of streaming trajectories, called ECO. Then, we formulate ECO as an optimization problem and prove that it can be solved approximately in linear time, which eliminates the iterative processes employed in previous proposals and improves efficiency significantly.
Further, we propose a minimal group structure and a seed point shifting strategy that facilitate temporal smoothing. We also present the algorithms necessary to enable evolutionary clustering along with a 
set of optimization techniques that aim to enhance performance.
Extensive experiments with two real-life datasets show that ECO outperforms existing state-of-the-art proposals in terms of clustering quality and running time  efficiency.

In future research, it is of interest to deploy ECO on a distributed platform and to exploit more information for smoothing such as road conditions and driver preferences.